%% file: aaa-TF-compression-main-arxiv.tex
\documentclass[journal]{vgtc}             

\onlineid{1929}

\vgtccategory{Research}

\vgtcpapertype{algorithm/technique}

\graphicspath{{./figs/}{./}} 

\usepackage{tabu}                      
\usepackage{booktabs}                  

\usepackage{amsfonts}
\usepackage{amsmath}
\usepackage{amsthm}
\usepackage{overpic}
\usepackage{relsize}
\usepackage{enumitem}
\usepackage{thmtools}
\usepackage{thm-restate}
\usepackage[export]{adjustbox}
\usepackage{multicol}
\usepackage{caption}

\newcommand{\Rspace}{\mathbb{R}}
\newcommand{\Tspace}{\mathbb{T}}
\newcommand{\Xspace}{\mathbb{X}}

\newcommand{\R}{\mathbb{R}}
\newcommand{\N}{\mathbb{N}}
\newcommand{\X}{\mathbb{X}}
\newcommand{\T}{\mathbb{T}}
\newcommand{\tr}{\text{tr}}
\newcommand{\sign}{\text{sign}}
\newcommand{\para}{\noindent\textbf}

\newcommand{\TF}{f}

\newcommand{\toolname}{\text{TFZ}}
\newcommand{\newcolor}{black}
\newcommand{\new}{\textcolor{\newcolor}}

\newcommand{\SFIX}{\mathsf{S}}
\newcommand{\RSIGN}{\mathsf{R}}
\newcommand{\DSIGN}{\mathsf{D}}
\newcommand{\ROVERS}{\mathsf{RS}}
\newcommand{\DLARGEST}{\mathsf{D_{m}}}

\newcommand{\StressA}{{Stress A}}
\newcommand{\StressB}{{Stress B}}
\newcommand{\BrainA}{{Brain A}}
\newcommand{\BrainB}{{Brain B}}
\newcommand{\Ocean}{{Ocean}}
\newcommand{\Miranda}{{Miranda}}
\newcommand{\VS}{{Vortex Street}}
\newcommand{\HC}{{Heated Cylinder}}

\newcommand{\DP}{{D_+}}
\newcommand{\DN}{{D_-}}
\newcommand{\RP}{{R_+}}
\newcommand{\RN}{{R_-}}
\newcommand{\SA}{{S}}
\newcommand{\RRP}{{r_{r+}}}
\newcommand{\RRN}{{r_{r-}}}
\newcommand{\SRP}{{s_{r+}}}
\newcommand{\SRN}{{s_{r-}}}
\newcommand{\DDP}{{d_{d+}}}
\newcommand{\DDN}{{d_{d-}}}
\newcommand{\SE}{{\mathcal{S}}}

\declaretheorem[name=Lemma]{lemma}
\declaretheorem[name=Corollary]{corollary}
\usepackage{enumitem}

\title{TFZ: Topology-Preserving Compression of 2D Symmetric and Asymmetric Second-Order Tensor Fields}

\author{
  \authororcid{Nathaniel Gorski}{0009-0001-8205-5640}, 
  \authororcid{Xin Liang}{0000-0002-0630-1600}, 
  \authororcid{Hanqi Guo}{0000-0001-7776-1834},
  and \authororcid{Bei Wang}{0000-0002-9240-0700}
}

\authorfooter{
\item Nathaniel Gorski is with the University of Utah.
  	E-mail: gorski@sci.utah.edu.
\item Xin Liang is with the University of Kentucky. 
 	E-mail: xliang@uky.edu.
\item Hanqi Guo is with the Ohio State University. 
	E-mail: guo.2154@osu.edu. 
\item Bei Wang is with the University of Utah. 
	E-mail: beiwang@sci.utah.edu. 
}

\abstract{
\input{sec-abstract}

}

\keywords{Lossy compression, tensor fields, topology preservation, topological data analysis, topology in visualization}

\teaser{
\centering
\includegraphics[width=1.0\linewidth, alt={A visualization of stress tensor field with classic and Augmented SZ3 compressors.}]{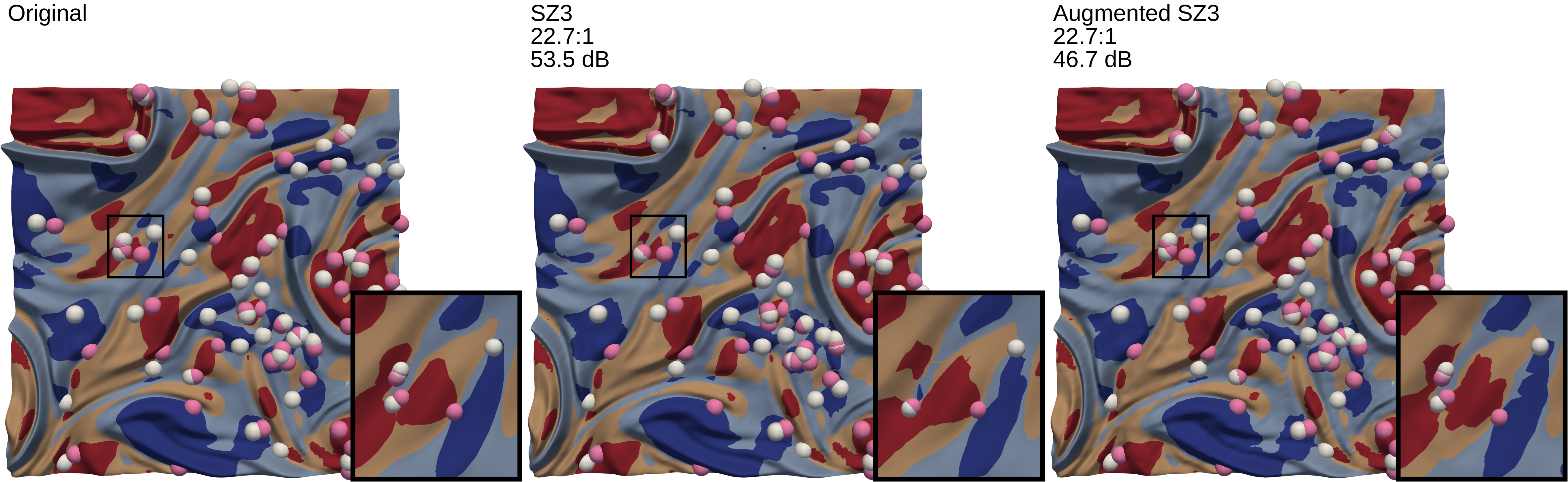}
\vspace{-4mm}
\caption{Visualizing the eigenvector partition of a 2D asymmetric second-order tensor field compressed with classic SZ3 and Augmented SZ3 compressors with topological controls. Left: input data visualized with  degenerate points of the dual-eigenvector field. Trisectors are in white, wedges are in pink. Middle: reconstructed data using SZ3, labeled with compression ratio and \new{Peak Signal-To-Noise Ratio (PSNR)}. Right: reconstructed data using Augmented SZ3 along with compression ratio and PSNR. We provide zoomed-in views that highlight the differences between the classic and Augmented SZ3 results. The colormap is based on a partition of eigenvector manifold. \new{The position of each point on the Z axis corresponds to the Frobenius norm with smoothing applied.}}
\label{fig:teaser}
}

\begin{document}

\maketitle

\input{sec-introduction.tex}
\input{sec-related-work.tex}
\input{sec-background.tex}
\input{sec-method.tex}

\input{sec-results.tex}
\input{sec-limitations.tex}
\input{sec-conclusion.tex}

\acknowledgments{
This work was supported in part by grants from National Science Foundation OAC-2313122, OAC-2313123, and OAC-2313124.}

\bibliographystyle{abbrv-doi-hyperref}
\bibliography{refs-tf-compression.bib}
\newpage
\appendix
\input{appendix-datasets.tex}
\input{appendix-evaluation-metrics.tex}
\input{appendix-parameter-configurations.tex}
\input{appendix-additional-experiments.tex}
\input{appendix-algorithm-details.tex}
\input{appendix-runTimes.tex}
\input{appendix-extra-figures.tex}
\input{appendix-special-cases.tex}
\input{appendix-lemma-proofs.tex}
\input{appendix-partition-correctness-assumptions.tex}
\input{appendix-partition-correctness-support.tex}
\input{appendix-partition-correctness-main-results.tex}
\input{appendix-partition-correctness-edge-cases.tex}

\end{document}

%% file: sec-abstract.tex
In this paper, we present a novel compression framework, {\toolname}, that preserves the topology of 2D symmetric and asymmetric second-order tensor fields defined on flat triangular meshes. A tensor field assigns a tensor---a multi-dimensional array of numbers---to each point in space. \new{Tensor fields, such as the stress and strain tensors, and the Riemann curvature tensor, are essential to both science and engineering.} The topology of tensor fields captures the core structure of data, and is useful in various disciplines, such as graphics (for manipulating shapes and textures) and neuroscience (for analyzing brain structures from diffusion MRI). Lossy data compression may distort the topology of tensor fields, thus hindering downstream analysis and visualization tasks. {\toolname} ensures that certain topological features are preserved during lossy compression. Specifically, {\toolname} preserves degenerate points essential to the topology of symmetric tensor fields and retains eigenvector and eigenvalue graphs that represent the topology of asymmetric tensor fields. \new{{\toolname} scans through each cell, preserving the local topology of each cell, and thereby ensuring certain global topological guarantees.} We showcase the effectiveness of our framework in enhancing \new{the} lossy scientific data compressors SZ3 and SPERR.


%% file: sec-introduction.tex
\section{Introduction}
\label{sec:introduction}

The concept of a tensor is intricate and has evolved across various disciplines from differential geometry to machine learning~\cite{Guo2021}. 
Intuitively, a \emph{tensor} is a mathematical entity that generalizes scalars, vectors, and matrices to higher dimensions while adhering to specific transformation rules. 
It can be represented as a multidimensional array of numbers. 
A \emph{tensor field} is a function that assigns a tensor to each point in a space. 
Tensor fields naturally arise across a variety of scientific domains. 
In fluid dynamics, the velocity gradient tensor field describes how the fluid's velocity changes from one point to another within the flow, providing information about the deformation and rotation of fluid elements \cite{meneveau2011lagrangian, haimes1999velocity}. 
In general relativity, the energy-momentum tensor field describes the matter and energy content of the universe \cite{callan1970new}, whereas the Riemann curvature tensor field describes the geometry of spacetime itself \cite{ahsan2008riemann}. 
In materials science, the stress and strain tensor fields are fundamental concepts that respectively describe the internal forces within a material and how the material deforms under external loads \cite{savage1981stress,schlenker1978strain}. 
In neuroscience, diffusion tensor fields (from diffusion tensor imaging or DTI) captures  the diffusion of water molecules in the white matter for the study of neural pathways \cite{le2001diffusion}.

The topology of tensor fields reflects the core structure of data and proves valuable across various fields, including graphics~\cite{alliez2003anisotropic, chen2008interactive, xu2017autonomous} (for manipulating shapes and textures) and neuroscience~\cite{jankowai2019robust} (for analyzing brain structures through DTI). 
The topology of a tensor field depends on the dimension of the domain, the size of the tensor, and whether or not the tensors are symmetric. 
For instance, the stress tensor is symmetric, while the velocity gradient tensor is asymmetric.
In this work, we consider 2D symmetric and asymmetric second-order tensor fields, where each field maps points in $\Rspace^2$ to $2 \times 2$ matrices. 

The topology of a 2D symmetric second-order tensor field is defined by its corresponding eigenvector fields~\cite{jankowai2019robust}, which are obtained by calculating the eigenvectors and eigenvalues of each tensor. The topology includes \emph{degenerate points} (where the eigenvalues are equal) and their connections through \emph{tensorlines}  (lines aligned with the eigenvector fields). \new{The topology determines high-level patterns in the tensorlines, which are one of the main tools used to analyze symmetric tensor fields. It is also often used to derive various types of visualizations \cite{jankowai2019robust, tricoche2001tensor}. Further, in certain scenarios, the degenerate points of tensor fields have real physical interpretations \cite{lavin1997singularities, zhang2017applying, Filho2016Automatic}. }
On the other hand, the topology of a 2D asymmetric second-order tensor field is derived from certain partitions of the domain according to decompositions of asymmetric tensors. \new{Asymmetric tensor fields are very difficult to analyze, and their topology is one of the only tools available for analysis and visualization}~\cite{zhang2008asymmetric, LinYehLaramee2012, palke2011asymmetric, khan2019multi, auer2013automatic, hergl2021visualization}.

Despite their utility, working with tensor fields is often constrained by significant storage requirements. For instance, storing a tensor field composed of \( 2 \times 2 \) matrices necessitates 256 bits of storage per data point (assuming double precision). Consequently, the storage and transmission of tensor field data can present substantial challenges within scientific workflow\new{s}.
One possible solution to address this issue is to utilize lossy compressors developed specifically for scientific data, which typically ensure strict pointwise error bounds.  
However, applying these compressors to tensor fields often distorts \new{high-level geometric or topological patterns in the data,} even with minimal error bounds, thereby hindering downstream analysis and visualization. \new{This is especially problematic if topological analysis is prioritized}.

In this paper, we introduce {\toolname}, a novel framework that augments any error-bounded lossy compressor to preserve the topology of 2D symmetric and asymmetric second-order tensor fields defined on flat triangular meshes. It iterates through the cells in the mesh, enforcing specific local properties that, when preserved across all cells, ensure that global topological properties are preserved. In summary: 
\begin{itemize}[noitemsep,leftmargin=*]
\item {\toolname} could be used to enhance any error-bounded lossy compressor to provide topological guarantees. For symmetric tensor fields, it preserves the degenerate points in the eigenvector fields. For asymmetric tensor fields, it maintains the integrity of eigenvector and eigenvalue graphs. Additionally, {\toolname} guarantees a strict pointwise error bound for each entry of every tensor.
\item As a secondary contribution, we present a new topological invariant that characterizes the topology of a mesh cell in a 2D asymmetric tensor field. 
\item We demonstrate the effectiveness of our framework through a comprehensive evaluation, enhancing two lossy scientific data compressors—SZ3 and SPERR—while ensuring topological guarantees.
\end{itemize}

%% file: sec-related-work.tex
\section{Related Work}
\label{sec:related-work}

We review error-bounded lossy compression for scientific data, the topology and visualization of tensor fields, and topology-preserving compression.

\para{Lossy data compression.} Lossless compression ensures perfect data recovery but achieves limited compression ratios for scientific data, typically less than \(2\times\) according to~\cite{son2014data}. 
Given the massive size of scientific datasets, lossy compression is generally preferred, \new{and} error-bounded lossy compression has been developed \new{for applications that} require guarantees on data distortion.~\new{See \cite{di2024survey} for a survey.}

There are two main types of error-bounded lossy compressors: prediction-based and transformation-based. Prediction-based methods use interpolation strategies to generate an initial guess for the data. They also compute and encode any corrections that must be made to the initial guess to ensure that a strict error bound is maintained. ISABELA~\cite{lakshminarasimhan2013isabela}, one of the first error-controlled prediction-based compressors, uses B-splines to predict data. SZ3~\cite{liang2022sz3,zhao2021optimizing,liang2018error}, the most recent general release in the SZ compressor family, uses a Lorenzo predictor~\cite{ibarria2003out}, cubic spline interpolation, and linear interpolation. 

Recently, deep learning models have been employed as predictors for data compression, such as the autoencoder~\cite{le2023hierarchical} and implicit neural representation (INR)~\cite{lu2021compressive}.
An autoencoder is a neural network with two components: an encoder and a decoder.
The encoder produces low-dimensional representations of the input data, while the decoder reconstructs the original input data from the output of the encoder.  
An INR trains a small neural network that can be used to approximate the ground truth. 
The neural network itself is shipped as a compressed file, and to decompress it, one must simply evaluate the network on an appropriate input. 
One notable INR model for volumetric scalar fields is Neurcomp~\cite{lu2021compressive}. However, neural-based compression is computationally expensive and suffers from limited performance.

Transform-based lossy compressors rely on domain transformations for data decorrelation. 
For instance, ZFP~\cite{lindstrom2014fixed} divides data into small blocks. Each block is separately compressed using exponent alignment for fixed point conversion, a near-orthogonal domain transform, and embedded encoding. TTHRESH~\cite{ballester2019tthresh} treats the entire dataset as one tensor, and uses singular value decomposition (SVD) to improve the decorrelation efficiency for high-dimensional data. SPERR~\cite{li2023lossy} uses wavelet transforms and SPECK coding~\cite{pearlman2004efficient} to compress data.

\para{Tensor compression.}
Data are typically stored in multidimensional arrays across a variety of domains. As such, tensor compression strategies have been developed to compress multidimensional arrays. However, these strategies are not designed for compressing tensor fields specifically. Rather, they compress large individual multidimensional arrays. Some strategies use tensor decompositions such as \new{SVD}, \new{canonical polyadic decomposition} or Tucker \new{decomposition}. One example is TTHRESH~\cite{ballester2019tthresh}, which uses SVD. Recently, deep-learning based approaches have been proposed for tensor compression, such as NeuKron~\cite{kwon2023neukron}, TensorCodec~\cite{kwon2023tensorcodec}, and ELiCiT~\cite{Ko2024Elicit}. 
\new{While tensor compressors are not fundamentally different from scientific compressors, scientific compressors are not typically labeled as tensor compressors due to differences in terminology across domains.}

\para{\new{Topology and visualization of 2D tensor fields.}} 
2D tensor fields are inherently difficult to visualize. Visualization typically involves glyphs \cite{ennis2005visualization}, color plots, \cite{LinYehLaramee2012}, or streamlines \cite{zheng2003hyperlic}; see~\cite{hergl2021visualization} for a survey. The topology of 2D symmetric tensor fields was first applied to visualization by Delmarcelle \cite{delmarcelle1995visualization}. Since then, the topology of 2D symmetric tensor fields has been applied to numerous problems in computer graphics such as remeshing \cite{alliez2003anisotropic}, street layout generation \cite{chen2008interactive} and scene reconstruction \cite{xu2017autonomous}. It has also been used in the visualization of tensor fields \new{\cite{jankowai2019robust, tricoche2001tensor}} and rotation fields \cite{palacios2010interactive}. \new{Outside of visualization, the degenerate points have been shown to have meaningful physical interpretations in certain domains \cite{lavin1997singularities, zhang2017applying, Filho2016Automatic}}.
The topology of asymmetric tensor fields has been utilized to support their visualization~\cite{LinYehLaramee2012,palke2011asymmetric,khan2019multi}.

\para{Topology-preserving compression.}~Recently, a number of topology-preserving data compression techniques have been developed. Most of the work thus far has pertained to scalar fields; to our knowledge, no topology-preserving compression technique has been developed for tensor field data. Soler et al.\cite{soler2018topologically} developed a compressor that preserves the persistence diagram of a scalar field, up to \new{a user-specified} persistence \new{threshold}. 
Yan et al.~\cite{yan2023toposz} developed TopoSZ, which preserves the contour tree of a scalar field, up to a user-specified persistence threshold. Gorski et al. \cite{gorski2025general} developed a contour tree preserving compression strategy with improved performance compared to TopoSZ. Li et al. developed mSZ \cite{li2024msz}, which preserves the Morse--Smale segmentation of a scalar field. For vector fields, Liang et al. \cite{liang2022toward} developed cpSZ, which preserves the critical points of a vector field. \new{Later, Xia et al. \cite{Xia2025tspsz} developed TspSZ, which preserves the entire topological skeleton.}

%% file: sec-background.tex
\section{Technical Background}
\label{sec:background}

\subsection{Tensors and Tensor Fields} 
\label{sec:tensors}

A tensor may be represented as a multidimensional array. The rank of a tensor indicates the number of indices required to specify its components: a scalar has rank 0, a vector has rank 1, and a matrix has rank 2. 
Tensors play a significant role in science and engineering, such as the Riemann curvature tensor from general relativity, the \new{stress and strain} tensor\new{s} in mechanics, and diffusion tensor imaging in medicine.
We consider 2D symmetric and asymmetric second-order tensor fields, which we refer to simply as tensor fields when the context is clear.  
In this paper, a \emph{tensor} \( T \) is a linear operator that maps any vector \( v \) to another vector \( u = Tv \), where both \( v \) and \( u \) belong to the vector space \( \mathbb{R}^2 \). 
With a chosen basis of $\Rspace^2$, $T$ can be represented by a $2 \times 2$ matrix 
$$T =    \begin{pmatrix}
    	T_{11}      &  T_{12} \\
		T_{12}      &  T_{22} 
	\end{pmatrix}.$$ 
$T$ is \emph{symmetric} when $T_{ij} = T_{ji}$ and \emph{asymmetric} otherwise.
Let $\Tspace$ denote the space of second-order tensors.
A \emph{tensor field} $\TF: \Xspace \to \Tspace$ assigns each point $x$ from a domain $\Xspace$ to a tensor $\TF(x) \in \Tspace$. 

Consistent with previous work \cite{LinYehLaramee2012, jankowai2019robust}, we work with piecewise-linear (PL) tensor fields. For {\toolname}, $\Xspace$ is a flat 2D triangular mesh (i.e. all vertices have zero Gaussian curvature). $\TF$ is stored at vertices of the mesh and PL interpolated. 
Specifically, let $\sigma$ be a triangular cell with vertices $v_1$, $v_2$, and $v_3$. Any point $x \in \sigma$ can be written uniquely based on barycentric coordinates, that is, $x = a_1v_1 + a_2v_2 + a_3v_3$ where $a_i \geq 0$ and $\sum_i a_i=1$. We set $f(x) = a_1f(v_1) + a_2f(v_2) + a_3f(v_3)$. When studying the topology of tensor fields, Khan et al. claimed that this interpolation scheme can lead to discontinuities when vertices have nonzero Gaussian curvature \cite{khan2019multi}. Thus, we restrict $\Xspace$ to be flat.

\subsection{Topology of Tensor Fields}
\label{sec:topology}

The topology of a tensor field varies significantly depending on the order, dimension, and symmetry of the tensor field (e.g.,~\cite{hesselink1997topology, jankowai2019robust, zhang2008asymmetric, roy2018robust, hung2021feature}). 
We focus on the topology of 2D second-order tensor fields. 

\subsubsection{Topology of Second-Order Symmetric Tensor Fields} 
The topology of a second-order symmetric tensor field $f$ consists of \emph{degenerate points} (i.e., points with equal eigenvalues) and their connections via \emph{tensorlines} (i.e., lines that follow the direction of the eigenvector fields)~\cite{DelmarcelleHesselink1994}.  
At a fixed location $x$, \new{let} $f(x) = T$. A symmetric tensor $T$ has two linearly independent real eigenvectors $v_1$ and $v_2$ that correspond to real eigenvalues $\lambda_1$ and $\lambda_2$ respectively. \new{If we order} the eigenvalues such that $\lambda_1 \geq \lambda_2$, then $v_1$ are $v_2$ are the major and minor eigenvectors, respectively. 
The major (resp.,~minor) \emph{eigenvector field} of $f$, denoted as  $e_1$ (resp.,~$e_2$), maps each $x$ to the major (resp.,~minor) eigenvector of $f(x)$. Integrating the eigenvector fields \new{yields} two families of continuous curves, referred to as the major and minor \emph{tensorlines}. 

A \emph{degenerate point} of $f$ is a point where the eigenvalues are identical and the eigenvectors are no longer defined uniquely, that is, $\lambda_1 = \lambda_2$. 
Degenerate points behave like the zeros of a vector field. The topology of $f$ is defined in terms of its degenerate points, and the tensorlines that connect them. 
In the PL setting, there are two types of degenerate points that can occur within a cell: a trisector and a wedge, as shown in \cref{fig:degenerate-points}(A) and (B) respectively. 
They can be classified and detected according to the behavior of the eigenvector fields around them. 

\begin{figure}[!ht]
\centering
    \vspace{-2mm}
    \includegraphics[width=0.7\linewidth]{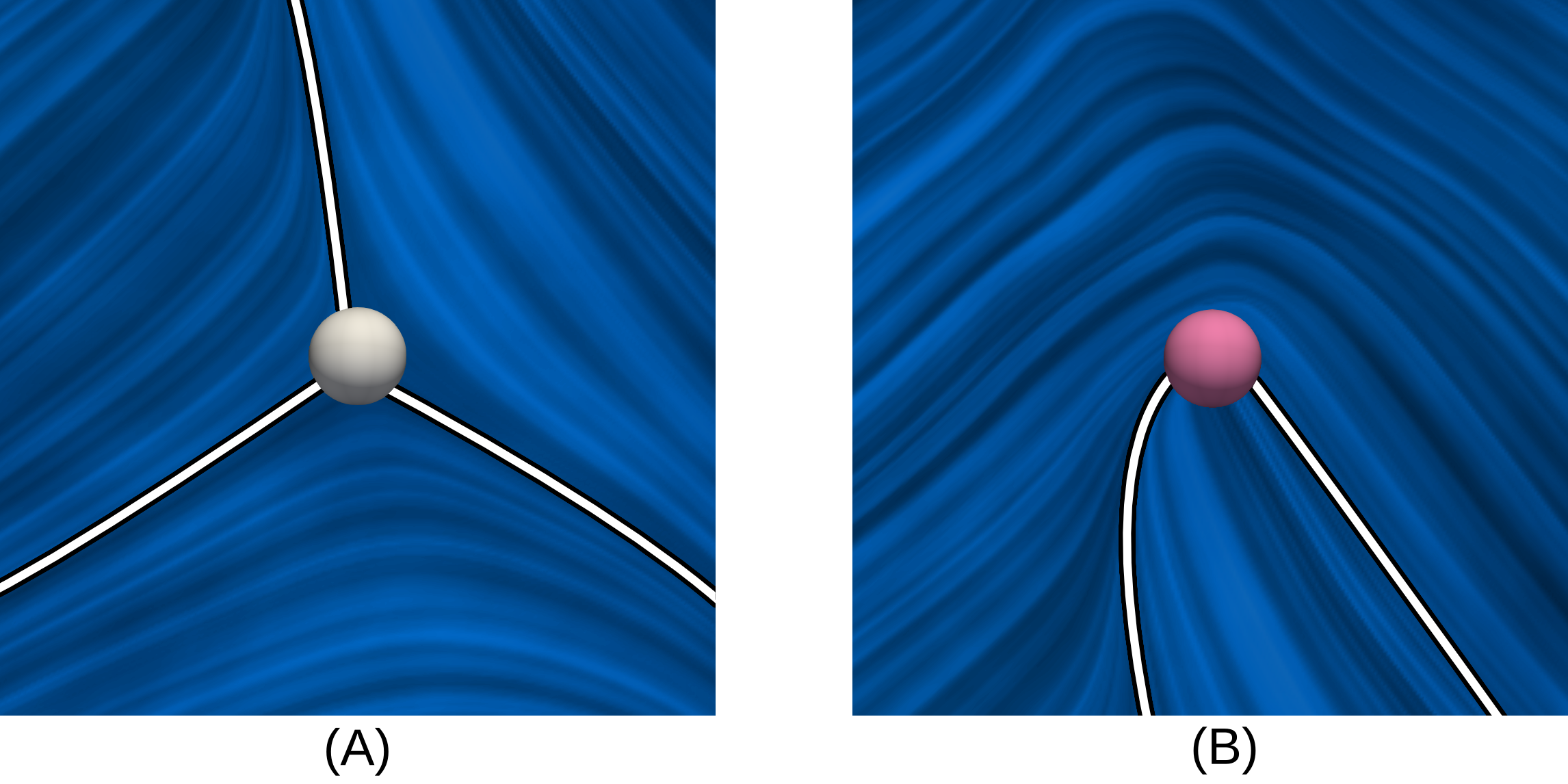}
    \vspace{-4mm}
    \caption{Degenerate points of a 2D second-order symmetric tensor field. The major eigenvector field is visualized using the Line Integral Convolution (LIC) together with major tensorlines passing through the degenerate points; trisectors (A) are in white, wedges (B) are in pink.}
    \label{fig:degenerate-points}
    \vspace{-2mm}
\end{figure}

We detect a degenerate point within a cell $\sigma$ using the concept of a deviator following Jankowai et al.~\cite{jankowai2019robust}. 
Given a cell $\sigma$, let $T_1$, $T_2$, and $T_3$ be tensors (represented as matrices) associated with its vertices in a clockwise order. 
Define the deviator of a tensor $T$ by $D(T) := T - \frac{1}{2}\tr(T)$, where $\tr(T)$ is the trace.  
It has been shown that $D(T)$ has the same eigenvectors as $T$, and $\tr(D(T)) = 0$ by construction~\cite{jankowai2019robust}. For each $T_i$, denote the \new{entries} of its deviator by
\begin{equation}
D(T_i) := \begin{pmatrix} \Delta_i & F_i \\ F_i & -\Delta_i \end{pmatrix}.
\end{equation}
For a pair $T_i$ and $T_j$, let 
\begin{equation}
\label{eqn:lij}
l_{i,j} := \sign(F_j\Delta_i-F_i\Delta_j).
\end{equation}
Then $\sigma$ contains a wedge if $l_{1,2} = l_{2,3} = l_{3,1} =1$, and a trisector if $l_{1,2}=l_{2,3}=l_{3,1}=-1$; otherwise, if $l_{i,j} \neq 0$ for all pairs, $\sigma$ contains no degenerate points. 
Jankowai et al.~\cite{jankowai2019robust} did not consider the scenario where \( l_{i,j} = 0 \) and instead applied small perturbations to prevent it from occurring. We address such a scenario explicitly in~\cref{sec:method-symmetric}.

\subsubsection{Topology of Second-Order Asymmetric Tensor Fields} 
\label{sec:background-asymmetric}

Zhang et al.~\cite{zhang2008asymmetric} introduced the tensor decomposition of any second-order tensor $T$, 
\begin{equation}
\label{eqn:decomposition}
T = \gamma_d \begin{pmatrix} 1 & 0 \\ 0 & 1 \end{pmatrix} + \gamma_r \begin{pmatrix} 0 & -1 \\ 1 & 0 \end{pmatrix} + \gamma_s \begin{pmatrix} \cos(\theta) & \quad\sin(\theta) \\ \sin(\theta) & -\cos(\theta) \end{pmatrix}, 
\end{equation}
where $\gamma_d = \frac{T_{11}+T_{22}}{2}$, $\gamma_r = \frac{T_{21}-T_{12}}{2}$, and 

$$
\gamma_s = \frac{\sqrt{(T_{11}-T_{22})^2 + (T_{12}+T_{21})^2}}{2}.  
$$
$\gamma_d$, $\gamma_r$, and $\gamma_s$ reflect the strengths of isotropic scaling, rotation, and anisotropic stretching, respectively~\cite{zhang2008asymmetric}; where $\gamma_s \geq 0$ and $\gamma_d, \gamma_r \in \Rspace$. 
$\theta \in [0, 2\pi)$ is the angular component of the vector 
\[
\begin{pmatrix}
T_{11} - T_{22}\\
T_{12} + T_{21} 
\end{pmatrix}.
\]

In a PL setting, $\gamma_d$ and $\gamma_r$ interpolate linearly, whereas the interpolation of $\gamma_s$ is more involved; see~\cite{zhang2008asymmetric} for details.

\para{Dual-eigenvector field.} 
Although the eigenvectors of an asymmetric tensor can be complex, Zheng and Pang~\cite{zheng20052d} introduced the major and minor \emph{dual-eigenvector fields}, which are real-valued eigenvector fields derived from an asymmetric tensor field $f$. For an asymmetric tensor $T$, the major and minor \emph{dual-eigenvectors} of $T$ are the major and minor eigenvectors of a symmetric tensor \cite[Theorem 4.2]{zhang2008asymmetric}: 
\begin{equation}
P_T = \frac{\gamma_r}{|\gamma_r|}\gamma_s\begin{pmatrix} \cos\left(\theta+\frac{\pi}{2}\right) & \quad\sin\left(\theta+\frac{\pi}{2}\right) \\ \sin\left(\theta+\frac{\pi}{2}\right) & -\cos\left(\theta+\frac{\pi}{2}\right)
\end{pmatrix}.
\end{equation}
Thus the dual-eigenvector fields are the major and minor eigenvector fields of the symmetric tensor field $x \mapsto P_{f(x)}$. 

Recall that any square matrix T can be written as $T = T_S + T_A$, where $T_S = \frac{1}{2} (T + T^{\top})$ is the symmetric part of $T$ and $T_A$ is the antisymmetric part. The symmetric part of a tensor $T$ is defined as the average of the tensor with its transpose. 
Following \cref{eqn:decomposition}, 
\begin{equation}
\label{eqn:symmetric-part}
T_S = \gamma_d \begin{pmatrix} 1 & 0 \\ 0 & 1 \end{pmatrix} + \gamma_s \begin{pmatrix} \cos(\theta) & \quad\sin(\theta) \\ \sin(\theta) & -\cos(\theta) \end{pmatrix}.  
\end{equation}
Accordingly, let $f_S$ be the symmetric part of $f$. 
A point is \new{an isolated degenerate point (and not part of a larger degenerate curve or region)} in the dual-eigenvector field of $f$ if and only if it is \new{an isolated} degenerate \new{point} in the eigenvector field of $f_S$. Further, the type of that degenerate point (trisector or wedge) remain the same \cite[Theorem 4.4]{zhang2008asymmetric}.

\para{Eigenvector manifold and eigenvalue manifold.} 
To illustrate the structures in asymmetric tensor fields, Zhang et al.~\cite{zhang2008asymmetric} introduced the notions of eigenvalue manifold (a hemisphere) and eigenvector manifold (a sphere), that encode the relative strengths of three components in \cref{eqn:decomposition} affecting the eigenvalues and eigenvectors in the tensor. For this purpose, we assume that $\gamma_d^2 + \gamma_r^2 + \gamma_s^2 = 1$.
Following~\cite{zhang2008asymmetric}, the \emph{eigenvector manifold} is defined as a sphere  
\begin{equation}
\label{eqn:eigenvector-manifold}
\{(\gamma_r,\gamma_s,\theta) \mid \gamma_r^2 + \gamma_s^2 = 1 \text{ and } \gamma_s \geq 0 \text{ and } 0 \leq \theta < 2\pi \}, 
\end{equation}
whereas the \emph{eigenvalue manifold} is defined as a hemisphere 
\begin{equation}
\label{eqn:eigenvalue-manifold}
\{(\gamma_d,\gamma_r,\gamma_s) \mid \gamma_d^2 + \gamma_r^2 + \gamma_s^2 = 1 \text{ and } \gamma_s \geq 0\}.
\end{equation}

\begin{figure}[!ht]
\centering
    \includegraphics[width=1.0\linewidth]{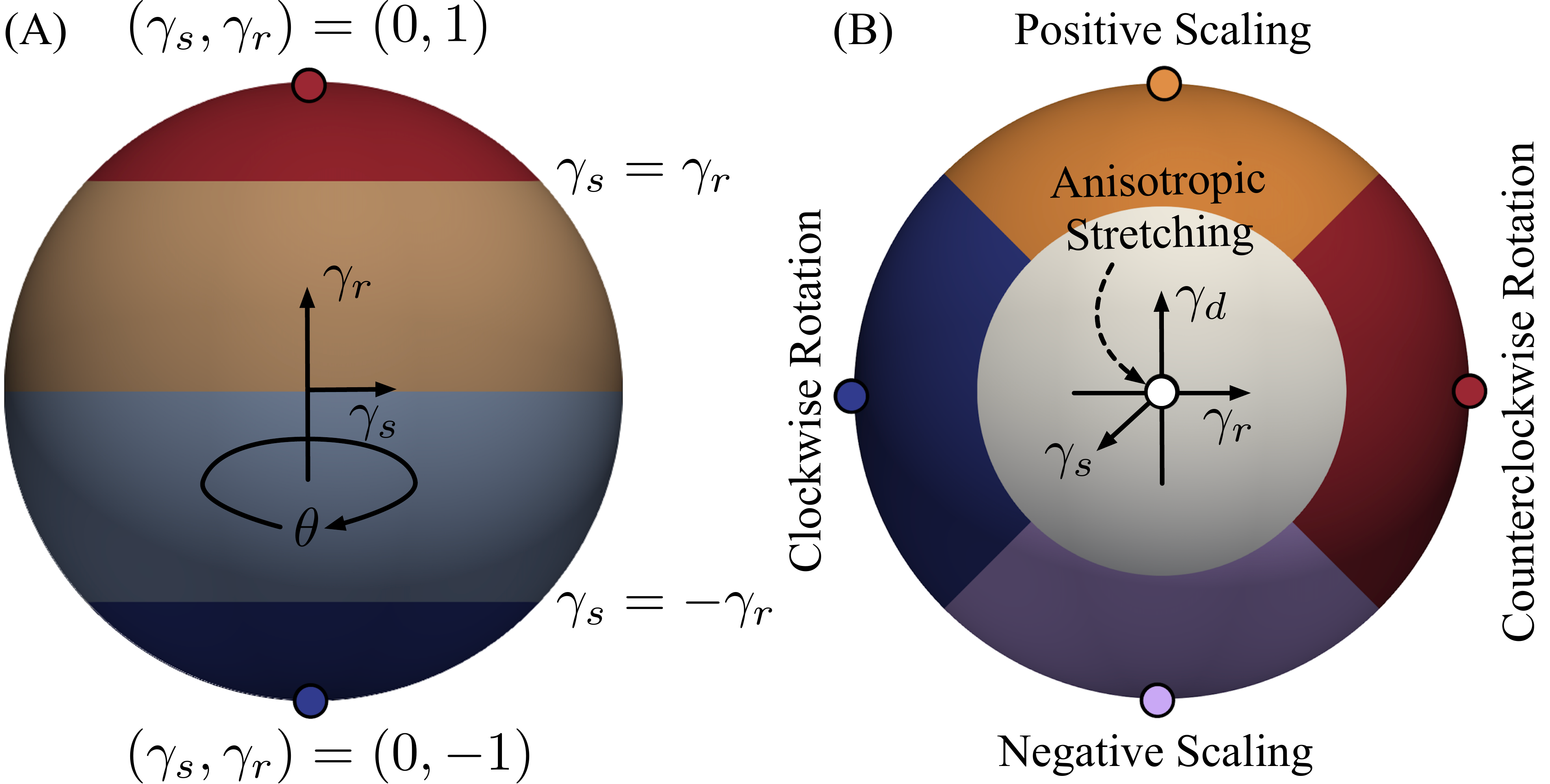}
    \vspace{-6mm}
    \caption{(A) The eigenvector manifold colored according to its partition based on $\gamma_r$ and $\gamma_s$. (B) The eigenvalue manifold colored according to its  partition based on $\gamma_d$, $\gamma_r$, and $\gamma_s$.}
    \label{fig:manifolds}
    \vspace{-3mm}
\end{figure}

\para{Eigenvector partition.}
We visualize the eigenvector manifold of \cref{eqn:eigenvector-manifold} in \cref{fig:manifolds}(A). 
To construct the manifold, we first define a semicircle $\gamma_r^2 + \gamma_s^2 = 1$ (where $\gamma_r \geq 0$). 
Then, we rotate such a semicircle around a vertical axis to obtain a sphere. On this sphere, each longitude line that runs from the North Pole to the South Pole corresponds to a value of $\theta$, and each latitude line corresponds to a value of $\gamma_r$. 

Next, we partition the eigenvector manifold based on (a) whether \(|\gamma_r|\) or \(\gamma_s\) is greater and (b) the sign of \(\gamma_r\).~\cref{fig:manifolds}(A) illustrates four regions from such a partition using a categorical colormap, from the top to bottom: the first region  near the North Pole corresponds to the parameter setting where $\gamma_r > 0$ and $|\gamma_r| > \gamma_s$; the second region  corresponds to $\gamma_r > 0$ and $|\gamma_r| < \gamma_s$; the third region corresponds to $\gamma_r < 0$ and $|\gamma_r| < \gamma_s$; and the fourth region near the South Pole corresponds to $\gamma_r < 0$ and $|\gamma_r| > \gamma_s$. 
The two regions near the equator correspond to real eigenvalues, and the regions near the Poles correspond to complex eigenvalues. 

For a tensor field $f$, each tensor $f(x)$ at a location $x \in \Xspace$ maps to a point within a partition of the eigenvector manifold. Therefore, we can visualize the domain of $f$ based on the partition of its corresponding eigenvector manifold; this is referred to as the \emph{eigenvector partition}, as illustrated in~\cref{fig:manifold-topology}(A). The north and south poles correspond to \new{isolated} degenerate points of the dual eigenvector field, and such points are typically part of the partition visualization.

\begin{figure}[!t]
\centering
    \vspace{-2mm}
    \includegraphics[width=1.0\linewidth]{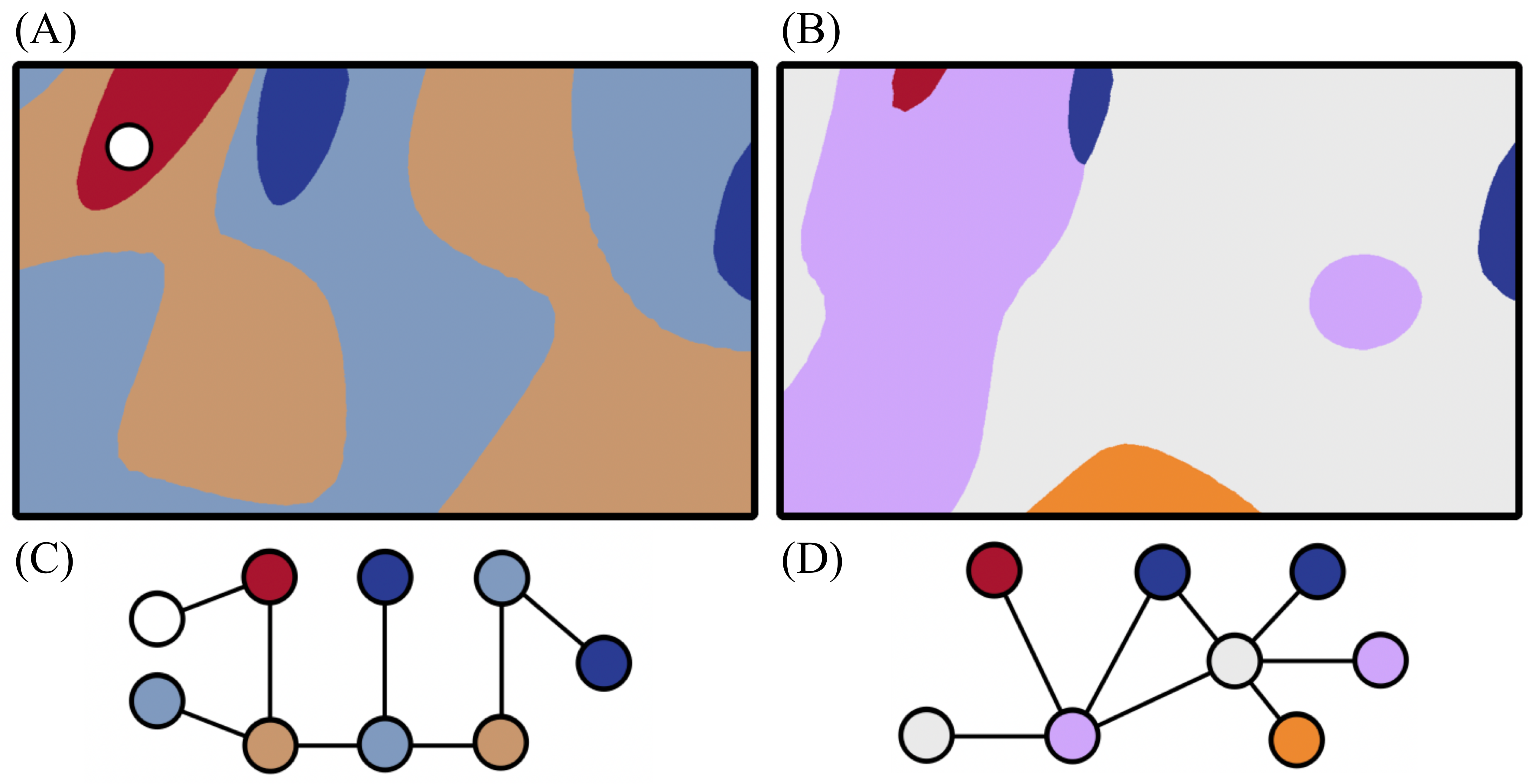}
    \vspace{-8mm}
    \caption{A portion of the Miranda dataset. (A) Eigenvector partition of the domain following the colormap of \cref{fig:manifolds}(A); degenerate points  of the dual-eigenvector field are shown in white (trisectors) or pink (wedges). (B) Eigenvalue partition of the domain following the colormap of \cref{fig:manifolds}(B). (C) Eigenvector graph corresponding to the eigenvector partition in (A). (D) Eigenvalue graph corresponding to the eigenvalue partition in (B).}
    \label{fig:manifold-topology}
    \vspace{-6mm}
\end{figure}

\para{Eigenvalue partition.} The eigenvalue manifold of \cref{eqn:eigenvalue-manifold} is designed to highlight the dominant operation between isotropic scaling ($\gamma_d$), rotation ($\gamma_r$), and anisotropic stretching ($\gamma_s$). We visualize its corresponding hemisphere in~\cref{fig:manifolds}(B), where the point $(\gamma_r,\gamma_d,\gamma_s) = (0,0,1)$ corresponds to the point labeled ``anisotropic stretching.'' 

We then partition the eigenvalue manifold according to (a) which coefficient has the largest magnitude and (b) the sign of that coefficient. Each region in the partition corresponds to a different type of linear transformation. 
In \cref{fig:manifolds}(B): the top orange region (Positive Scaling) corresponds to when $\gamma_d$ has the largest magnitude and $\gamma_d > 0$; the red region (Counterclockwise Rotation) corresponds to when $\gamma_r$ has the largest magnitude and $\gamma_r > 0$; the purple region (Negative Scaling) corresponds to when $\gamma_d$ has the largest magnitude and $\gamma_d < 0$; the blue region (Clockwise Rotation) corresponds to when $\gamma_r$ has the largest magnitude and $\gamma_r < 0$; and finally, the white region in the middle corresponds to when $\gamma_s$ has the largest magnitude. 

Similarly, a partition of the eigenvalue manifold could be used to visualize the domain of a tensor field, referred to as the \emph{eigenvalue partition}, as illustrated in~\cref{fig:manifold-topology}(B).   

\para{Eigenvector graph and eigenvalue graph.}
To describe the topology of an asymmetric tensor field, Lin et al.~\cite{LinYehLaramee2012} introduced the notion of an \emph{eigenvector graph}, that is, the dual graph of the eigenvector partition of the domain.  
Each node corresponds to a region of the partition or one of the poles of the eigenvector manifold, and is colored by the region type (or the degenerate point type). 
Each edge encodes an adjacency relation among a pair of regions (or a region and a degenerate point).
We visualize a tensor field based on its eigenvector partition in~\cref{fig:manifold-topology}(A), and its eigenvector graph in (C).  

As illustrated in \cref{fig:graph-construction}(A), to compute the eigenvector graph, we first trace the boundary curves defined by $\gamma_r^2 = \gamma_s^2$ (a conic section) and $\gamma_r = 0$ (a line) within each cell. These curves do not intersect with each other and simply divide a cell according to the eigenvector partition. Next, we iteratively merge adjacent regions of the same type in neighboring cells. Finally, we detect each connected component of the partition, from which we compute the eigenvector graph. 

Analogously, Lin et al.~\cite{LinYehLaramee2012} also introduced the \emph{eigenvalue graph} as the dual of the eigenvalue partition of the domain. We visualize a tensor field based on its eigenvalue partition in~\cref{fig:manifold-topology}(B), and its corresponding eigenvalue graph in~\cref{fig:manifold-topology}(D). Following~\cite{LinYehLaramee2012}, as illustrated in \cref{fig:graph-construction}(B), to compute the eigenvalue graph, we trace the boundary curves within each cell: $\gamma_d^2 = \gamma_r^2$ (two intersecting lines), $\gamma_d^2 = \gamma_s^2$ (a conic section), and $\gamma_r^s = \gamma_s^2$ (another conic section). We then  merge adjacent regions of the same type in the same cell (see~\cref{fig:graph-construction}(C)) and among neighboring cells to construct the eigenvalue graph.  

For compression purposes, we focus on preserving the topology inside each cell, that is, the cell-wise eigenvector (or eigenvalue) partition, without explicitly constructing the eigenvector (or eigenvalue) graph. 

\begin{figure}[!t]
\centering
    \includegraphics[width=1.0\linewidth, trim={0 0 0 0mm}, clip]{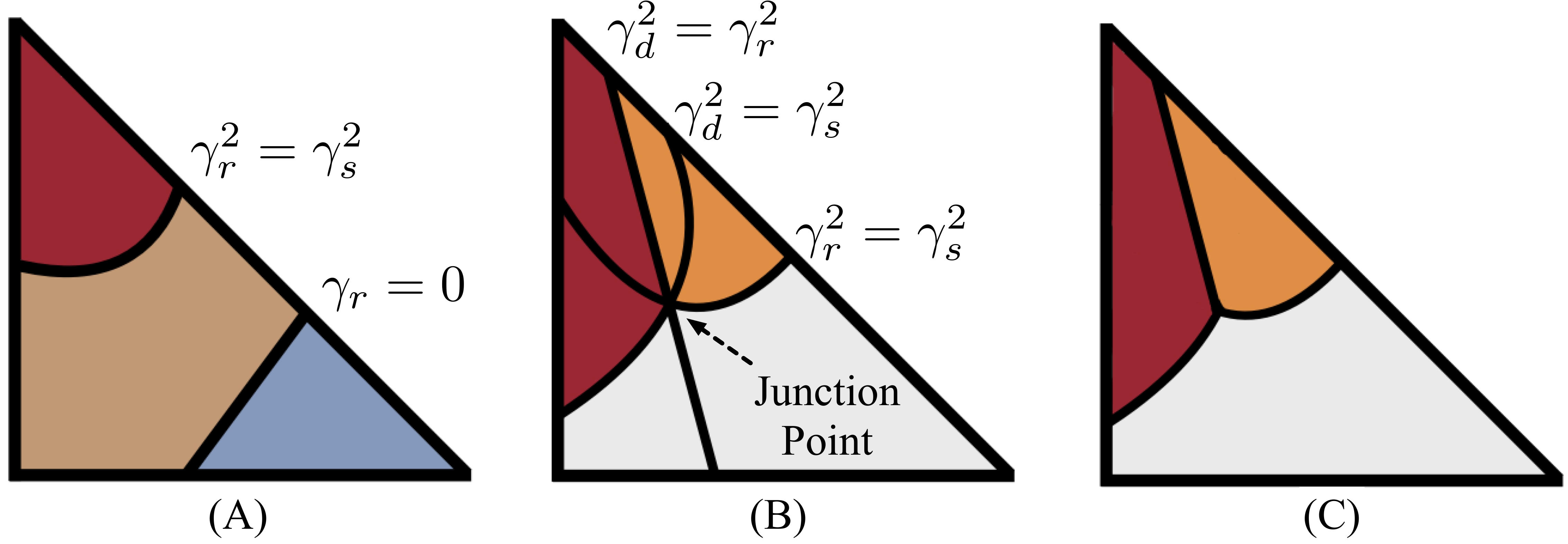}
    \vspace{-6mm}
    \caption{(A) A cell is partitioned into three regions according to the eigenvector manifold by tracing boundary curves $\gamma_r^2=\gamma_s^2$ and $\gamma_r = 0$. 
    (B) A cell is partitioned into six regions by tracing the boundary curves of $\gamma_d^2=\gamma_r^2$, $\gamma_d^2 = \gamma_s^2$ and $\gamma_r^2 = \gamma_s^2$ that intersect at the junction point. The regions are classified according to the eigenvalue partition. (C) Adjacent regions of the same type based on eigenvalue partition in (B) are merged inside a cell.}
    \label{fig:graph-construction}
    \vspace{-6mm}
\end{figure}

\subsection{Quantization in Lossy Compression}
\label{sec:quantization}

We review two quantization techniques, linear-scaling quantization and logarithmic-scaling quantization, used by {\toolname}.

\para{Linear-scaling quantization.}~A popular lossy data compressor, SZ1.4 \cite{tao2017significantly}, introduces \emph{linear-scaling quantization} that ensures a pointwise absolute error bound $\xi$ is maintained.
Suppose that $f:\Xspace \rightarrow \Rspace$ (defined on a finite domain $\Xspace$) is a scalar field to be compressed. Suppose that $g:\Xspace \rightarrow \Rspace$ is an initial guess for $f$ (e.g.,~from a regression predictor). Let $f':\Xspace \rightarrow \Rspace$ be the final decompressed data. 
Then for each $x \in \Xspace$, we can shift $g(x)$ by an integer multiple of $2\xi$ to obtain a final value of $f'(x)$ such that $|f'(x) - f(x)| \leq \xi$.
To conceptualize linear-scaling quantization for a point $x$ in~\cref{fig:linear-scaling-quantization}, we divide the real line into intervals of width $2\xi$, one of which is centered on $g(x)$. Then we compute how many intervals to shift $g(x)$ to obtain a value for $f'(x)$ that lies in the same interval as $f(x)$. By construction, if $f(x)$ lies in an interval of width $2\xi$ centered on $f'(x)$, it must hold that $|f(x) - f'(x)| \leq \xi$. 

\begin{figure}[!ht]
    \vspace{-4mm}
    \includegraphics[width=\linewidth]{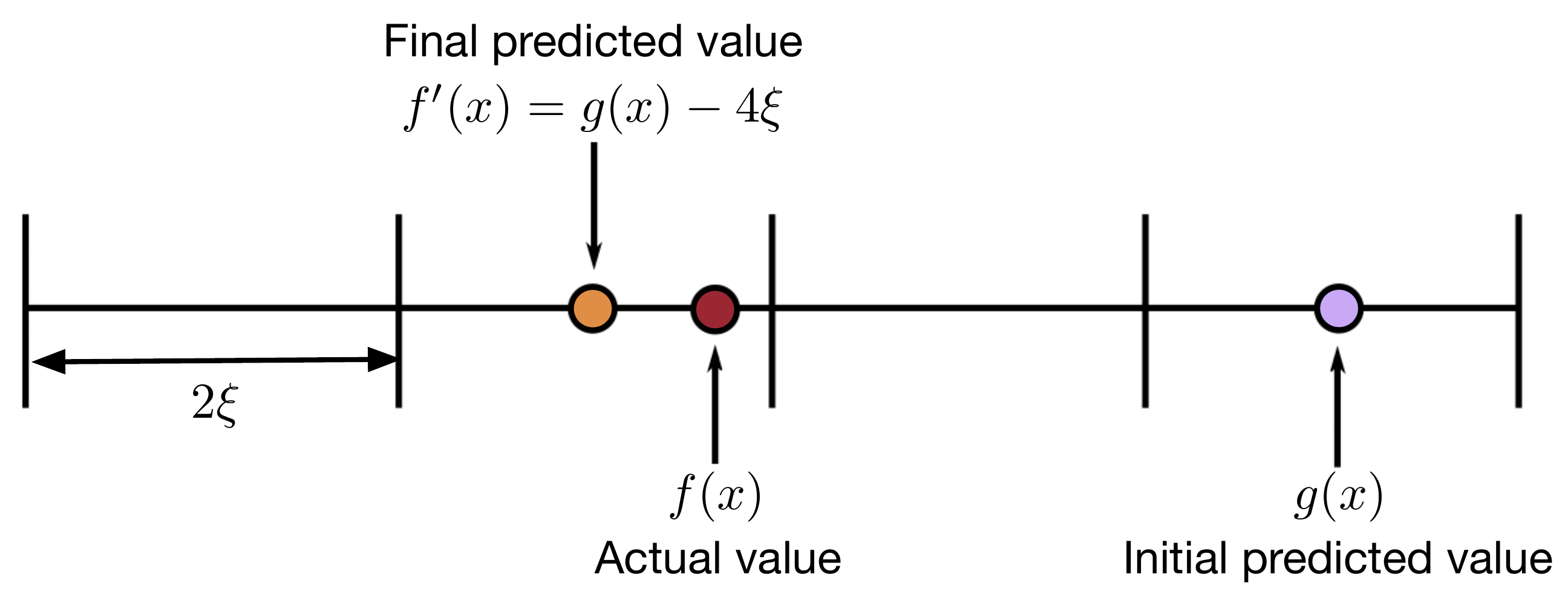}
    \vspace{-8mm}
    \caption{A standard implementation of linear-scaling quantization.}
    \label{fig:linear-scaling-quantization}
\end{figure}

During construction, we assign to each $x \in \X$ some integer $n_x$ such that $f'(x) = g(x) + n_x(2\xi)$, e.g., $n_x = -2$ in~\cref{fig:linear-scaling-quantization}. If the entropy of the $\{n_x\}$ is low (e.g. the $\{n_x\}$ are mostly zero), then the $\{n_x\}$ can be efficiently compressed using an entropy-based compressor, such as the Huffman coding. More accurate predictions for $g(x)$ generally lead to distributions of $\{n_x\}$ with lower entropy.

\para{Logarithmic-scaling quantization.}~In topology-preserving data compression, it can be advantageous for certain data points to have predicted values that are nearly identical to the ground truth. Gorski et al.~\cite{gorski2025general} introduced logarithmic-scaling quantization, enabling a compressor to apply linear-scaling quantization with varying error bounds for different points without needing to store additional information about the specific error bound used. 
Each $x \in \Xspace$  is assigned a precision $p_x$, so that its error bound is $2\xi / (2^{p_x})$. 
The value of $p_x$ is set independently for each point. If the compressor determines that, for a point $x$, $f'(x)$ must be very close to $f(x)$, then it can set $p_x$ to be high.

To encode a point using logarithmic-scaling quantization, we perform standard linear-scaling quantization on each $x \in \Xspace$ using an error bound of $2\xi / (2^{p_x})$ to obtain a quantization number $n_x$, such that $x$ will have the decompressed value $f'(x) = g(x) + 2\xi n_x / (2^{p_x})$. Let $P$ be the maximum value for $p_x$. When encoding the compressed data, we store the integer $a_x \gets n_x \times 2^{P-p_x}$. Then, when decompressing the data, the decompressed value for $x$ is $f'(x) = g(x) + 2\xi a_x / (2^{P})$. Notice that $x$ will have its correct decompressed value because 
\begin{equation}
    f'(x) = g(x) + \frac{2\xi a_x}{2^P} = g(x) + \frac{2\xi n_x2^{P-p_x}}{2^{P}} = g(x) + \frac{2\xi n_x}{2^{p_x}}.
\end{equation}

%% file: sec-method.tex
\section{Method}
\label{sec:method}

Although the topology and topological guarantees differ for symmetric and asymmetric tensor fields, our overall compression pipeline remains largely the same in both cases. We describe our general pipeline in~\cref{sec:method-overview}. We describe the specifics for symmetric tensor fields in~\cref{sec:method-symmetric} and asymmetric tensor fields in~\cref{sec:method-asymmetric}.

\subsection{An Overview of Compression Pipeline}
\label{sec:method-overview}

We store an asymmetric tensor field as four scalar fields that correspond to the four entries of a $2 \times 2$ matrix. For a symmetric tensor field, we only store three scalar fields due to symmetry (i.e.,~for any symmetric tensor $T$, $T_{1,2} = T_{2,1}$). In addition to providing topological guarantees, our compression pipeline ensures that individual entries of a tensor do not vary by more than a user-specified error bound $\xi$. 
 
Following the previous work by Gorski et al.~\cite{gorski2025general}, we design an \emph{augmentation layer} that runs on top of an existing compressor (referred to as a \emph{base compressor}) and corrects the topology of the decompressed output. {\toolname} can augment any error-bounded lossy compressor to provide topological guarantees. Our pipeline is as follows: 

\para{\underline{Step 1. Base compressor.}}~We compress each scalar field using the base compressor with an error bound $\xi$. This operation ensures that the error of individual entries of each tensor is at most $\xi$. We then decompress the data to analyze any changes that must be made to preserve topology. We refer to the compressed-then-decompressed output of the base compressor as the \emph{intermediate data}. 

\para{\underline{Step 2. Cell correction.}}~{\toolname} ensures that the topology of each individual cell is maintained. \new{We scan} across each cell, making corrections to the tensors at the vertices of each cell as needed. We describe this step in more detail in~\cref{sec:method-symmetric} and~\cref{sec:method-asymmetric}. 

\new{If, when processing a cell $\sigma$, we alter the tensor $f(x)$ at some vertex $x$, other cells that share $x$ as a vertex may be altered topologically. Thus, we re-process every previously processed cell $\sigma'$ that shares $x$ as a vertex, and track which cells must be re-processed using a stack. This strategy can create cycles between adjacent cells. To prevent infinite loops, we store a tensor losslessly if it is modified 20 times.}

\para{\underline{Step 3. Lossless compression.}}~Our cell compression step outputs integer values for each cell according to quantization strategies detailed in~\cref{sec:quantization}. We also generate some integers that serve as flags. We encode these integers using Huffman coding. We combine the output of Huffman coding, any information stored losslessly, and the four compressed fields from the base compressors into one tar archive, and then compress the archive using ZSTD (a lossless compressor).

\subsection{Compression of Symmetric Tensor Fields}
\label{sec:method-symmetric}

\para{Topological guarantees.}~{\toolname} ensures that each cell retains its classification (i.e.,~degeneracy type): whether it contains no degenerate points, a single degenerate point, a degenerate line (i.e.,~a line of degenerate points), or is entirely degenerate (i.e.,~a cell of degenerate points). For cells with a single degenerate point in their interior, we further distinguish between trisectors and wedges. Additionally, for cells containing a degenerate line or a degenerate vertex, {\toolname} preserves the precise location of the degenerate feature. 

\para{Theoretical basis.}~We provide the theoretical basis for the cell preservation step, where the lemmas are proven in the supplement. In our first lemma, we apply the decomposition in \cref{eqn:decomposition} to symmetric tensors. In the symmetric case, we must have $\gamma_r = 0$. We then obtain the decomposition of \cref{eqn:symmetric-part}, where the $\gamma_s$ term corresponds to the deviator.
\vspace{-6mm}
\begin{restatable}{lemma}{lemmaSymmetricTheta}
Let $T_1$ and $T_2$ be two $2\times 2$ symmetric tensors with nonzero deviators. 
Let $l_{1,2}$ be defined following~\cref{eqn:lij}. 
Then $l_{1,2}$ depends only on $\theta_1$ and $\theta_2$, where each $\theta_i$ is taken from the decomposition in~\cref{eqn:symmetric-part}.
\label{lemma:symmetric-theta}\
\end{restatable}
\vspace{-9mm}

\begin{restatable}{lemma}{lemmaSymmetricZero}
Let $x_1$, $x_2$, and $x_3$ be the vertices of a cell. 
\begin{itemize}[noitemsep]
\item[(a)] Suppose that exactly one vertex has a tensor $T$ with a deviator $D(T)$ equal to zero. (w.l.o.g.,~suppose~$D(f(x_1)) = 0$).
\begin{itemize}[leftmargin=*]
\item[(i)] If $l_{2,3} \neq 0$, then $x_1$ is the only degenerate point in the cell.
\item[(ii)] If $l_{2,3} = 0$, then there exists some $k$ such that $D(f(x_2)) = kD(f(x_3))$. If $k > 0$ then $x_1$ is the only degenerate point in the cell. If $k < 0$ then the cell contains a degenerate line.
\end{itemize}
\item[(b)] If exactly two vertices have a tensor with a deviator equal to zero, then the cell contains a degenerate line connecting them.
\end{itemize}
\label{lemma:symmetric-zero}
\end{restatable}
\vspace{-4mm}
\para{Cell correction.} We scan through each cell $\sigma$ of the mesh. Let $T_1$, $T_2$, and $T_3$ be the tensors at the vertices of $\sigma$ for the ground truth data, and let $T_1'$, $T_2'$, and $T_3'$ be the corresponding tensors in the intermediate data.
Following the tensor decomposition of~\cref{eqn:symmetric-part}, we obtain coefficients $\gamma_{d,1}$ for $T_1$, $\gamma_{d,1}'$ for $T_1'$, etc. We compute $l_{1,2}$, $l_{2,3}$, and $l_{3,1}$ for the ground truth data at $\sigma$ and evaluate the following two cases. 

\underline{Case 1.}~All $l_{i,j} \neq 0$. We first verify whether the degeneracy type of \(\sigma\) (i.e., trisector, wedge, or non-degenerate) in the intermediate data matches the ground truth. If it does, we proceed to the next cell. Otherwise, we examine \(\theta_1'\), \(\theta_2'\), and \(\theta_3'\) and identify the \(\theta_i'\) value that deviates the most from the ground truth. This value is then corrected to be $\theta_i''$ using linear-scaling quantization with an error bound of \(2\pi / (2^6-1)\), ensuring that the quantization number \(n_x\) can be stored using six bits. Once \(\theta_i''\) is computed, we use it along with the coefficients \(\gamma_{d,i}'\) and \(\gamma_{s,i}'\) to reconstruct a new tensor \(T_i''\) based on the decomposition. 

According to~\cref{lemma:symmetric-theta}, adjusting the values of $\theta_i'$ is sufficient to ensure that the $l_{i,j}$ match between the ground truth and reconstructed data. Likewise, Jankowai et al.~\cite{jankowai2019robust} showed that the $l_{i,j}$ determine the degeneracy type of a cell $\sigma$. If adjusting one of the $\theta_i'$ does not correct the topology of the cell, we adjust the other two values as necessary.

In rare cases, due to limited precision, adjusting all three of the $\theta_i'$ using linear-scaling quantization does not preserve the topology. It is also possible that the error bound may not be maintained. In such cases, we store the deviator losslessly, and if necessary the entire tensor. We provide specifics in the supplement.

\underline{Case 2}.~One of the $l_{i,j} = 0$. In this case, if none of the $T_i$ are equal to zero, we store the vertices of $\sigma$ losslessly. Otherwise, we store the $T_i$ that equal zero losslessly, and compute the topology of $\sigma$ following~\cref{lemma:symmetric-zero} for both the ground truth and reconstructed data. If the topology between the ground truth and the reconstructed data does not match, we store the vertices of $\sigma$ losslessly. If any of the $T_i$ satisfy $D(T_i) = 0$ but $T_i \neq 0$, then we also store the vertices of $\sigma$ losslessly.

\subsection{Compression of Asymmetric Tensor Fields}
\label{sec:method-asymmetric}

\para{Topological guarantees.}~For each cell \(\sigma\), {\toolname} preserves the topological structure of both the eigenvalue and eigenvector partitions. Consequently, the eigenvector and eigenvalue graphs are maintained across the entire domain. Since the \new{isolated} degenerate points of the dual-eigenvector fields are part of the eigenvector graph, we also preserve the topology of the dual-eigenvector fields using the same approach as for symmetric tensor fields, \new{operating on the symmetric component of the tensor field. However, when the deviator would be stored losslessly, we instead store $\theta$ losslessly.}
{\toolname} allows users to choose whether to preserve topology based on the eigenvector manifold, eigenvalue manifold, or both. While we describe the process for preserving both, it can be easily adapted to maintain only one. 

For any cell \(\sigma\), preserving its topology involves two main steps. First, we ensure that each vertex of \(\sigma\) is correctly classified within the appropriate type of region in both the eigenvalue and eigenvector partitions. Next, we preserve the internal topology of \(\sigma\) by maintaining the locations and connectivity of the different partition regions. These processes are detailed in \cref{sec:method-asymmetric-vertex} and \cref{sec:method-asymmetric-internal}, respectively. Additionally, our method accounts for various special cases that require careful handling, which we discuss in the supplement. 

\subsubsection{Asymmetric Cell Correction: Vertex Correction}
\label{sec:method-asymmetric-vertex}

In this portion of the algorithm, we ensure that if $x$ is a vertex of a cell $\sigma$, and then the classifications of $f(x)$ and $f'(x)$ match according to the eigenvalue and eigenvector partitions.

\para{Theoretical basis.} We provide the theoretic basis below, see the supplement for proofs. We define a function $\sign : \Rspace \rightarrow \Rspace$ that satisfies $\sign(x) = 1$ if $x \geq 0$ and $\sign(x) = -1$ otherwise.
\vspace{-2mm}
\begin{restatable}{lemma}{lemmaErrorBounds}
Suppose \( T \) and \( T' \) are \( 2 \times 2 \) tensors, where each entry of \( T' \) differs from the corresponding entry of \( T \) by at most \( \xi \),  i.e.,~$|T_{1,1} - T_{1,1}'| \leq \xi$, and so on. 
Denote their coefficients from the decomposition as $\gamma_d$, $\gamma_r$, $\gamma_s$, $\gamma_d'$, $\gamma_r'$, and $\gamma_s'$, respectively. 
Then $|\gamma_d - \gamma_d'| \leq \xi$, $|\gamma_r - \gamma_r'| \leq \xi$, and $|\gamma_s - \gamma_s'| \leq \sqrt{2}\xi$.
\label{lemma:error-bounds}
\end{restatable}
\vspace{-4mm}
\begin{restatable}{lemma}{lemmaSignSwap}
Suppose that $x \in \Rspace$ and $x'$ is a guess for $x$ within an error bound $\xi$, i.e.,~$|x - x'| \leq \xi$. 
\begin{itemize}[noitemsep,leftmargin=*]
\item If $x > 0$ but $x' < 0$. Then $x' + \xi > 0$ and $|x - (x' + \xi)| \leq \xi$. That is, $x' + \xi$ is a valid guess for $x$ that is within an error bound $\xi$ and has the same sign as $x$. 
\item If $x < 0$ but $x' > 0$, then $x' - \xi$ is analogously valid.
\end{itemize}
\label{lemma:sign-swap}
\end{restatable}
\vspace{-2mm}
By applying \cref{lemma:sign-swap}, we can generate a guess \( x' \) for \( x \) that shares the same sign as \( x \). In this context, if \( x' \) and \( x \) already have the same sign, no further adjustment is needed.
\vspace{-2mm}
\begin{restatable}{lemma}{lemmaMagnitudeSwap}
Suppose that $x,y \in \Rspace$ and $x'$ and $y'$ are guesses for $x$ and $y$, respectively, within an error bound $\xi$, such that $|x-x'| \leq \xi$ and $|y - y'| \leq \xi$. Suppose that $|x| > |y|$ but $|x'| < |y'|$. 

Let $x'' = \sign(x')|y'|$, $y'' = \sign(y')|x'|$. Then $|x - x''| \leq \xi$ and $|y - y''| \leq \xi$, and $|x''| > |y''|$. That is, by swapping the magnitudes of $x'$ and $y'$, we can obtain two valid guesses $x''$ and $y''$ with $|x''| > |y''|$.
\label{lemma:magnitude-swap}
\end{restatable}
\vspace{-5mm}

\para{Vertex correction algorithm.}~Let $T$ be a tensor from the ground truth and $T'$ be the corresponding tensor in the intermediate data. Denote their coefficients from the decomposition as $\gamma_d$, $\gamma_r$, $\gamma_s$, $\gamma_d'$, $\gamma_r'$, and $\gamma_s'$ respectively. To ensure that the classifications of $T'$ matches $T$, we must preserve four properties. 
For the eigenvector manifold, we must preserve: (1) whether $|\gamma_r| > \gamma_s$, and (2) the sign of $\gamma_r$. 
For the eigenvalue manifold, we must preserve: (3) which of $|\gamma_r|$, $|\gamma_d|$, or $\gamma_s$ is the largest, and (4) the sign of that coefficient. We also ensure that each of $\gamma_d'$, $\gamma_r'$ and $\gamma_s'$ differs from its ground truth value by at most $\xi$.

For each $T$ represented as a matrix, we check whether these four properties have been preserved in $T'$. 
If not, we correct $T'$. To do so, we introduce five variables, $\SFIX$, $\RSIGN$, $\DSIGN$, $\ROVERS$, and $\DLARGEST$.  
Each variable serves as a flag indicating whether an adjustment is needed for \(\gamma_d'\), \(\gamma_r'\), or \(\gamma_s'\) to ensure the preservation of the four properties. In particular,
\begin{itemize}[noitemsep,leftmargin=*]
\item \(\SFIX\) is enabled when \(\gamma_s'\) needs adjustment to satisfy \( |\gamma_s - \gamma_s'| \leq \xi \). According to~\cref{lemma:error-bounds}, this can be achieved by adding or subtracting \((\sqrt{2}-1)\xi\) from \(\gamma_s'\).  
\item \(\RSIGN\) is enabled when \(\gamma_r'\) requires a sign change. The sign is adjusted using the approach outlined in~\cref{lemma:sign-swap}. 
\item \(\DSIGN\) is enabled when \(\gamma_d'\) requires a sign change, which is performed using the strategy described in~\cref{lemma:sign-swap}.  
\item \(\ROVERS\) is enabled if the ordering of \( |\gamma_r'| \) and \(\gamma_s'\) is incorrect (i.e., one is larger than the other when it should be smaller). Their magnitudes are swapped using the approach outlined in~\cref{lemma:magnitude-swap}. 
\item $\DLARGEST$ is enabled if $|\gamma_d'|$ should have a larger magnitude than $|\gamma_r'|$ and $|\gamma_s'|$, but  does not. We also set $\DLARGEST$ if $|\gamma_d'|$ has the largest magnitude, but it should be smaller than $|\gamma_r'|$ or $|\gamma_s'|$. In such cases, we swap the magnitudes using the strategy in~\cref{lemma:magnitude-swap}.
\end{itemize}
Once we have made appropriate adjustments to $\gamma_d'$, $\gamma_r'$ and $\gamma_s'$, we reassemble $T'$ from the coefficients according to~\cref{eqn:decomposition}. 
Although unlikely, it is possible that after doing so, \( T' \) may not adhere to the element-wise error bound \( \xi \). 
In such a case, we quantize $\gamma_d'$, $\gamma_r'$ and $\gamma_s'$ using logarithmic-scaling quantization with increased precision and repeat; we report specifics on these adjustments in the supplement.

\subsubsection{Asymmetric Cell Correction: Cell Topology Correction} 
\label{sec:method-asymmetric-internal}
  
Preserving the classifications of the vertices is often (but not always) sufficient to preserve the cell topology. In this step, we ensure the topology of each cell $\sigma$ is preserved.

First, we preserve the degeneracies of the symmetric component of the tensor field using the same strategy applied to symmetric tensor fields. However, when the procedure would allow the deviator to be stored losslessly, we instead store the entire matrix losslessly. 

Second, we preserve the topology of each cell $\sigma$ to ensure that the global eigenvector and eigenvalue graphs are maintained. These graphs are computed by determining the topology of each cell individually and then merging regions of the same type in adjacent cells. We preserve both the eigenvector and eigenvalue graphs for $\sigma$, as well as the edges and vertices of $\sigma$ that border each region of the corresponding partitions. This approach guarantees that the merging step in the global eigenvalue and eigenvector graph computation will produce consistent results. 

\para{Theoretical basis.} The following lemmas are relevant to computing the cell topology of $\sigma$, whose proofs are in the supplement.
\vspace{-2.5mm}
\begin{restatable}{lemma}{lemmaEigenvectorTopology}
Except in special cases, the topology of the eigenvector partition of $\sigma$ is determined by the following conditions:
\begin{itemize}[noitemsep]
\item[(a)] Whether each of the curves $\gamma_r = \gamma_s$ and $-\gamma_r = \gamma_s$ intersects the interior of $\sigma$.
\item[(b)] How many times the curves $\gamma_r = \gamma_s$ and $-\gamma_r = \gamma_s$ intersect each edge of $\sigma$.
\item[(c)] The classification of each vertex of $\sigma$ according to the eigenvector partition.
\end{itemize}
\label{lemma:eigenvector-topology}
\end{restatable}
\vspace{-6.5mm}

\begin{restatable}{lemma}{lemmaEigenvalueTopology}
Except in special cases, the topology of the eigenvalue partition of $\sigma$ is determined by the following conditions:
\begin{itemize}[noitemsep]
\item[(a)] The classification of each vertex in the eigenvalue partition.
\item[(b)] For each of the curves \(\gamma_d = \gamma_s\), \(-\gamma_d = \gamma_s\), \(\gamma_r = \gamma_s\), and \(-\gamma_r = \gamma_s\), determine the order in which the following points are encountered when traveling counterclockwise around the curve. 
\begin{itemize}[noitemsep]
\item[(i)] Each boundary point (i.e, a point on the boundary between two regions of the eigenvalue partition) where the curve enters or leaves $\sigma$, and the orientation of the curve at that point.
\item[(ii)] Each junction point, along with the orientations of the curves that intersect at that point. 
\end{itemize}
\item[(c)] For each edge \( e \) of \(\sigma\), for each point identified in (b.i) that lies on \( e \), determine which point is closest to each vertex of \( e \). 
\item[(d)] For each region type in the eigenvalue partition, except where \(\gamma_s > |\gamma_d|\) and \(\gamma_s > |\gamma_r|\), if its boundary curve does not intersect any edges or other curves, check whether that region is present in the eigenvalue partition of \(\sigma\).
\end{itemize}
\label{lemma:eigenvalue-topology}
\end{restatable}
\vspace{-3mm}
Both~\cref{lemma:eigenvector-topology} and~\cref{lemma:eigenvalue-topology} admit special cases proved and discussed in the supplement. 

\para{Cell topology preservation algorithm.}~To accomplish our goal, we create a topological invariant that allows us to verify whether the topology of \(\sigma\) in the intermediate data aligns with the ground truth. This invariant is represented as a data structure containing information about the cell. As part of the invariant, we include the classification of each vertex based on the eigenvector and eigenvalue manifolds. Additionally, we incorporate other relevant information, which we will describe next.

\noindent\underline{Eigenvector manifold invariant.}~By~\cref{lemma:eigenvector-topology}, the topology of the eigenvector partition of $\sigma$ is determined by the boundaries of the regions where $\gamma_r > \gamma_s$ and $-\gamma_r > \gamma_s$. We divide the curve $\gamma_r^2 = \gamma_s^2$ into two connected components: $\gamma_r = \gamma_s$ and $-\gamma_r = \gamma_s$. Next, in our invariant, we track how many times each curve intersects each edge. If the curves do not intersect any edges but intersect the interior of \(\sigma\), we also record this case in the invariant. In \cref{fig:method-intersections}(A), we illustrate one possible example. In this example, we would record that the curve \(\gamma_r = \gamma_s\) intersects the hypotenuse once and the left edge once, while the curve \(-\gamma_r = \gamma_s\) does not intersect the cell.

\begin{figure}[!h]
    \begin{overpic}[width=\linewidth]{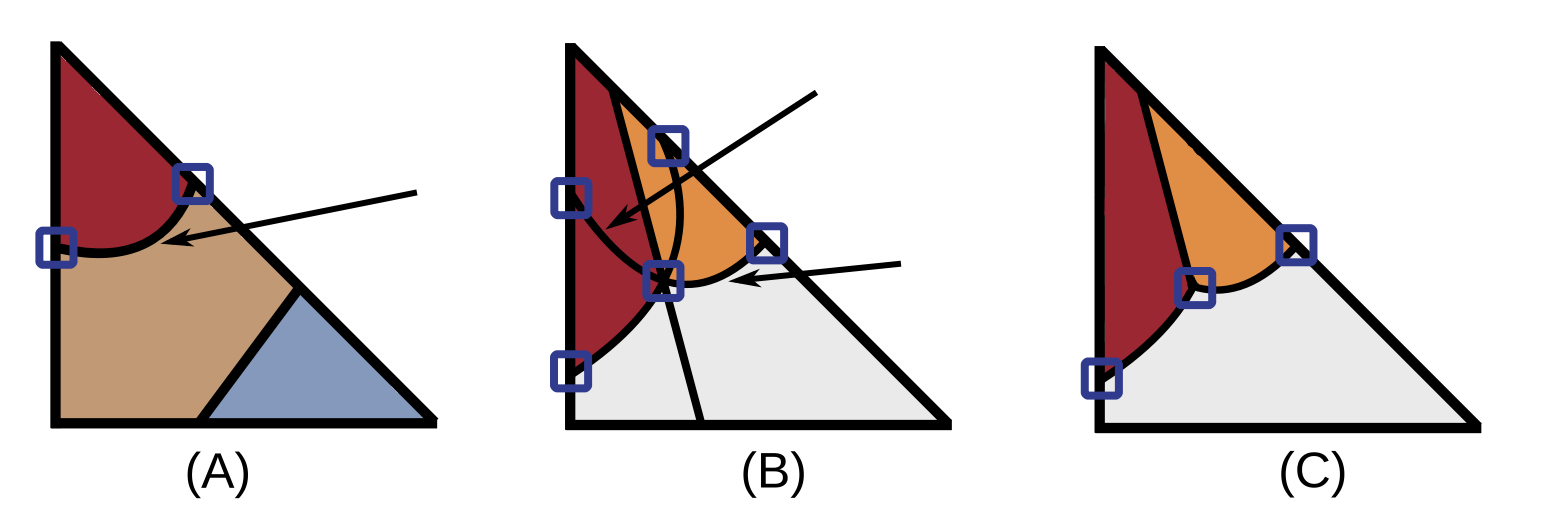}

    \put (22,22) {\smaller[1]{$\gamma_r=\gamma_s$}}
    
    \put (48,28) {\smaller[1]{$\gamma_r=\gamma_s$}}
    \put (53,18) {\smaller[1]{$\gamma_d=\gamma_s$}}

    \put (72,7) {\smaller[1]{$a$}}
    \put (78,12) {\smaller[1]{$b$}}
    \put (82.5,14.5) {\smaller[1]{$c$}}
    
    \end{overpic}
    \vspace{-7mm}
    \caption{(A) We find the edge intersections for the curve $\gamma_r = \gamma_s$. (B) We compute the edge intersections and junction points for the curves $\gamma_d = \gamma_s$ and $\gamma_r = \gamma_s$. (C) We only retain cell intersections that are topologically significant.}
    \label{fig:method-intersections}
    \vspace{-3mm}
\end{figure}

\noindent\underline{Eigenvalue manifold invariant.}~Following~\cref{lemma:eigenvalue-topology}, to describe the cell topology based on the eigenvalue manifold, we need to track where each of \(\gamma_d^2 = \gamma_s^2\) and \(\gamma_r^2 = \gamma_s^2\) intersects each edge of the cell \(\sigma\), as well as where they intersect each other. 
To accomplish this, we divide each curve into two components, resulting in four total curves: \(\gamma_d = \gamma_s\), \(-\gamma_d = \gamma_s\), \(\gamma_r = \gamma_s\), and \(-\gamma_r = \gamma_s\). We then trace each curve clockwise, recording the order in which it intersects an edge or another curve, and noting the specific edge or curve it intersects in each case, and the orientation. If an intersection does not affect the topology, we omit it from our tracking. This process produces four lists of intersections, which we incorporate into our invariant. 

In \cref{fig:method-intersections}(B), we provide an example showing the intersection points we compute, which \new{are} four edge intersections and one junction point. However, two of these intersections do not affect the topology, so we exclude them from our tracking. In \cref{fig:method-intersections}(C), we show the three intersections that we do track. Ultimately, our invariant records that the curve \(\gamma_d = \gamma_s\) starts by intersecting the hypotenuse at \(c\) and ends at a junction point at \(b\). The curve \(\gamma_r = \gamma_s\) starts at a junction point at \(b\) and ends by intersecting the left edge at \(a\). The curves \(-\gamma_d = \gamma_s\) and \(-\gamma_r = \gamma_s\) do not contribute to the topology. In our invariant, we do not track the specific locations \(a\), \(b\), and \(c\); instead, we track whether the intersection is a junction point, and if not, which edge it lies on.

\noindent\underline{Topology correction.}~We compute our invariant for $\sigma$ for both the ground truth and reconstructed data. If there is a mismatch, then we procedurally quantize $\theta$ using linear-scaling quantization and lower the error bounds of $\gamma_d$, $\gamma_r$, and $\gamma_s$ using logarithmic-scaling quantization.  We provide details in the supplement. \new{After the degenerate points are preserved, it is possible that one will lie in the wrong region type in the eigenvector partition. In such a case, we use the same adjustment process used to correct degenerate point errors.}

\para{Encoding of adjustments.} Finally, once the correction for each cell \(\sigma\) is complete, we generate several lists of integers. Since many of our variables require fewer than 8 bits, we combine multiple variables into a single byte. For each matrix \(T\), we combine \(\SFIX\) with the quantization number for \(\theta\) to create an 8-bit integer. We also group \(\RSIGN\), \(\DSIGN\), \(\ROVERS\), and \(\DLARGEST\) into another 8-bit integer. These integers, along with the quantization numbers for \(\gamma_d\), \(\gamma_r\), and \(\gamma_s\), are encoded using Huffman coding and included in the final compressed file.

%% file: sec-results.tex
\section{Experimental Results}
\label{sec:results}

We provide an overview of our experiments in \cref{sec:results-overview}. We evaluate the performance of {\toolname} with a comparison between augmented and base compressors in~\cref{sec:results-evaluation}. We include an analysis of run time in~\cref{sec:results-time}. We aim to preserve both eigenvalue and eigenvector graphs; preserving only one of them is included in the supplement, \new{where we also provide error maps, as well as statistics on the number of times that cells are visited and on the topological errors made by the base compressors.}

\para{Highlighted results.}~We highlight our experimental results below.
\begin{itemize}[noitemsep,leftmargin=*]
\item Applying SZ3 or SPERR to any of our eight datasets results in numerous topological errors, whereas augmenting them with {\toolname} eliminates topological errors in every case. 
\item Augmented SPERR produces a better tradeoff between bit-rate and \new{Peak Signal-To-Noise Ratio} (PSNR) compared to augmented SZ3.
\item {\toolname} \new{takes} \( O(nk) \) time, where \( n \) \new{is} the number of cells and \( k \) \new{is} the maximum number of times a cell is processed before being stored losslessly. In practice, most of the run time is spent on cell correction. 
\end{itemize}

\subsection{Overview of Experiments}
\label{sec:results-overview}

We present an experimental study of our framework, {\toolname}, by augmenting two state-of-the-art error-bounded lossy compressors for 2D scientific data: SZ3~\cite{liang2022sz3} and SPERR~\cite{li2023lossy}. We also experimented with ZFP~\cite{lindstrom2014fixed}, TTHRESH~\cite{ballester2019tthresh}, and Neurcomp~\cite{lu2021compressive}. However, augmented ZFP~\cite{lindstrom2014fixed} yielded poor compression ratios, while TTHRESH and Neurcomp did not natively support 2D data; therefore these compressors were excluded from our analysis.  

We tested {\toolname} on four symmetric and four asymmetric tensor fields (\cref{tab:results-datasets}). 
Each contains a number of 2D slices, and we compress each slice individually and report evaluation metrics that are aggregated across all slices. The Stress datasets are stress tensor fields. The Brain datasets are from brain MRI scans. The Ocean data is from an ocean flow simulation, the Miranda data is from a turbulence simulation, and the Vortex Street and the Heated Cylinder data come from \new{fluid dynamics} simulations. See the supplement for details on the datasets. 

\begin{figure*}[!t]
\begin{center}
\includegraphics[width=\textwidth]{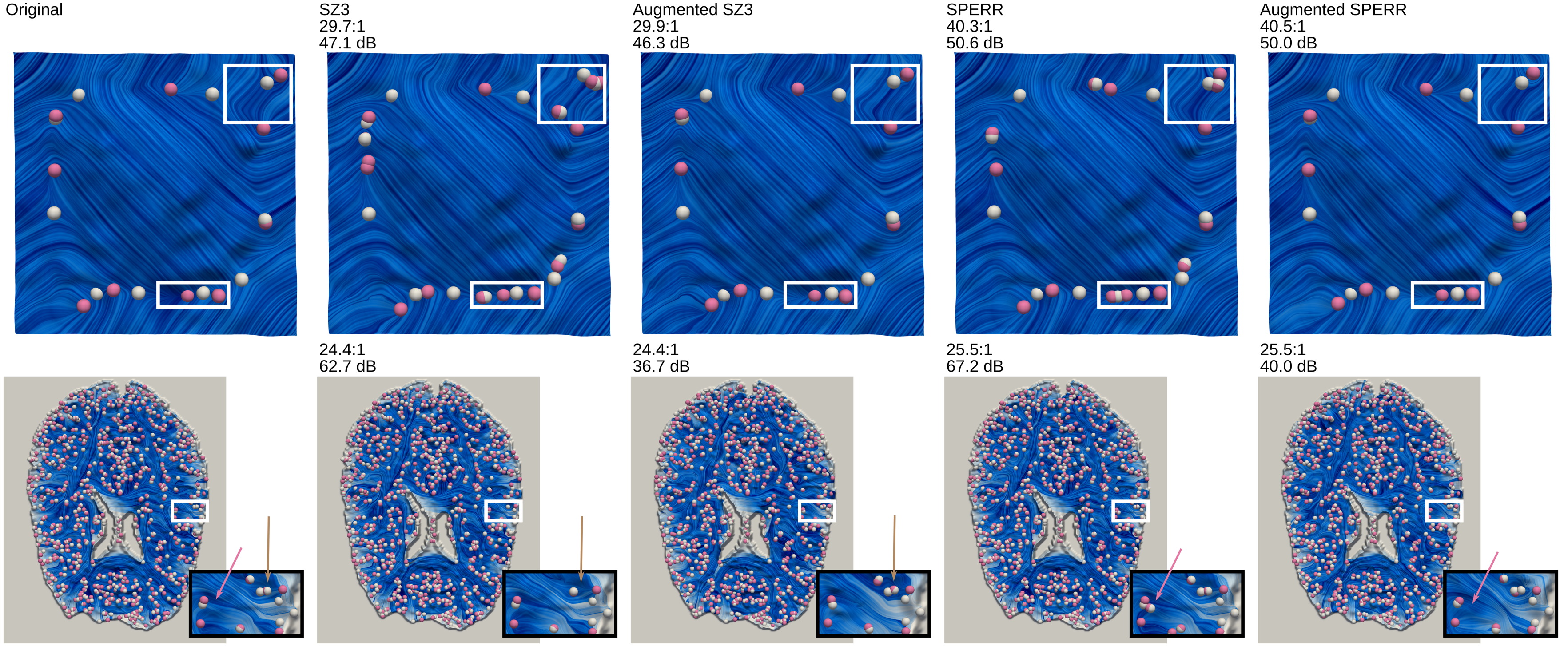}
\end{center}
\vspace{-7mm}
\caption{LIC visualization of the eigenvector fields of two 2D symmetric second-order tensor fields compressed with SZ3, augmented SZ3, SPERR, and augmented SPERR, along with the ground truth. Trisectors are in white, wedges are in pink. Top: Stress B data slice 10. Bottom: {\BrainB} data slice 50. In the top row, we highlight a region of interest. In the bottom row, we highlight and provide a zoomed-in view of a region of interest. We also provide an orange arrow highlighting a discrepancy from SZ3, and a pink arrow highlighting a discrepancy of SPERR. \new{The Z position of each point corresponds to the Frobenius norm with smoothing applied.}}
\vspace{-4mm}
\label{fig:results-syms}
\end{figure*}

\begin{figure*}[!h]
\begin{center}
\includegraphics[width=\textwidth]{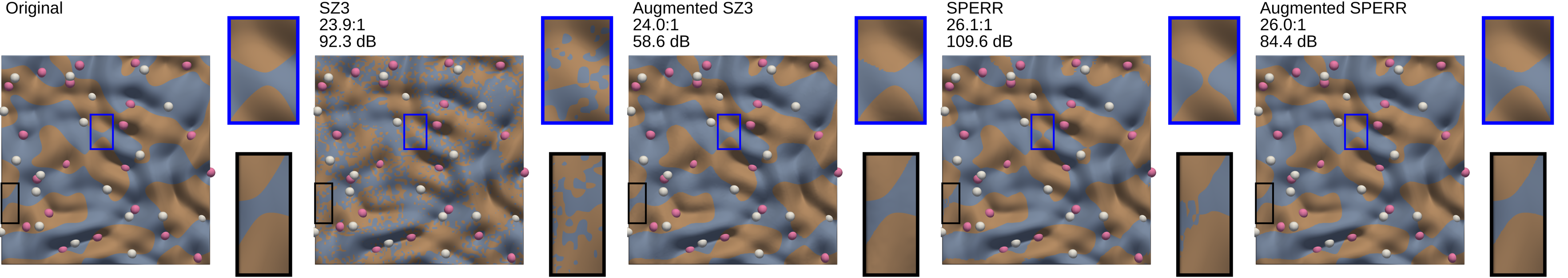}
\end{center}
\vspace{-7mm}
\caption{Visualizing the eigenvector partition of the Miranda dataset slice 30 compressed with SZ3, augmented SZ3, SPER, and augmented SPERR, along with the ground truth. We provide zoomed-in views that highlight the differences between the compressors and the ground truth. We also label compression ratio and PSNR. We use the same colormap as \cref{fig:manifolds}. \new{Z positions correspond to the Frobenius norm with smoothing applied.}}
\vspace{-7mm}
\label{fig:results-asym2}
\end{figure*}

When running a base compressor (SZ3 or SPERR) on each slice, we compress each of the four fields separately, similar to {\toolname}. 
We then combine the four compressed outputs into one tar archive, and compress it losslessly with ZSTD. 
For the Brain datasets, empty regions outside the brain are set to zero. 
In those regions, small perturbations introduced by lossy compressors may lead to many topological errors. When testing the lossy compressors on the brain datasets, we store all zero tensors losslessly, introducing minimal overhead to the compressed file while eliminating the topological errors outside of the brain. 
We denote the choice of a dataset, base compressor, and error bound as a trial.
We run each trial on a personal computer with an Intel Core i9-12900HK processor and 32GB of RAM. Our  implementation is written in Julia version 1.10.2, and does not include any parallelization.

\begin{table}[!ht]
\vspace{-3mm}
\begin{center}
\caption{Scientific datasets used for compression analysis. ``Sym.'' means symmetric and ``Asym.'' means asymmetric dataset respectively. ``Slice Dim." means the dimension of each 2D slice.}
\label{tab:results-datasets}
\vspace{-3mm}
\begin{tabular}{ccccc}
\hline
\textbf{Dataset} & \textbf{Type} & \textbf{Slice Dim.} & \textbf{\#Slices} & \textbf{Size (MB)} \\ \hline
Stress A         & Sym.           & $65 \times 65$     & $25$            & $3.4$              \\
Stress B         & Sym.           & $65 \times 65$     & $25$            & $3.4$              \\
Brain A         & Sym.           & $66 \times 108$    & $76$           & $17.3$             \\
Brain B          & Sym.           & $148 \times 190$   & $157$           & $141.3$            \\ \hline
Ocean            & Asym.          & $101 \times 101$   & $27$            & $8.8$              \\
Miranda          & Asym.          & $384 \times 384$   & $256$           & $1208.0$           \\
{\VS}    & Asym.          & $640 \times 80$    & $1501$          & $2459.2$           \\
{\HC}  & Asym.          & $150 \times 450$   & $2000$          & $4320.0$           \\ \hline
\end{tabular}
\end{center}
\vspace{-7mm}
\end{table}

\para{Evaluation metrics.}~For each symmetric tensor field, we evaluate how many cells have their degeneracy type (trisector, wedge, non-degenerate etc.) preserved compared to the ground truth. For each asymmetric tensor field, we evaluate how many cells have their eigenvector and eigenvalue partitions preserved. We check the topology using the invariant described in~\cref{sec:method-asymmetric-internal}. For each trial, we also study the tradeoff between compressed size and reconstruction quality, by measuring the standard compression metrics such as compression ratio, bit-rate and PSNR. Recall the compression ratio is the original file size divided by the compressed file size. Bit-rate is the average number of bits used to encode each tensor. PSNR measures the reconstruction quality (higher values are better). Finally, we quantify total compression and augmentation time, and decompression time.

\para{Parameter configurations.}~We define the range of a tensor field as the largest entry in any tensor minus the smallest entry in any tensor. We vary the pointwise error bound $\xi$, which represents the error bound as a percentage of the range. 
For example, $\xi = 0.01$ means that the global error bound is $1\%$ of the range. We recompute $\xi$ for each slice. We set the pointwise error bound for SZ3 and SPERR to $\xi$.

\subsection{Performance Evaluation}
\label{sec:results-evaluation}

\begin{figure}[!h]
    \includegraphics[width=\linewidth]{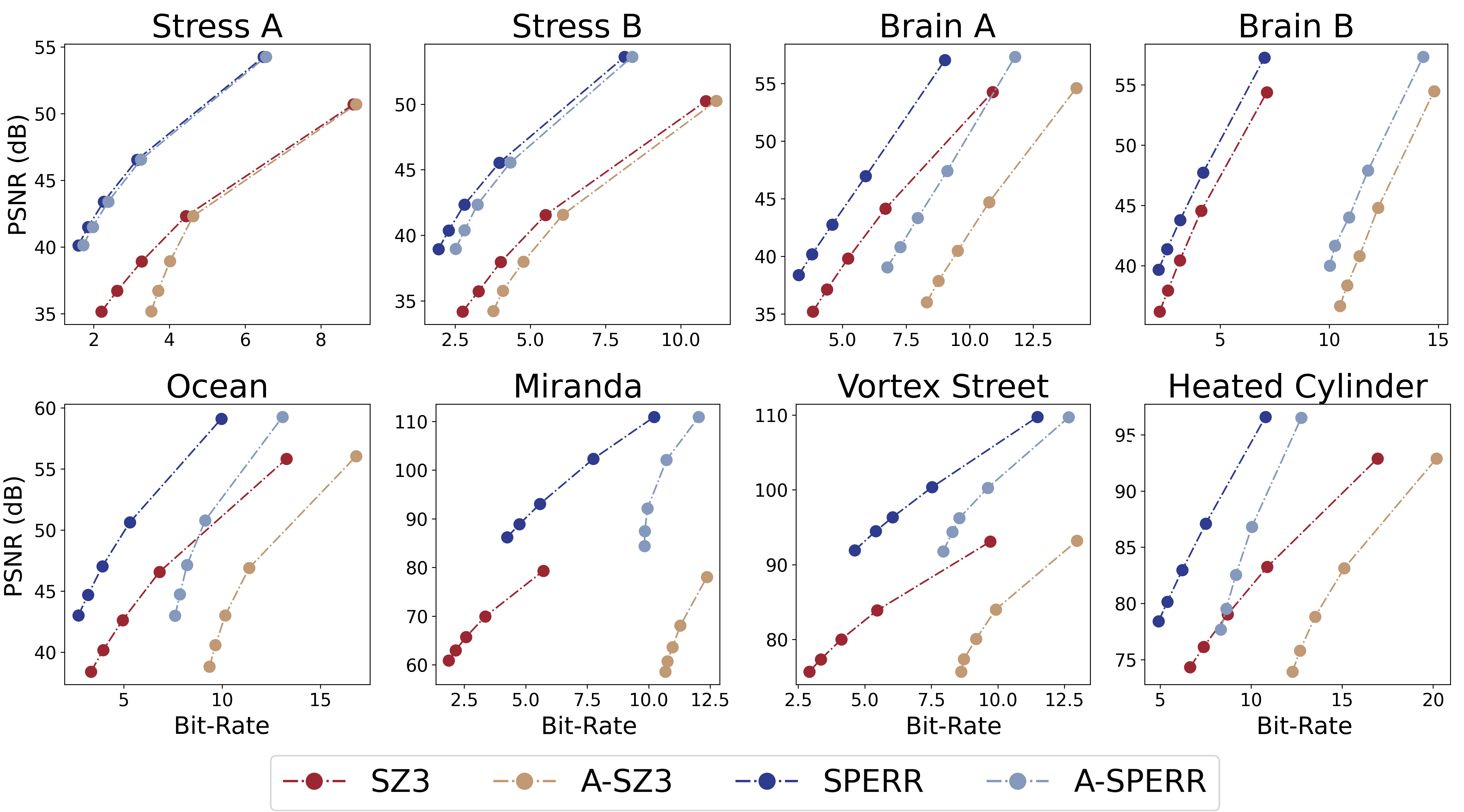}
    \vspace{-6mm}
    \caption{Curves showing the tradeoff between bit-rate and PSNR for SZ3 and SPERR on each dataset, as well as our augmented compressors given as A-SZ3 and A-SPERR.}
    \label{fig:results-rateDistortion}
    \vspace{-7mm}
\end{figure}

\begin{table*}
\begin{center}
\caption{Number of cells whose topology is misclassified by SZ3 and SPERR when compressed to the same ratio as augmented SZ3 and augmented SPERR, respectively. For symmetric data, we set the error bound of {\toolname} to be $\xi = 0.01$ and for asymmetric data we set $\xi = 0.001$. \#Cells: number of cells in the mesh. \#SZ3: number of cells misclassified by SZ3. \%SZ3: percent of cells misclassified by SZ3. \#SPERR: number of cells misclassified by SPERR. \%SPERR: percent of cells misclassified by SPERR.}
\vspace{-3mm}
\begin{tabular}{cc|ccc|ccc}
\hline
Dataset & \#Cells & \#SZ3 & \%SZ3 & \#A-SZ3 & \#SPERR & \%SPERR & \#A-SPERR \\ \hline
{\StressA} & 204,800 & 85 & 0.04\% & 0 & 47 & 0.02\% & 0  \\
{\StressB} & 204,800 & 582 & 0.28\% & 0 & 381 & 0.19\% & 0  \\
{\BrainA} & 1,057,160 & 2,144 & 0.20\% & 0 & 1,660 & 0.16\% & 0  \\
{\BrainB} & 8,723,862 & 16,831 & 0.19\% & 0 & 6,609 & 0.08\% & 0  \\
{\Ocean} & 540,000 & 17,185 & 3.2\% & 0 & 13,536 & 2.5\% & 0  \\
{\Miranda} & 75,104,768 & 16,359,113 & 21.8\% & 0 & 2,464,647 & 3.3\% & 0  \\
{\VS} & 151,543,962 & 32,160,009 & 21.2\% & 0 & 21,881,208 & 14.4\% & 0  \\
{\HC} & 267,604,000 & 50,817,728 & 19.0\% & 0 & 32,859,532 & 12.3\% & 0  \\\hline
\end{tabular}
\end{center}
\vspace{-9mm}
    \label{tab:results-equivEB}
\end{table*}

\para{Topological guarantees.}~We empirically verify that {\toolname} perfectly preserves the topology of each cell in each trial. By contrast, when we run the base compressors SZ3 and SPERR with very low error bounds, such that the resulting data has a comparable compressed size to {\toolname}, we still notice a large number of topological errors. 

In~\cref{tab:results-equivEB}, we list the number and percentage of cells topologically misclassified by SZ3 and SPERR. We chose error bounds for SZ3 and SPERR so that they achieve compression ratios similar to those from {\toolname} with $\xi=0.01$ for symmetric data and $\xi = 0.001$ for asymmetric data. We provide these bounds in the supplement to ensure reproducibility. We observe that the base compressors exhibit significant topological errors when applied to asymmetric data. While the percentage (error rate) is relatively low for a symmetric dataset, the percentage of cells that contain a degenerate point is typically very small. Thus, having a relatively small percentage of incorrectly predicted cells can produce a significant effect on the topology of the decompressed data.

In~\cref{fig:results-syms}, we visualize the {\StressB} and {\BrainB} datasets before and after compression with SZ3, augmented SZ3, SPERR, and augmented SPERR. We chose parameter configurations such that the compression ratio achieved by both compressors is the same (see the supplement for details). In each visualization, we highlight a region of interest, and provide a zoomed-in view when needed. We display the 
{\Miranda} dataset in the same fashion in~\cref{fig:results-asym2}, and provide zoomed-in views of the eigenvector and eigenvalue partitions. We also visualize the {\Ocean} dataset before and after compression with SZ3 and augmented SZ3 in~\cref{fig:teaser}. \new{In} the figures, we can see that SZ3 and SPERR make noticeable topological errors, while {\toolname} perfectly preserves the topology.

\para{Compressed file size.}~In \cref{fig:results-rateDistortion}, we visualize the tradeoff between bit-rate and PSNR for both the augmented and base compressors. Consistent with previous works on topology-preserving compression~\cite{gorski2025general,li2024msz}, we observe that when $\xi$ is high, increasing $\xi$ can lead to lower compression ratios. When evaluating the tradeoff between bit-rate and PSNR, we use low values of $\xi$ such that increasing $\xi$ increases compression ratios. The values for $\xi$ are different for each dataset and base compressor; we report them in the supplement for reproducibility.

In~\cref{fig:results-rateDistortion}, we can see that the tradeoff between bit-rate and PSNR for the augmented compressors mirrors that of the base compressors with some amount of overhead. Overall, augmented SPERR seems to achieve \new{better reconstruction quality} than augmented SZ3.

\subsection{Run Time Analysis}
\label{sec:results-time}

\para{Compression and decompression times.}~{\toolname} processes each cell independently and limits the number of processing iterations before storing it losslessly. Let $n$ be the number of cells and $k$ the number of iterations, {\toolname} runs in linear time $O(nk)$. Recall $k = 20$.

Empirically, we measure the compression time for each dataset using both base and augmented compressors; see the results in~\cref{tab:time-all}. For asymmetric data, we report the compression time while preserving both eigenvector and eigenvalue partitions, while cases where only one is preserved are provided in the supplement.

We find that augmented SPERR \new{generally} achieves slightly worse compression times than augmented SZ3. Varying the error bound $\xi$ does not have a significant effect on the run times; see the supplement.

\para{Analysis of compression times.}~In~\cref{tab:time-breakdown}, we provide a more detailed breakdown of the run times reported in \cref{tab:time-all} for augmented SZ3 and augmented SPERR. As expected, the cell correction step (labeled as “Cells”) accounts for the largest portion of the run time, typically dominating the overall processing time. For symmetric data, this step is dedicated entirely to preserving degenerate points. In contrast, for asymmetric data, it involves correcting vertices, degenerate points, and cell topology. Among these, cell topology correction is the most time-consuming, though all three contribute significantly to the total run time. We provide more detail in the supplement.

\begin{table}
\begin{center}
\caption{Compression and decompression times for each dataset using base and augmented compressors. A-SZ3 and A-SPERR denote augmented SZ3 and augmented SPERR respectively.}
\label{tab:time-all}
\vspace{-3mm}
\new{\begin{tabular}{c|cc|cc}
\hline
Dataset & SZ3 & A-SZ3 & SPERR & A-SPERR \\ \hline
\multicolumn{5}{c}{Total Compression Time (s)} \\ \hline
Stress A & 0.22 & 0.56 & 0.22 & 0.67 \\
Stress B & 0.22 & 0.57 & 0.22 & 0.66 \\
Brain A & 0.76 & 1.92 & 0.77 & 2.38 \\
Brain B & 2.11 & 9.17 & 2.22 & 10.93 \\ \hline
Ocean & 0.35 & 1.90 & 0.39 & 2.28 \\
Miranda & 7.00 & 274.14 & 9.53 & 287.16 \\
Vortex Street & 24.34 & 784.85 & 28.86 & 835.99 \\
Heated Cylinder & 37.98 & 762.21 & 48.23 & 713.11 \\ \hline
\multicolumn{5}{c}{Decompression Time (s)} \\ \hline
Stress A & 0.28 & 0.25 & 0.28 & 0.25 \\
Stress B & 0.28 & 0.25 & 0.28 & 0.26 \\
Brain A & 0.82 & 0.86 & 0.81 & 0.81 \\
Brain B & 2.15 & 2.31 & 2.08 & 2.47 \\ \hline
Ocean & 0.34 & 0.38 & 0.36 & 0.42 \\
Miranda & 4.73 & 11.44 & 6.34 & 12.60 \\
Vortex Street & 19.78 & 33.76 & 22.81 & 35.13 \\
Heated Cylinder & 29.01 & 50.78 & 36.68 & 53.44 \\ \hline
\end{tabular}}
\end{center}
\vspace{-10mm}
\end{table}

\begin{table}
\begin{center}
\caption{Breakdown of compression times for each dataset when compressed with augmented SZ3 and augmented SPERR. Times are in seconds. All trials use an error bound of $\xi = 0.001$. BC: Base compressor. Cells: Cell correction step. LL: Lossless storage. Save: Save compressed file. Clean: Remove intermediate files.}
\label{tab:time-breakdown}
\vspace{-3mm}
\resizebox{\linewidth}{!}{
\new{\begin{tabular}{c|cccccc|c}
\hline
Dataset & Setup & BC & Cells & LL & Save & Clean & Total \\ \hline
\multicolumn{8}{c}{Augmented SZ3} \\ \hline
Stress A & 0.01 & 0.20 & 0.13 & 0.03 & 0.06 & 0.13 & 0.56 \\
Stress B & 0.02 & 0.20 & 0.14 & 0.03 & 0.06 & 0.12 & 0.57 \\
Brain A & 0.05 & 0.48 & 0.69 & 0.08 & 0.22 & 0.39 & 1.92 \\
Brain B & 0.34 & 1.34 & 5.42 & 0.51 & 0.71 & 0.85 & 9.17 \\
Ocean & 0.13 & 0.32 & 1.13 & 0.06 & 0.09 & 0.18 & 1.90 \\
Miranda & 10.64 & 7.02 & 245.17 & 6.90 & 1.75 & 1.94 & 274.14 \\
Vortex Street & 21.05 & 19.86 & 708.43 & 15.73 & 7.99 & 10.39 & 784.85 \\
Heated Cylinder & 31.67 & 33.11 & 644.57 & 24.78 & 11.28 & 14.32 & 762.21 \\ \hline 
\multicolumn{8}{c}{Augmented SPERR} \\ \hline
Stress A & 0.02 & 0.30 & 0.13 & 0.03 & 0.06 & 0.13 & 0.67 \\
Stress B & 0.02 & 0.30 & 0.13 & 0.03 & 0.06 & 0.13 & 0.66 \\
Brain A & 0.05 & 0.82 & 0.80 & 0.08 & 0.22 & 0.41 & 2.38 \\
Brain B & 0.34 & 2.45 & 6.09 & 0.49 & 0.71 & 0.86 & 10.93 \\
Ocean & 0.08 & 0.57 & 1.27 & 0.07 & 0.09 & 0.18 & 2.28 \\
Miranda & 11.61 & 13.82 & 250.47 & 6.86 & 1.73 & 1.95 & 287.16 \\
Vortex Street & 21.27 & 37.22 & 744.16 & 13.64 & 7.78 & 10.54 & 835.99 \\
Heated Cylinder & 31.96 & 63.80 & 568.57 & 21.67 & 10.19 & 14.41 & 713.11 \\ \hline
\end{tabular}}
}
\end{center}
\vspace{-10mm}
\end{table}

%% file: sec-limitations.tex
\section{\new{Limitations and Discussion}}
\label{sec:limitations}

Our {\toolname} framework has its limitations. By enhancing  a base compressor and generating additional output files, the augmented compressor typically achieves lower compression ratios and requires more time compared to the base compressor. We highlight the key takeaways here and report detailed analysis in the supplement.

\para{Storage overhead.}~In \cref{fig:results-rateDistortion}, we can see that {\toolname} imposes up to $1\times$ storage overhead for symmetric data and $3\times$ for asymmetric data (\new{at a fixed PSNR}). These values are similar to other topology-preserving compressors at similar compression ratios~\cite{yan2023toposz, li2024msz}. For asymmetric data, preserving only one of the eigenvector or eigenvalue partitions will lower the overhead. \new{For a fixed compression ratio, we argue improved topological correctness can be more important than increased PSNR, depending on the application. We believe that there is room to improve the compression ratios by developing targeted strategies to preserve the internal cell topology of asymmetric tensor fields, as well as more efficient multiscale quantization schemes.}

\para{Time overhead.}~In \cref{tab:time-all}, we observe that {\toolname} introduces a significant compression time overhead compared to the base compressors. {\toolname} typically achieves a throughput of 3–6 MB/second. \new{We also timed {\toolname} compressing datasets with larger 2D slices and found similar throughputs. The results are in the supplement. Our throughputs} are comparable to other topology-preserving methods~\cite{gorski2025general,li2024msz}, \new{but} a direct comparison is challenging due to variations in processor speed. \new{Also, scientific datasets are often compressed once, and the compressed file is distributed many times, so this overhead may represent a one-time cost.} We report performance on a single CPU. {\toolname} could potentially benefit from parallel or GPU implementations. Since {\toolname} processes each cell independently, parallelization would be relatively straightforward.

\para{\new{Visual artifacts.}} \new{Our strategy can produce visual artifacts, particularly in the asymmetric case. Such artifacts are most noticeable in areas where the magnitude of tensors is low compared to the absolute error bound. We discuss the causes of such artifacts in detail in the supplement. Managing these artifacts is a possible area of future work.}


\para{Other limitations.} For symmetric tensor fields, {\toolname} makes no guarantees about the tensorlines that connect the degenerate points. The tensorlines are part of the definition of the topology of a symmetric tensor field, and one logical extension of our work is to preserve the tensorlines.  \new{However, our strategy still benefits applications that use the entire topology by preserving the degenerate points. Further, there are use cases that do not require tensorlines \cite{lavin1997singularities, zhang2017applying, Filho2016Automatic,tricoche2001tensor}.}

%% file: sec-conclusion.tex
\section{Conclusion}
\label{sec:conclusion}

We introduce a novel framework, {\toolname}, designed to augment any \new{error-bounded} lossy compressor in order to preserve the topology of 2D symmetric and asymmetric tensor fields. In both cases, we scan through each cell and correct the topology one cell at a time. Our experiments show that {\toolname} preserves the degenerate points of symmetric data and the eigenvalue and eigenvector partitions of asymmetric data, while introducing a reasonable overhead. Looking ahead, a common asymmetric tensor field is the gradient of a vector field, and the decomposition in \cref{eqn:decomposition} is used to visualize vector fields \cite{zhang2008asymmetric, auer2013automatic}. {\toolname} could be applied to vector fields to preserve the topology of the gradient.

%% file: appendix-datasets.tex
\section{An Overview of Datasets}
\label{appendix:datasets}

We provide details on the datasets used in our experiments. 
Each dataset contains a set of 2D slices, obtained by slicing a corresponding 3D dataset along the $z$-axis and discarding information associated with the $z$ direction. 
Any slices where the tensor field was entirely zero were removed.

The \textbf{Stress A} and \textbf{Stress B} datasets come from the public dataset associated with the work by Patel and Laidlaw~\cite{stressArticle}. 
The datasets are publicly accessible from the IEEE DataPort~\cite{stressData}; Stress A dataset corresponds to the first, whereas Stress B dataset corresponds to the third dataset in the DataPort.  

The \textbf{Brain A} dataset comes from the public dataset of brain MRI scans from Tian et al.~\cite{tian2022comprehensive}. We use the data from patient 23. We then employ Diffusion Imaging in Python (DIPY) to extract the diffusion tensor field from the data.

The \textbf{Brain B} dataset is a scan of Dr. Gordan Kindlmann's brain that he released for public use. The dataset can be accessed from his personal homepage~\cite{KindlemannBrain}. 
The dataset also comes with a mask indicating which areas are part of the brain. 
We set the tensors equal to zero in all areas outside of the brain (i.e. where the mask is  not equal to 1).

The asymmetric tensor fields are all derived from 2D flow fields $v:\R^2 \rightarrow \R^2$. We compute the gradient at each point of a flow field, producing a 2D matrix.
Now we describe the computation for $\frac{\partial v_x}{\partial x}$, which proceeds analogously for the other partial derivatives. Each flow field is defined on a regular grid $[n] \times [m]$ for some $m,n \in \N$. We describe the computation for a point $p = (p_x,p_y)$ at one of the grid points. The gradients are computed differently based on its location in the interior or at the boundary:
\begin{itemize}[noitemsep]
\item \underline{$1 < p_x < n$:} $\left.\frac{\partial v_x}{\partial x}\right|_p \gets v_x(p_x+1,p_y) - v_x(p_x-1,p_y)$.
\item \underline{$p_x = 1$:} $\left. \frac{\partial v_x}{\partial x}\right|_p \gets v_x(2,p_y) - v_x(1,p_y)$.
\item \underline{$p_x = n$:} $\left. \frac{\partial v_x}{\partial x} \right|_p \gets v_x(n,p_y) - v_x(n-1,p_y)$.
\end{itemize}

The \textbf{Ocean} dataset comes from the Global Ocean Physics Reanalysis dataset from the E.U. Copernicus Marine Service \cite{oceanData}. 
For the flow field, we use the $uo$ and $vo$ fields from the daily data (file name: ``cmems\_mod\_glo\_phy\_my\_0.083deg\_P1D-m''), which is then sliced in the range $x:$ 100-200, $y:$ 10-110 $z:$ 0-26, where $z$ corresponds to depth (all ranges are inclusive). We use the data on June 2, 2019.

The \textbf{Miranda} dataset comes from the hydrodynamics code for large turbulence simulations conducted by Lawrence Livermore National Laboratory. We use the $x$ velocity and $y$ velocity fields to derive the flow field. It can be accessed through the SDR Bench \cite{zhao2020sdrbench, SDRBench}.

The \textbf{Vortex Street} dataset comes from the Cylinder Flow with von Karman Vortex Street simulation \cite{gerrisflowsolver, Guenther17}. The \textbf{Heated Cylinder} dataset comes from the Heated Cylinder with Boussinesq Approximation simulation \cite{gerrisflowsolver, Guenther17}. Both datasets are accessed through the Computer Graphics Laboratory at ETH Zurich \cite{ETHCGLData}.

For each dataset, we eliminate any slices where all data points had the same value.

%% file: appendix-evaluation-metrics.tex
\section{Evaluation Metrics}
\label{appendix:evaluation-metrics}

We now discuss evaluation metrics. 
During compression, we compress each slice of each dataset separately, but report evaluation metrics aggregated across all datasets.

The \emph{compression ratio} is the size of the ground truth data file divided by the size of the compressed file. To report the compression ratio for a dataset, we report the total size of every ground truth data slice divided by the total size of every compressed slice.

The \emph{bit-rate} is the average number of bits used to encode each tensor. For each slice, we compute the bit-rate by dividing the compressed file size (in bits) by the number of tensors. We compute the bit-rate for a dataset by taking the average bit-rate across all slices.

\emph{Peak-Signal to Noise Ratio (PSNR)} is typically defined for the compression of a single scalar field. If $MSE$ is the mean squared error, and $R$ is the range of the data, then PSNR for a scalar field is defined as
\[ \text{PSNR} = 10\log_{10}\left(\frac{R^2}{\text{MSE}}\right).\]
Our PSNR calculation for tensor fields is a bit more complex. 
Let $D$ be a dataset containing a set of slices. 
If $s \in D$ is a slice, let $f_s:\R^2 \rightarrow \T$ be the ground truth tensor field and $f'_s:\R^2 \rightarrow \T$ the reconstructed tensor field. 
Let $v(s)$ be the set of vertices of each slice $s$. To compute the PSNR, we first compute two values for each slice $s$: a mean squared error $\text{MSE}_s$, and a range $R_s$. For asymmetric tensor fields, we define $\text{MSE}_s$ by
\[ \text{MSE}_s = \frac{1}{4|v(s)|} \sum_{x \in v(s)} \sum_{(i,j) \in [1,2]^2} (f(x)-f'(x))_{i,j}^2.\] 
For a symmetric tensor field, we define $\text{MSE}_s$ similarly, except that we ignore the $T_{21}$ entry of each tensor $T$ in our computation, because it is equal to the $T_{12}$ entry. Corresponding to this omission, we divide by $3|v(s)|$ instead of $4|v(s)|$. 

For a slice $s$, the range $R_s$ is defined as the maximum entry across all tensors in $s$ minus the minimum entry across all tensors of $s$. More formally, we define it as
\[ R_s = \max_{\substack{x \in v(s) \\ (i,j) \in [1,2]^2}}\{ f(x)_{i,j} \} - \min_{\substack{x \in v(s)\\ (i,j) \in [1,2]^2}}\{ f(x)_{i,j} \}. \]
Let $N$ be the number of slices. We then define the PSNR as:
\[ \text{PSNR} = 10\log_{10}\left( \frac{1}{N}\sum_{s \in D}  \frac{R_s^2}{\text{MSE}_s} \right).\]
We define the PSNR in this way to account for the fact that certain slices may have different ranges.

%% file: appendix-parameter-configurations.tex
\section{Parameter Configurations}
\label{appendix:parameter-configurations}

We include the parameter configurations for our experiments. 
In \cref{tab:rateDistortionParameters}, we include the error bounds used to generate the plots in \cref{fig:results-rateDistortion} and \cref{fig:additional-experiments-rateDistortion}. We use the same error bounds for the base and augmented compressors.

In \cref{tab:equivEBParameters}, we include the error bounds used to generate \cref{tab:results-equivEB}. 
Each row represents a pair of trials, one using a base compressor, and the other using {\toolname}. 
For each row, we run the base compressor (listed under `BC'), as well as {\toolname} augmenting that base compressor on the given dataset. 
For the base compressor, we use the error bound listed under `BC-EB'. 
For {\toolname}, we use the error bound listed under `A-EB'.

We include a similar table, \cref{tab:equivEBParametersFigures}, describing the error bounds used to generate \cref{fig:teaser}, \cref{fig:results-syms}, \cref{fig:results-asym2}, \cref{fig:extra-sym}, \cref{fig:extra-asym}, and \cref{fig:teaser-sperr}.

\begin{table}
\caption{Parameter configurations used to generate \cref{fig:results-rateDistortion} and \cref{fig:additional-experiments-rateDistortion}. BC: Base Compressor.}
\label{tab:rateDistortionParameters}
\vspace{-6mm}
\begin{center}
\begin{tabular}{c|c|p{0.5\linewidth}} \hline
Dataset & BC & Error Bounds ($\xi$) \\ \hline
Stress A & SZ3 & 0.006, 0.0195, 0.033, 0.0465, 0.06 \\ \hline
Stress A & SPERR & 0.006, 0.0195, 0.033, 0.0465, 0.06 \\ \hline
Stress B & SZ3 & 0.006, 0.0195, 0.033, 0.0465, 0.06 \\ \hline
Stress B & SPERR & 0.006, 0.0195, 0.033, 0.0465, 0.06 \\ \hline
Brain A & SZ3 & 0.006, 0.0195, 0.033, 0.0465, 0.06 \\ \hline
Brain A & SPERR & 0.006, 0.0195, 0.033, 0.0465, 0.06 \\ \hline
Brain B & SZ3 & 0.006, 0.0195, 0.033, 0.0465, 0.06 \\ \hline
Brain B & SPERR & 0.006, 0.0195, 0.033, 0.0465, 0.06 \\ \hline
Ocean & SZ3 & 0.003, 0.00975, 0.0165, 0.02325, 0.03 \\ \hline
Ocean & SPERR & 0.003, 0.00975, 0.0165, 0.02325, 0.03 \\ \hline
Miranda & SZ3 & 0.0003, 0.000975, 0.00165, 0.002325, 0.003 \\ \hline
Miranda & SPERR & $1 \times 10^{-5}$, $3 \times 10^{-5}$, $9.75 \times 10^{-5}$, 0.000165, 0.0002325 \\ \hline
Vortex Street & SZ3 & $5 \times 10^{-5}$, 0.0001625, 0.000275, 0.0003975, 0.0005 \\ \hline
Vortex Street & SPERR & $1 \times 10^{-5}$, $3.25 \times 10^{-5}$, $5.5 \times 10^{-5}$, $7.05 \times 10^{-5}$, 0.0001 \\ \hline
Heated Cylinder & SZ3 & $5 \times 10^{-5}$, 0.0001625, 0.000275, 0.0003975, 0.0005 \\ \hline
Heated Cylinder & SPERR & $5 \times 10^{-5}$, 0.0001625, 0.000275, 0.0003975, 0.0005 \\ \hline
\end{tabular}
\end{center}
\end{table}

\begin{table}[!ht]
\begin{center}
\vspace{-6mm}
\caption{Parameter configurations used to generate \cref{tab:results-equivEB}. BC: Base compressor used in each trial. A-EB: Error bound for an augmented compressor. BC-EB: Error bound for the base compressor that produces a similar compression ratio to the augmented counterpart (i.e.,~when the augmented compressor uses the error bound listed under A-EB).}
\label{tab:equivEBParameters}
\begin{tabular}{cc|cc} \hline
Dataset & BC & A-EB & BC-EB \\ \hline
Stress A & SZ3 & $0.01$ & $0.00953125$ \\ 
Stress A & SPERR & $0.01$ & $0.009609375$ \\ 
\hline
Stress B & SZ3 & $0.01$ & $0.0090625$ \\ 
Stress B & SPERR & $0.01$ & $0.009140625$ \\ 
\hline
Brain A & SZ3 & $0.01$ & $0.00265625$ \\ 
Brain A & SPERR & $0.01$ & $0.00203125$ \\ 
\hline
Brain B & SZ3 & $0.01$ & $0.00111328125$ \\ 
Brain B & SPERR & $0.01$ & $0.00056640625$ \\ 
\hline
Ocean & SZ3 & $0.001$ & $0.00075$ \\ 
Ocean & SPERR & $0.001$ & $0.00065625$ \\ 
\hline
Miranda & SZ3 & $0.001$ & $5.078125\times 10^{-5}$ \\ 
Miranda & SPERR & $0.001$ & $9.765625\times 10^{-6}$ \\ 
\hline
Vortex Street & SZ3 & $0.001$ & $6.4453125\times 10^{-5}$ \\ 
Vortex Street & SPERR & $0.001$ & $3.466796875\times 10^{-5}$ \\ 
\hline
Heated Cylinder & SZ3 & $0.001$ & $0.00013671875$ \\ 
Heated Cylinder & SPERR & $0.001$ & $0.0001484375$ \\ 
\hline
\end{tabular}
\end{center}
\vspace{-4mm}
\end{table}

\begin{table}[!ht]
\caption{Parameter configurations used to generate \cref{fig:teaser}, \cref{fig:results-syms}, \cref{fig:results-asym2}, \cref{fig:extra-sym}, \cref{fig:extra-asym}, and \cref{fig:teaser-sperr}. BC: Base compressor used in each trial. A-EB: Error bound used for an augmented compressor. BC-EB: Error bound for the base compressor that produces a similar compression ratio to the augmented counterpart (i.e.,~when the augmented compressor uses the error bound listed under A-EB.)}
\label{tab:equivEBParametersFigures}
\vspace{-2mm}
\begin{tabular}{cc|cc} \hline
Dataset & BC & A-EB & BC-EB \\ \hline
Stress A & SZ3 & $0.06$ & $0.028798828125$ \\ 
Stress A & SPERR & $0.06$ & $0.052265625$ \\ 
\hline
Stress B & SZ3 & \new{$0.01$} & $0.03609375$ \\ 
Stress B & SPERR & \new{$0.01$} & $0.038671875$ \\ 
\hline
Brain A & SZ3 & $0.06$ & $0.0196875$ \\ 
Brain A & SPERR & $0.06$ & $0.0234375$ \\ 
\hline
Brain B & SZ3 & \new{$0.03$} & $0.00228515625$ \\ 
Brain B & SPERR & \new{$0.03$} & $0.001875$ \\ 
\hline
Ocean & SZ3 & $0.01$ & $0.003984375$ \\ 
Ocean & SPERR & $0.01$ & $0.00359375$ \\ 
\hline
Miranda & SZ3 & $0.003$ & $5.859375\times 10^{-5}$ \\ 
Miranda & SPERR & $0.0002325$ & $1.1806640625\times 10^{-5}$ \\ 
\hline
Vortex Street & SZ3 & $0.0005$ & $6.4453125\times 10^{-5}$ \\ 
Vortex Street & SPERR & $0.0001$ & $2.8515625\times 10^{-
5}$ \\ 
\hline
\hline
\end{tabular}
\end{table}

%% file: appendix-additional-experiments.tex
\section{\textcolor{\newcolor}{Reconstruction Quality}}
\label{appendix:additional-experiments}

\new{In this section, we give further analysis of reconstruction quality. We analyze the tradeoff between bit-rate and PSNR for asymmetric tensor fields in \cref{appendix:additional-experiments-asymmetric}. We provide error maps in \cref{appendix:additional-experiments-error-maps}. We give a distribution of the errors imposed by the base compressors (and fixed during augmentation) in \cref{appendix:additional-experiments-errors-corrected}. Finally, we analyze the causes of visual artifacts in \cref{appendix:additional-experiments-visual-artifacts}}

\subsection{\textcolor{\newcolor}{Compression of Asymmetric Tensor Fields}}
\label{appendix:additional-experiments-asymmetric}

In addition to the experiments highlighted in \cref{sec:results}, we study the tradeoff between bit-rate and PSNR for {\toolname} when compressing asymmetric tensor fields. 
We preserve the topology of either the eigenvector or the eigenvalue partition, but not both, as illustrated in \cref{fig:additional-experiments-rateDistortion}.
\begin{figure}[!ht]
    \includegraphics[width=\linewidth]{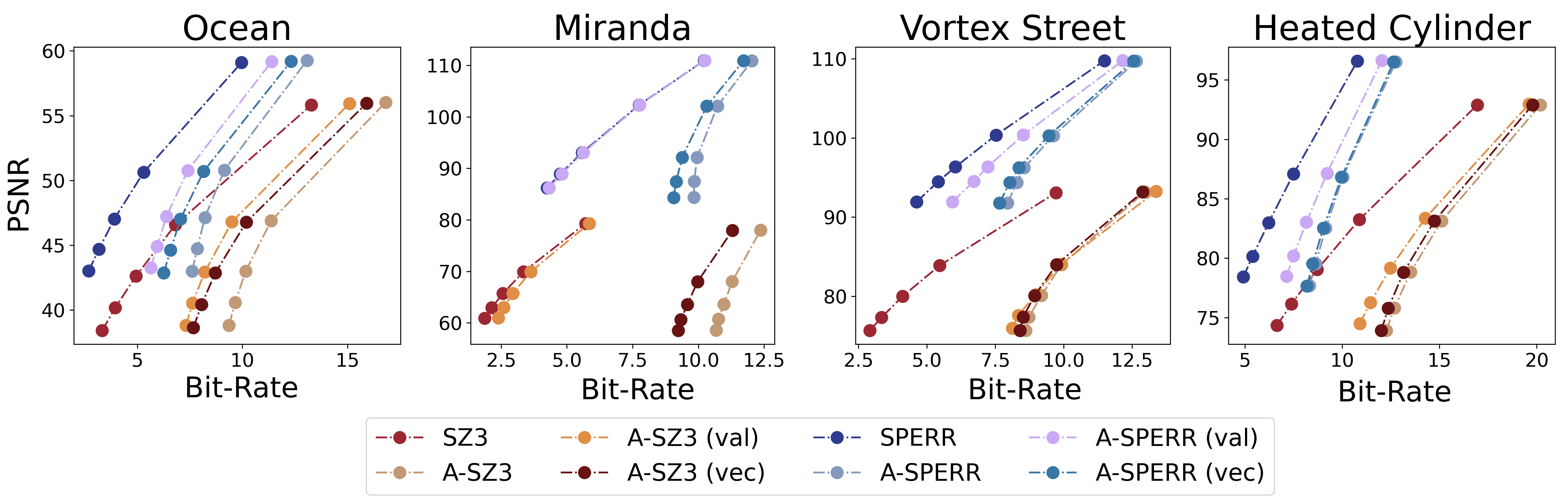}
    \vspace{-4mm}
    \caption{Plots showing the tradeoff between bit-rate and PSNR for SZ3 and SPERR on each dataset, as well as the augmented compressors. SZ3 (val) and SPERR (val) respectively preserve the topology of the eigenvalue partition only. 
    SZ3 (vec) and SPERR (vec) respectively preserve the topology of the eigenvector partition only.}
    \label{fig:additional-experiments-rateDistortion}
    \vspace{-4mm}
\end{figure}
In \cref{fig:additional-experiments-rateDistortion}, we plot the tradeoff between bit-rate and PSNR for SZ3, SPERR, augmented SZ3 and augmented SPERR. We show the same tradeoff  for augmented SZ3, which preserves the topology of either the eigenvalue or the eigenvector partition (but not both), denoted as A-SZ3 (val) or A-SZ3 (vec). 
We also show the corresponding curves for augmented SPERR. 

In \cref{fig:additional-experiments-rateDistortion}, we observe that preserving only one type of partition results in lower storage overhead compared to preserving both. 
In general, preserving the topology of the eigenvalue partition requires less storage than preserving the topology of the eigenvector partition. Specifically, for the Miranda dataset, preserving the topology of the eigenvalue partition incurs almost no storage overhead, whereas significant overhead is needed to preserve the topology of its eigenvector partition. 
   
\subsection{\textcolor{\newcolor}{Error Maps}}
\label{appendix:additional-experiments-error-maps}

\new{In \cref{fig:additional-experiments-error-map-vortex-street}, We include an error map of the Vortex Street dataset compressed with the SZ3 compressor both (A) before augmentation with {\toolname} and (B) after augmentation. Here, we can see that the error profile of the dataset is generally the same before and after augmentation. However, there are a few areas where augmentation increases or decreases the error in order to correct the topology.}

\new{In \cref{fig:additional-experiments-error-map-stress}, we include an error map of the Stress B dataset. In the top row, we include an error map similar to that showed in \cref{fig:additional-experiments-error-map-vortex-street} compressed with SZ3 (A) before augmentation and (B) after augmentation with $\xi = 0.01$. We show a $65 \times 65$ pixel grid, where each pixel corresponds to a tensor. We can see that the error profile is nearly identical in both cases. In the bottom row, we show an error profile for the eigenvector directions of the Stress B dataset (C) before augmentation and (D) after augmentation with $\xi = 0.01$. We show a linearly interpolated view, rather than a pixel grid as before, due to the nontrivial fashion in which the tensorlines behave under interpolation. We can see that the error profile is nearly identical before and after augmentation, except that, in a few areas of where the base compressor has significantly distorted the eigenvector directions, {\toolname} corrects the error in those regions during augmentation.}

\begin{figure}[!ht]
    \includegraphics[width=\linewidth]{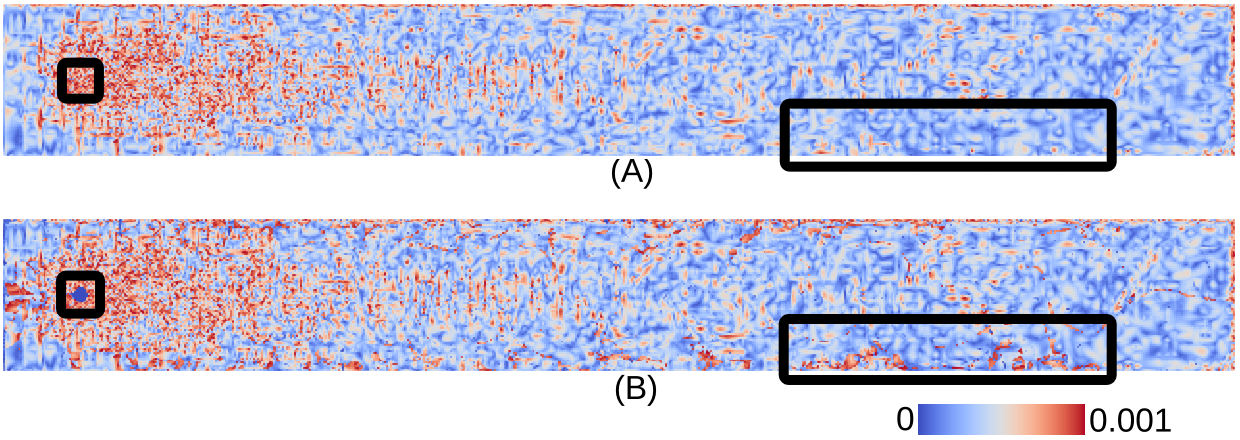}
    \vspace{-6mm}
    \caption{Error maps of Vortex Street dataset compressed with SZ3 both (A) before augmentation and (B) after augmentation. The error value is given as the maximum error across the four entries of each tensor as a percentage of the range. We used $\xi = 0.001$ for compression.}
    \label{fig:additional-experiments-error-map-vortex-street}
\end{figure}

\begin{figure}[!ht]
    \includegraphics[width=\linewidth]{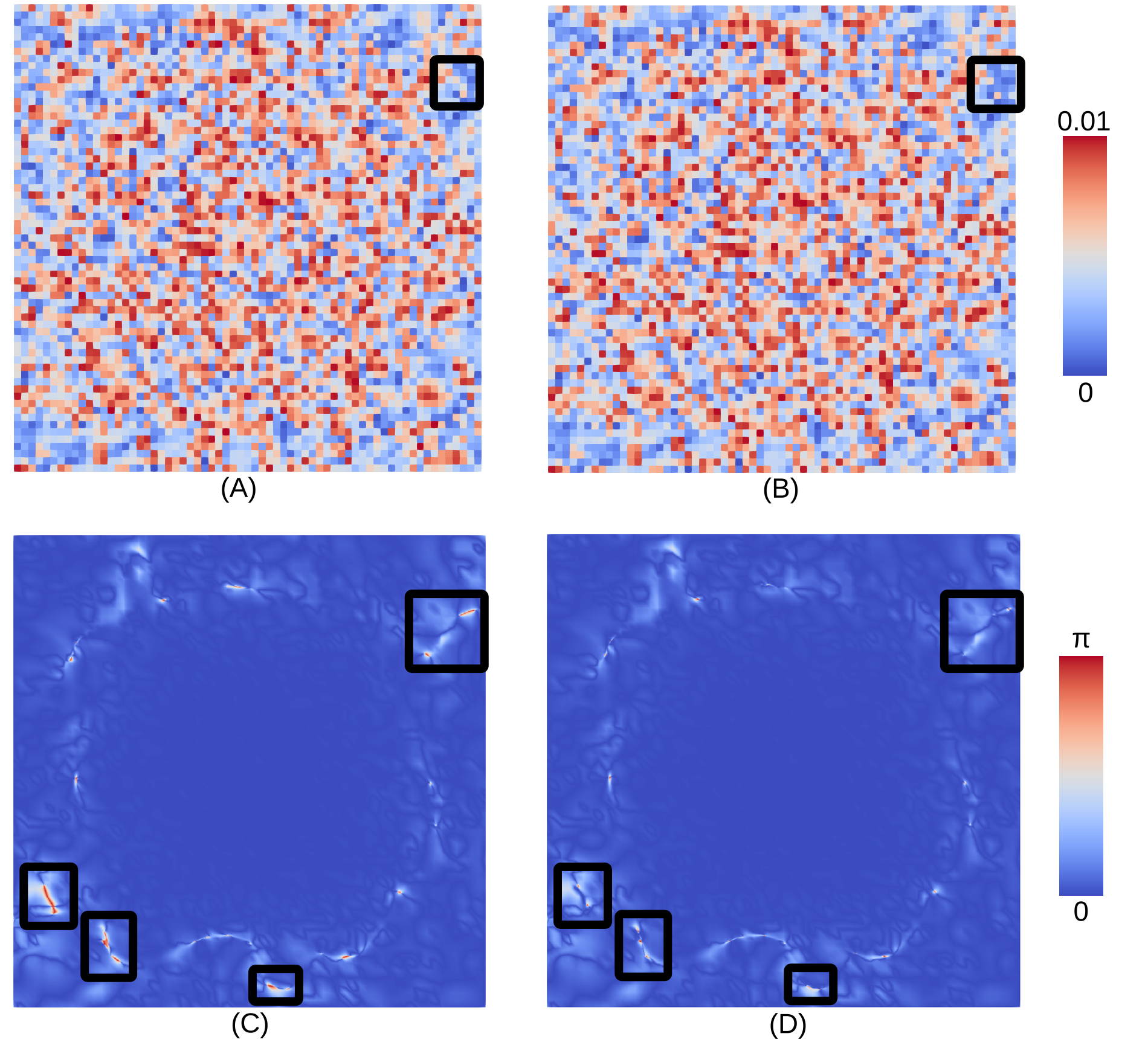}
    \vspace{-6mm}
    \caption{Error maps of the stress B dataset compressed with SZ3 and augmented SZ3. Top row: error map of the maximum error across the four entries of each tensor as a percentage of the range (A) before augmentation and (B) after augmentation. Each pixel represents a tensor. Bottom row: error map of the angle of the tensorlines (C) before augmentation and (D) after augmentation. We interpolate the tensor field to show distortions in the interior of cells. We used $\xi = 0.01$ for compression.}
    \label{fig:additional-experiments-error-map-stress}
    \vspace{-4mm}
\end{figure}

\subsection{Errors Corrected}
\label{appendix:additional-experiments-errors-corrected}
In \cref{fig:additional-experiments-error-histograms}, we give distributions of the different errors achieved on each asymmetric dataset by each base compressor before augmentation with {\toolname}. Here, we can see that the percentage of points that are misclassified according to the eigenvector and eigenvalue partitions are less than the total number of cells whose topology is distorted according to the eigenvector and eigenvalue partitions, respectively. This is logical, because whenever a point has the wrong classification type according to the eigenvalue or eigenvector partition, it distorts the topology of all six surrounding cells. Likewise, we found that, in most cases, correcting the vertices of a cell is sufficient to correct the internal topology of the cell, and further correction of the internal topology of a cell is not needed. We can see that there are typically very few errors related to degenerate points. To generate these numbers, we used the largest error bound for each combination of a dataset and base compressor listed in \cref{tab:rateDistortionParameters}.

\begin{figure}[!ht]
    \includegraphics[width=\linewidth]{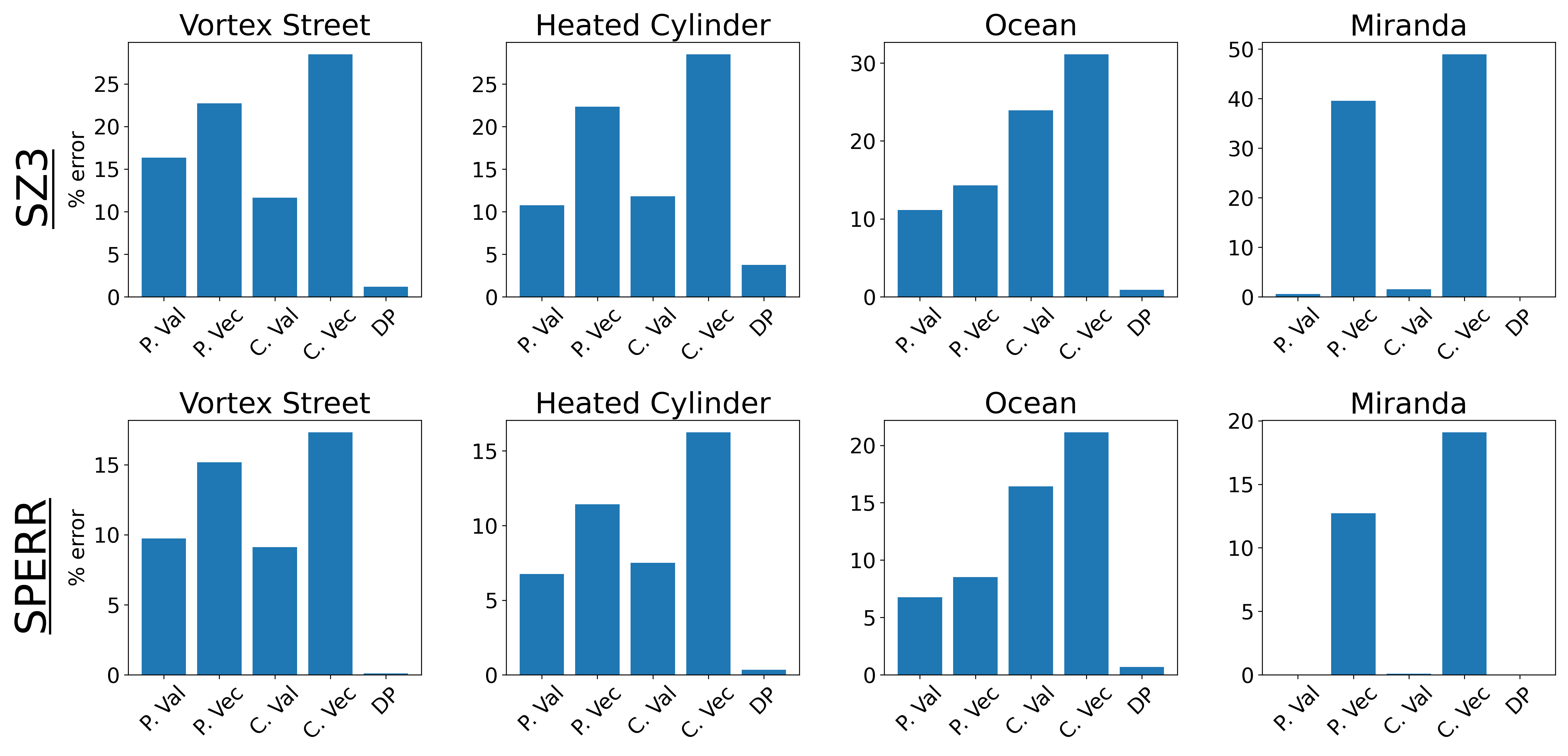}
    \vspace{-6mm}
    \caption{Histograms demonstrating the percentage of points and cells in each asymmetric dataset that exhibit certain types of errors before augmentation with {\toolname}. P. Val: Incorrect point classification in eigenvalue partition. P. Vec: Incorrect point classification in eigenvector partition. C. Val: Incorrect cell classification in eigenvalue partition. C. Vec: Incorrect cell classification in eigenvector partition (excluding degenerate points). DP: Errors pertaining to degenerate points within each cell.}
    \label{fig:additional-experiments-error-histograms}
    \vspace{-4mm}
\end{figure}

\subsection{\new{Visual Artifacts}}
\label{appendix:additional-experiments-visual-artifacts}

\new{One limitation of {\toolname} is that the decompressed data may exhibit visual artifacts, notably when using tensorline LIC visualization for symmetric data, and eigenvector/eigenvalue partition visualization  for asymmetric data. In the symmetric case, such artifacts primarily originate from the base compressor, rather than our augmentation layer. One can verify mathematically that the eigenvector directions are determined entirely by the value $\theta$ in the decomposition equation. Our augmentation layer can only improve the value of $\theta$ produced from the base compressor, and thus will not further distort the eigenvector directions. As shown in~\cref{fig:additional-experiments-error-map-stress}, our strategy produces an error profile for the eigenvector directions that is very similar to the base compressor, except for a few areas where {\toolname} has corrected some distortions. We observed that distortions are most prominent in regions where tensors magnitudes are low compared to the absolute error bound, allowing for greater relative distortion.}

In the asymmetric case, our strategy can cause visual artifacts. We particularly notice jagged artifacts, such as those in the zoomed-in region of \cref{fig:teaser}. These artifacts are most prominent in regions where the decomposition coefficients $(\gamma_d,\gamma_r,\gamma_s)$ are small compared to the absolute error bound $\xi$. In such cases, if our strategy replaces e.g. $\gamma_d' \gets \gamma_d' + \xi$, then the magnitude of $\gamma_d'$ can be significantly increased. 
If one coefficient of a tensor is increased in this manner while a neighboring tensor still has very small coefficients, the resulting discrepancy can produce jagged artifacts. 
However, our strategy ensures that the boundaries of partition regions will not shift by more than the width of a cell. Thus, in high resolution datasets, such as \cref{fig:results-asym2}, these artifacts will be less noticeable.  

%% file: appendix-algorithm-details.tex
\section{Algorithm Details}
\label{appendix:algorithm-details}

In this section, we provide additional details about our topology-preserving compression algorithm, for both symmetric (\cref{appendix:algorithm-details-symmetric}) and asymmetric (\cref{appendix:algorithm-details-asymmetric}) cases.

\subsection{Compression of Symmetric Tensor Fields}
\label{appendix:algorithm-details-symmetric}

Let $\sigma$ be a (triangular) cell. 
Let $T_1$, $T_2$, and $T_3$ be the tensors at the vertices of $\sigma$ for the ground truth data, and $T_1'$, $T_2'$ and $T_3'$ be the corresponding tensors in the reconstructed data. 
Following the decomposition in \cref{eqn:decomposition}, we obtain coefficients $\gamma_{d,1}$ for $T_1$, $\gamma_{d,1}'$ for $T_1'$ and so forth.

Suppose that we have handled case 1 of the cell correction step identified in \cref{sec:method-symmetric}. It is possible that, due to the limited precision afforded by linear-scaling quantization, the cell topology will not be preserved. 
In this case, we store the deviators of each tensor losslessly, one at a time, until the cell topology is preserved. In particular, for $i \in \{1,2,3\}$, we store $\Delta_i = \gamma_{s,i}\cos(\theta)$ and $F_i = \gamma_{s,i}\sin(\theta)$ losslessly. During decompression, we compose a tensor according to
\begin{equation}
T_i'' \gets \gamma_{d,i}' + \begin{pmatrix} \Delta_i & \quad F_i \\ F_i & -\Delta_i
\end{pmatrix}.
\end{equation}
In extremely rare cases, even storing all three deviators losslessly will not preserve the cell topology or respect the error bound, in which case we preserve the tensors completely losslessly. Likewise, it is possible that after adjusting one of the $\theta_i'$ using linear-scaling quantization, the error bound will not be respected. If this occurs, we store the deviator losslessly (and store the tensor losslessly if necessary).

\subsection{Compression of Asymmetric Tensor Fields}
\label{appendix:algorithm-details-asymmetric}

In rare cases, it is possible for the vertex correction step in \cref{sec:method-asymmetric-vertex} to fail. There are two ways that it can fail.

First, although each coefficient individually respects the error bound, it is possible—though unlikely—that $T''$ (the reconstructed tensor) does not. In such cases, if $\theta'$ is not currently quantized, we quantize it using the same strategy as in the symmetric case. Otherwise, we revert $\theta'$ to its unquantized value and apply logarithmic-scaling quantization to $\gamma_d'$, $\gamma_r'$, and $\gamma_s'$, using their values from the intermediate data prior to the procedure described in \cref{sec:method-asymmetric-vertex}. By \cref{lemma:error-bounds}, these values deviate from their respective ground truth values by at most $\xi$, $\xi$, and $\sqrt{2}\,\xi$, respectively. We then use logarithmic-scaling quantization to halve these error bounds and repeat the logic from the vertex correction step. This process of quantizing $\theta'$ and halving each error bound is repeated up to ten times, after which $T$ is stored losslessly.

Second, two coefficients may have equal magnitude in the reconstructed data but not in the ground truth, potentially causing classification issues. In such cases, as well as when an error is detected in the internal cell topology during the topology preservation step, we apply the process described above.

%% file: appendix-runTimes.tex
\section{Additional Description on Run Time Analysis}
\label{appendix:runTimes}

In this section, we provide additional information pertaining to the run times of {\toolname}. We explore the tradeoff between the total compression time, augmentation time,  and the error bound $\xi$ in \cref{sec:appendixRunTimesAblation}. 
We report in \cref{appendix:runTimesLimitedPreservation} the run times for asymmetric data when preserving one type of partition (but not both). 
We give a disaggregation of the run times during the cell correction step; see~\cref{appendix:runTimesCellCompression}. \new{We give statistics on how many times cells are visited in \cref{appendix:runTimesIterationStats}. Finally, we demonstrate the abilities of {\toolname} to compress data with larger 2D slices in \cref{appendix:runTimesLargeDatasets}.}

\begin{figure}[!h]
    \includegraphics[width=\linewidth]{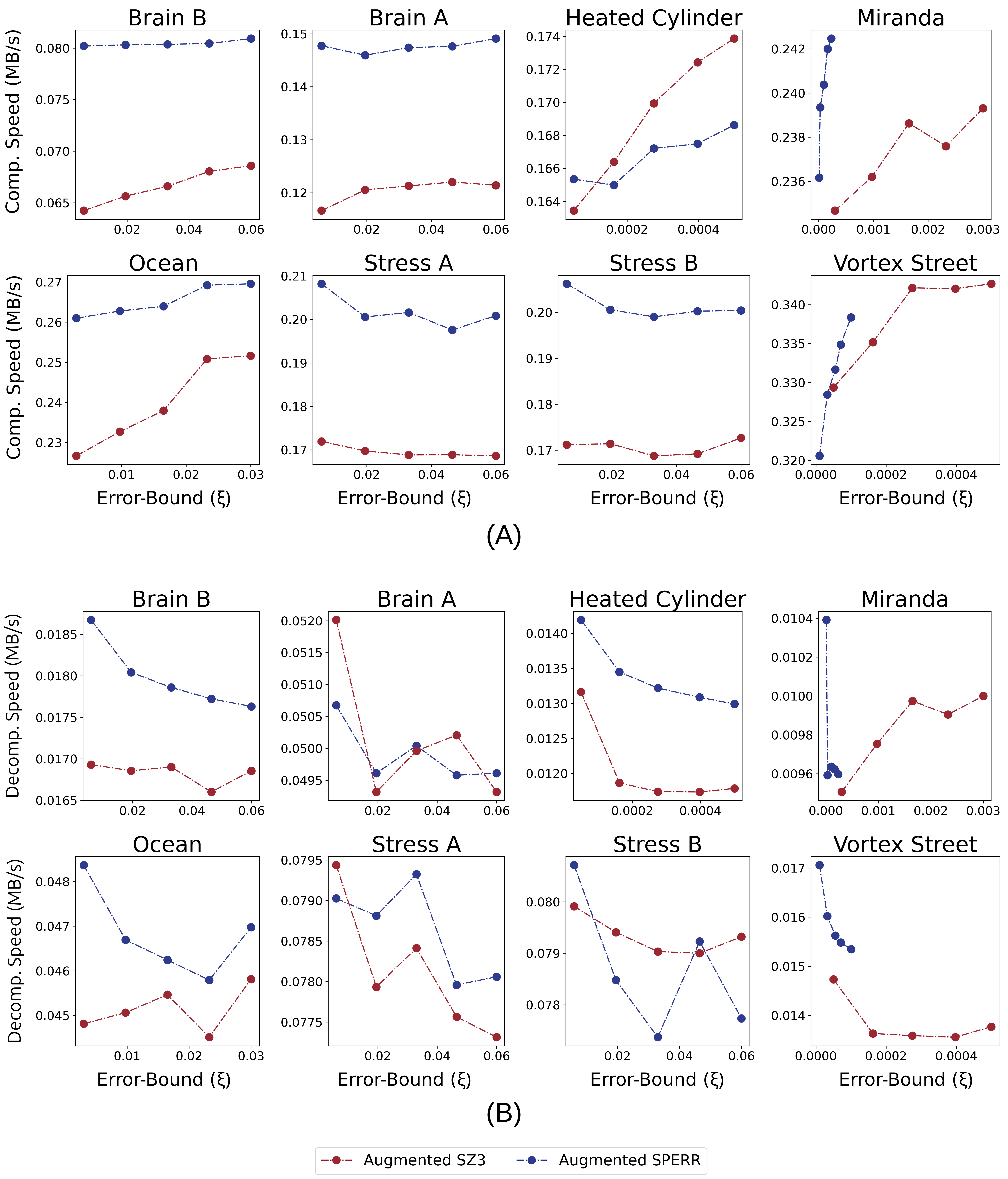}
    \caption{Plots showing throughput (in MB/s) versus error bound $\xi$ for each dataset, using augmented SZ3 and augmented TTHRESH: (A) throughput of compression; (B) throughput of decompression.}
    \label{fig:appendixRunTimesAblation}
\end{figure}

\subsection{Run Time Versus Error Bound}
\label{sec:appendixRunTimesAblation}

In \cref{fig:appendixRunTimesAblation}, we provide plots demonstrating the throughput in MB/s versus error bound $\xi$ for both augmented SZ3 and augmented TTHRESH on all eight datasets. We display the throughput for compression in (a) and decompression in (b). In general, there is no clear trend between the error bound and the throughput of either compression or decompression.

\begin{table*}
\caption{Run times (compression and decompression times) for SZ3, augmented SZ3, SPERR, and augmented SPERR on each of the four asymmetric datasets. We also provide the run times for augmented compressors preserving one type of partition, as A-SZ3 (val), A-SZ3 (vec), etc. We provide similar run times for augmented SPERR. $\xi = 0.001$. Times are in seconds.}
\label{tab:appendixRunTimesLimitedPreservation}
\vspace{-2mm}
\begin{tabular}{c|cccc|cccc}
Dataset & SZ3 & A-SZ3 & A-SZ3 (val) & A-SZ3 (vec) & SPERR & A-SPERR & A-SPERR (val) & A-SPERR (vec) \\ \hline
\multicolumn{9}{c}{Total Compression Time} \\ \hline
Ocean & 0.35 & 1.90 & 1.59 & 1.61 & 0.39 & 2.28 & 1.87 & 1.97 \\
Miranda & 7.00 & 274.14 & 161.84 & 213.26 & 9.53 & 287.16 & 171.51 & 221.46 \\
Vortex Street & 24.34 & 784.85 & 526.90 & 755.19 & 28.86 & 835.99 & 490.62 & 776.49 \\
Heated Cylinder & 37.98 & 762.21 & 655.79 & 717.39 & 48.23 & 713.11 & 609.04 & 692.66 \\ \hline
\multicolumn{9}{c}{Decompression Time} \\ \hline
Ocean & 0.34 & 0.38 & 0.39 & 0.38 & 0.36 & 0.42 & 0.42 & 0.42 \\
Miranda & 4.73 & 11.44 & 7.06 & 11.21 & 6.34 & 12.60 & 9.12 & 11.93 \\
Vortex Street & 19.78 & 33.76 & 32.99 & 32.96 & 22.81 & 35.13 & 34.57 & 34.97 \\
Heated Cylinder & 29.01 & 50.78 & 46.56 & 49.56 & 36.68 & 53.44 & 51.22 & 53.07 \\ \hline
\end{tabular}
\vspace{-4mm}
\end{table*}

\subsection{Preserving Eigenvalue or Eigenvector Partitions}
\label{appendix:runTimesLimitedPreservation}
In \cref{tab:appendixRunTimesLimitedPreservation}, we provide the compression and decompression times for SZ3, augmented SZ3, SPERR, and augmented SPERR. 
We also provide the run times preserving eigenvalue or eigenvector partition (but not both). All times reported are for $\xi = 0.001$. 

Preserving the topology of only one partition (eigenvector or eigenvalue) can reduce compression time by up to about $45\%$, with more time saved when preserving the eigenvector partition (as opposed to the eigenvalue partition). 
Decompression time can be reduced by up to about $40\%$. The gains in decompression time are generally smaller than those for compression time, typically under $10\%$.
However, the time savings vary across datasets and (base) compressors. 

\subsection{Run Time Analysis of Cell Compression Step}
\label{appendix:runTimesCellCompression}

For asymmetric data, the cell correction step consists of three main parts: (1) correcting the vertex classifications, (2) preserving the degenerate points of the dual-eigenvector fields, and (3) preserving the topology of each cell.  
In \cref{tab:appendixRunTimesCellCompression}, we present the run times (in seconds) for each step across datasets and compressors.  
It is evident that preserving cell topology requires the most time in each trial, although all three steps contribute significantly to the overall run time of the cell correction process.

\begin{table}
\begin{center}
\caption{Run times (in seconds) during cell correction for three asymmetric datasets. We display the times for augmented SZ3 in the top and augmented SPERR in the bottom. `Vertices': correcting the vertex classifications; `Degen. Pts.': preserving the degenerate points of the dual-eigenvector fields; `Cell Top.': preserving the topology of each cell.}
\label{tab:appendixRunTimesCellCompression} 
\vspace{-2mm}
\begin{tabular}{c|ccc} \hline
Dataset & Vertices & Degen. Pts. & Cell Top. \\ \hline
\multicolumn{4}{c}{Augmented SZ3} \\ \hline
Ocean & 0.23 & 0.17 & 0.68 \\
Miranda & 53.00 & 31.52 & 153.31 \\
Vortex Street & 192.46 & 125.35 & 368.00 \\
Heated Cylinder & 179.34 & 102.34 & 334.12 \\ \hline
\multicolumn{4}{c}{Augmented SPERR} \\ \hline
Ocean & 0.25 & 0.17 & 0.79 \\
Miranda & 54.19 & 32.12 & 156.66 \\
Vortex Street & 207.51 & 139.99 & 372.44 \\
Heated Cylinder & 166.35 & 91.92 & 283.27 \\ \hline
\end{tabular}
\vspace{-4mm}
\end{center}
\end{table}

\subsection{Statistics on Iterations}
\label{appendix:runTimesIterationStats}

In this section, we analyze how many times cells are visited to provide insight into the overall running times. In \cref{fig:additional-experiments-iteration-counts}, we show how many times each cell is processed for each dataset for SZ3. The numbers are similar for SPERR. Here, we can see that almost all points (90-95\%) are processed only once in general, while very few cells are processed six or more times. The notable exception is the Vortex Street dataset, where many points are processed six or more times. One can verify that the Vortex Street dataset has the lowest throughput overall, and this may be the cause.

To generate \cref{fig:additional-experiments-iteration-counts}, we use the largest error bound for each dataset given in \cref{tab:rateDistortionParameters} for augmented SZ3.

\begin{figure}[!ht]
    \includegraphics[width=\linewidth]{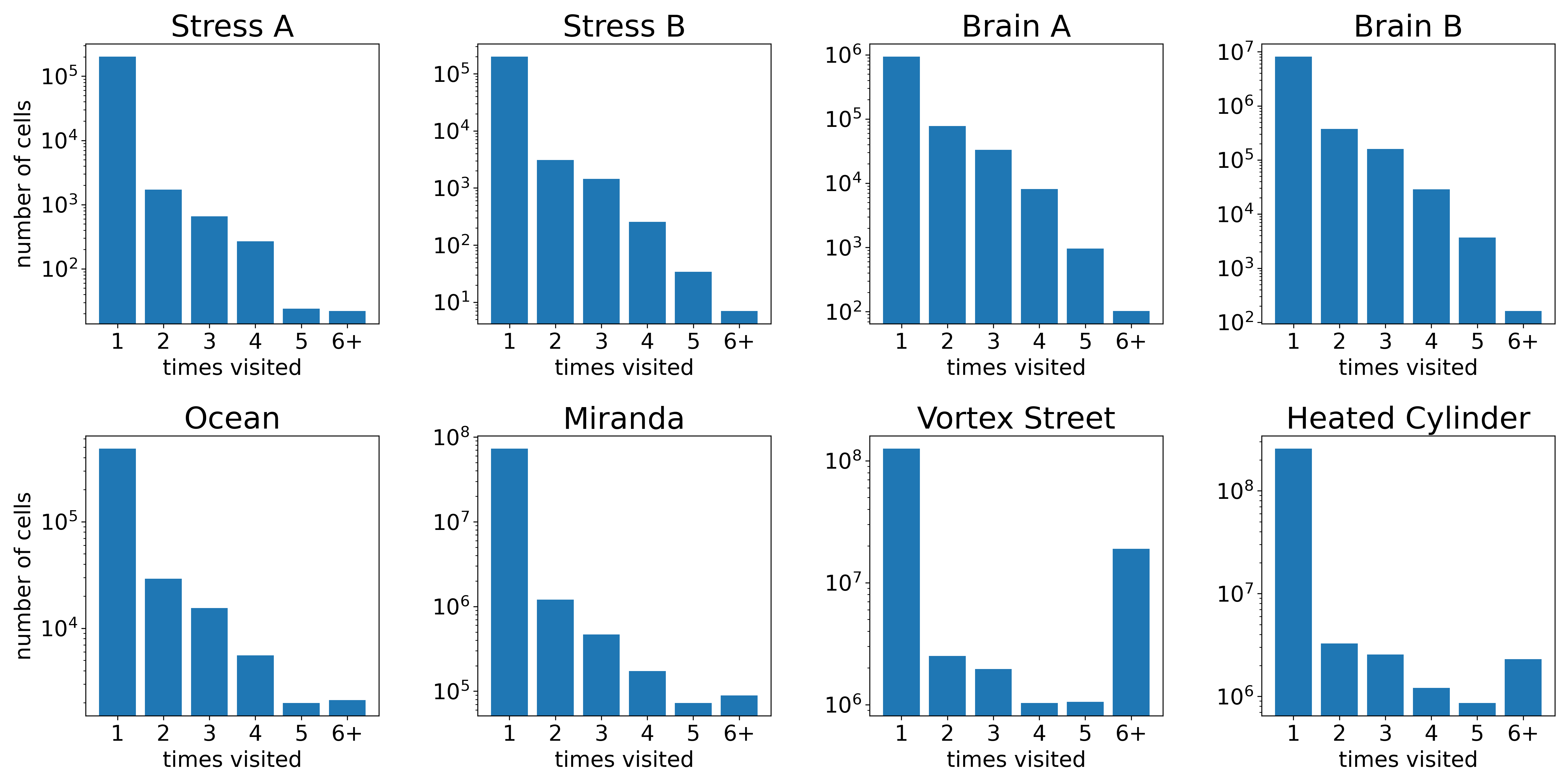}
    \vspace{-2mm}
    \caption{Histograms demonstrating the distribution of how many times cells were processed. For example, the ``1'' bar shows the number of cells that were processed only once. ``6+'' gives the number of cells that were processed six or more times.}
    \label{fig:additional-experiments-iteration-counts}
    \vspace{-4mm}
\end{figure}

\subsection{\textcolor{\newcolor}{Testing on Data With Large Slices}}
\label{appendix:runTimesLargeDatasets}
In most of our experiments, we use relatively small 2D slices—smaller than those in typical use cases. Since our algorithm runs in linear time, we believe these experiments are sufficient to demonstrate the framework’s throughput. 
To verify that the throughput observed on smaller datasets also scales to larger data, we run {\toolname} on datasets with larger 2D slices. Specifically, we derive an asymmetric tensor field from the E.U. Copernicus Marine Service Global Ocean Physics Reanalysis dataset (the source of our main ocean dataset). We select the largest rectangular slice available that does not significantly overlap with land.
Our slice is taken from the South Pacific, bounded by the following coordinates: North: $6.870530532659802$, East: $-80.59961032889228$, South: $-71.2633145566259$, and West: $-181.429859882464$. The depth ranges from 1.54 m at the top to 2533.3359375 m at the bottom. We derive the asymmetric tensor field using the same strategy applied to the other asymmetric datasets, and obtain a symmetric tensor field by taking the symmetric part of the tensor field.

Overall, our tensor fields have 42 slices of size $1210 \times 938$. The symmetric tensor field has a file size of 1144MB, while the asymmetric tensor field is 1525MB. We compress the symmetric dataset with $\xi=0.01$ and the asymmetric dataset with $\xi=0.001$. We obtain compression times of 121.7s for the symmetric tensor field, and 416.9s for the asymmetric tensor field, yielding respective throughput of 9.4 MB/s and 3.6 MB/s. These numbers are comparable to those yielded by the smaller datasets. The symmetric dataset has a bit-rate of 4.9 and a PSNR of 47.3, while the asymmetric dataset has a bit-rate of 13.6. and a PSNR of 65.7. These numbers are consistent with the reconstruction qualities reported on the smaller datasets. 

%% file: appendix-extra-figures.tex
\section{Additional Renderings for Visual Comparison}
\label{appendix:extra-figures}

In this section, we provide additional renderings of our experimental datasets for visual comparison. 
In \cref{fig:extra-sym}, we visualize Stress A and Brain A datasets compressed with SZ3, augmented SZ3, SPERR, and augmented SPERR, respectively.
We chose an error bounds $\xi$ such that each augmented compressor achieves a similar compression ratio to its corresponding base compressor. 
For the Stress A dataset, we highlight a region of interest. 
And for the Brain A dataset, we provide a zoomed-in view to mark the discrepancies between the base compressors and the ground truth, not visible in the augmented compressors.
In \cref{fig:extra-asym}, we provide a similar view for the Vortex Street \new{and Heated Cylinder} dataset\new{s}. 
In \cref{fig:teaser-sperr}, we include a variation of \cref{fig:teaser} using SPERR and augmented SPERR (instead of SZ3 and augmented SZ3). 
We report the error bounds $\xi$ used to generate each figure in \cref{appendix:parameter-configurations}.

\begin{figure*}
\includegraphics[width=\textwidth]{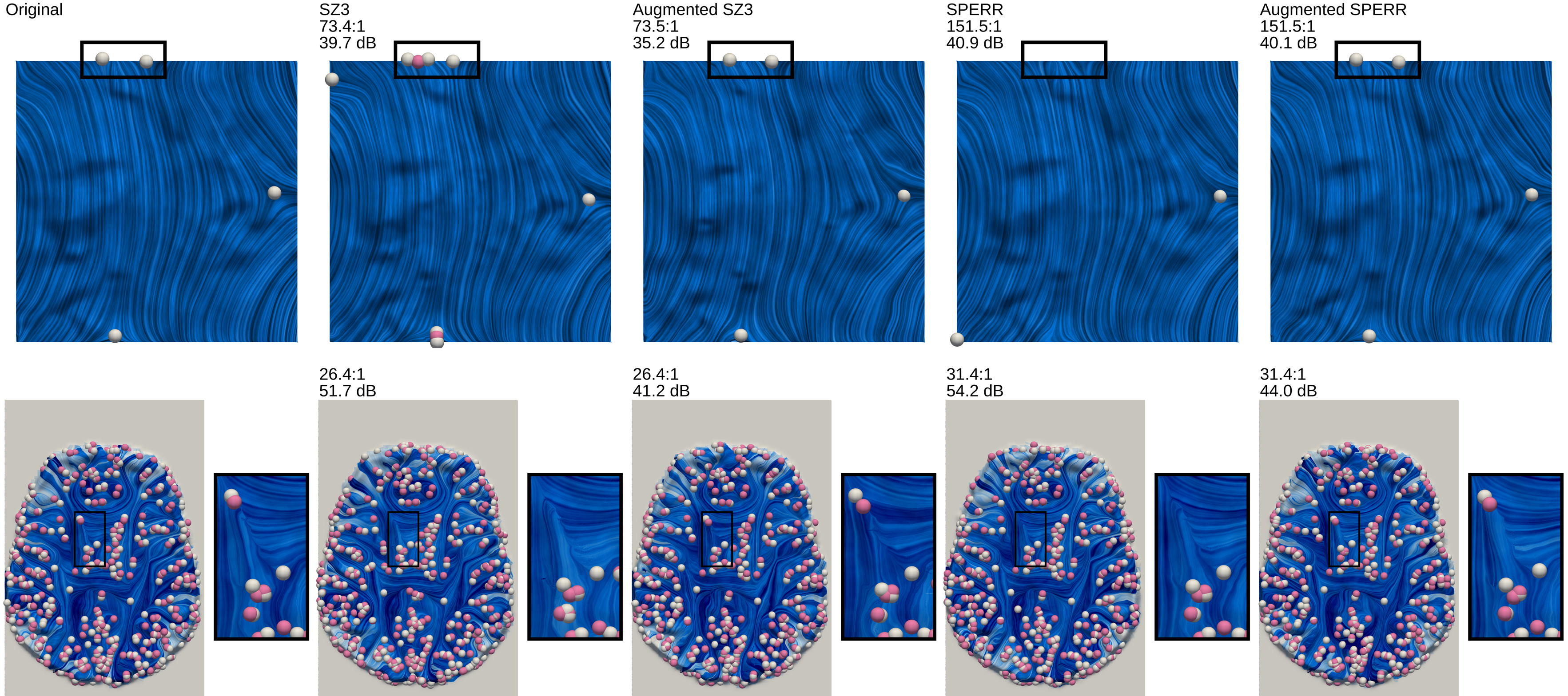}
\vspace{-2mm}
\caption{LIC visualization of the eigenvector fields of two 2D symmetric second-order tensor fields compressed with SZ3, augmented SZ3, SPERR, and augmented SPERR, along with the ground truth. Trisectors are in white, wedges are in pink. Top: Stress A data slice 13. Bottom: {\BrainA} data slice 50. In the top row, we highlight a region of interest in black boxes. In the bottom row, we provide a zoomed-in view of a region of interest encloded by black boxes. The Z position of each point corresponds to the Frobenius norm with smoothing applied.}
\label{fig:extra-sym}
\vspace{-8mm}
\end{figure*}

\begin{figure*}
\includegraphics[width=\textwidth]{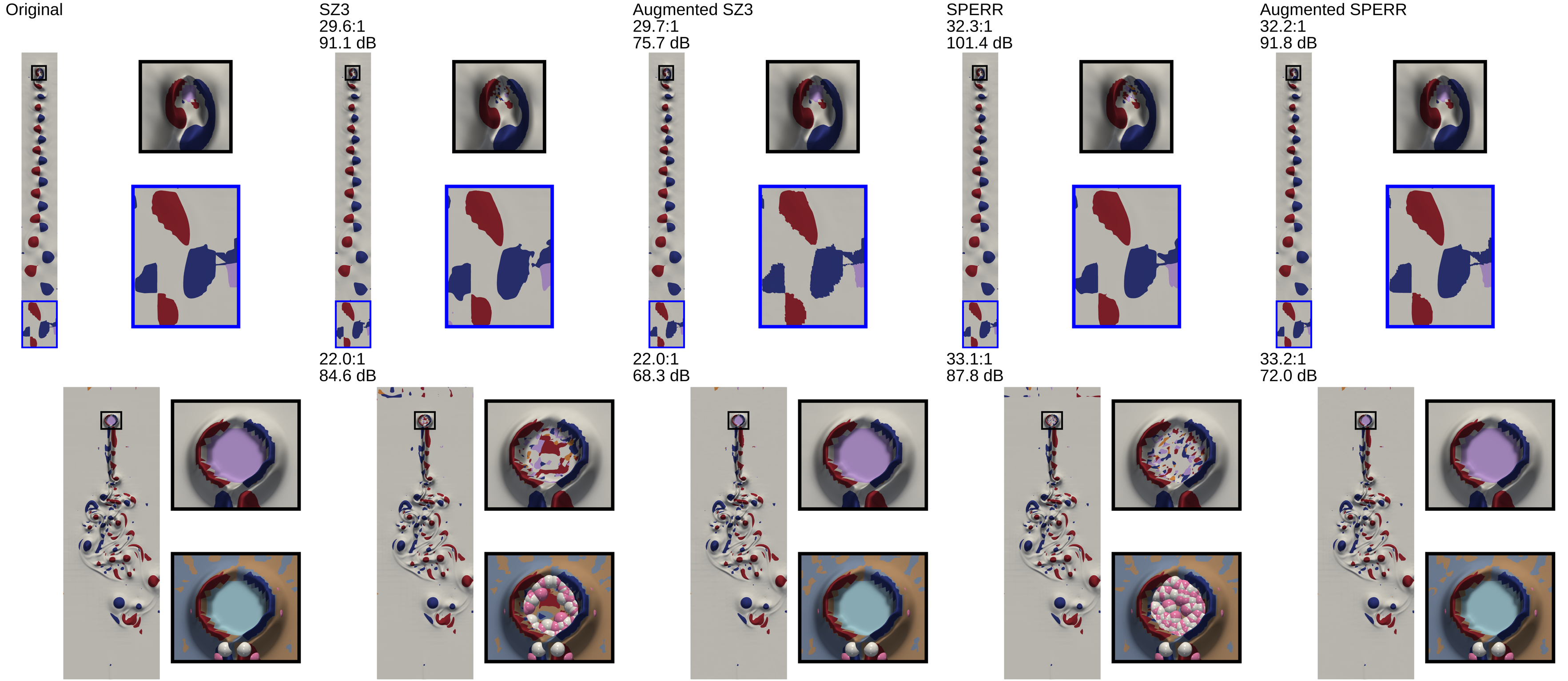}
\vspace{-2mm}
\caption{Visualizing the eigenvalue partition of the Vortex Street (top) and Heated Cylinder (bottom) datasets (slices 1000 and 800 resp.) compressed with SZ3, augmented SZ3, SPER, and augmented SPERR, along with the ground truth. We provide zoomed-in views (of black and blue boxes) that highlight the differences between the compressors and the ground truth. For the heated cylinder dataset, one zoomed in view corresponds to the eigenvector partition. We also label compression ratio and PSNR. We use the same colormap as \cref{fig:manifolds}. The Z position of each point corresponds to the Frobenius norm with smoothing applied.}
\label{fig:extra-asym}
\vspace{-8mm}
\end{figure*}

\begin{figure*}
\begin{center}
\includegraphics[width=0.8\textwidth]{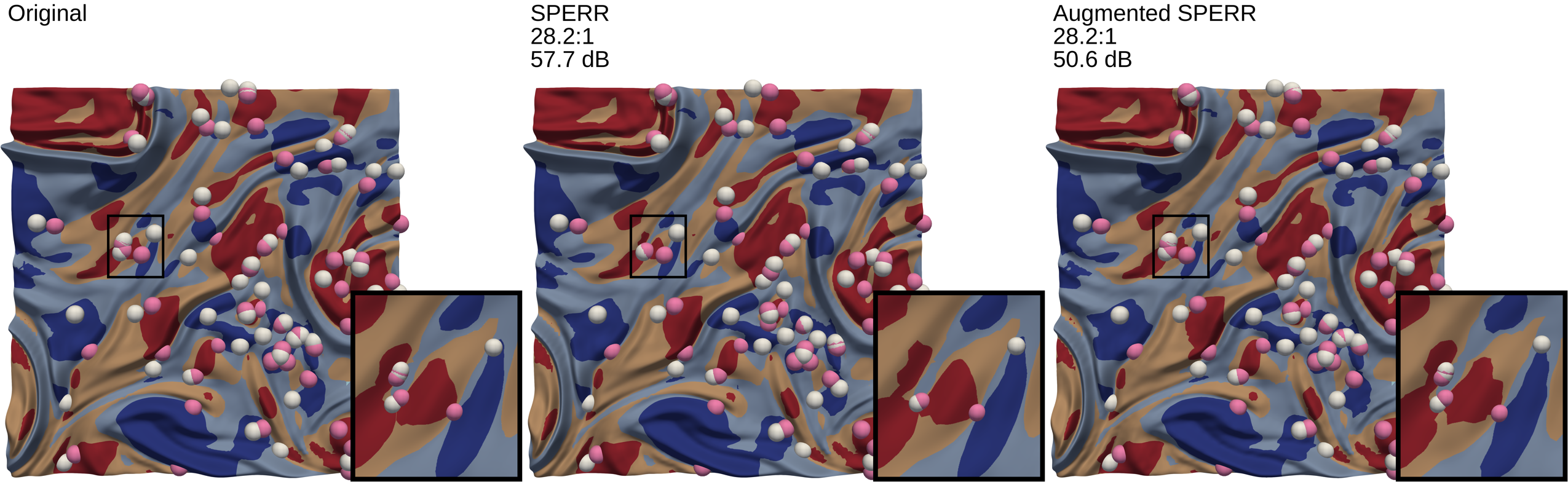}
\end{center}
\vspace{-4mm}
\caption{Visualizing the eigenvector partition of the Ocean dataset compressed with SPERR and augmented SPERR. Left: input data visualized with degenerate points of the dual-eigenvector field. Trisectors are in white, wedges are in pink. Middle: reconstructed data using SPERR, labeled with compression ratio and PSNR. Right: reconstructed data using augmented SPERR along with compression ratio and PSNR. We also provide zoomed-in views that highlight the differences between the classic and augmented SPERR results. We use the same colormap as \cref{fig:manifolds}. \new{The Z position  of each point corresponds to the Frobenius norm with smoothing applied.}}
\label{fig:teaser-sperr}
\end{figure*}

%% file: appendix-special-cases.tex
\section{Edge Cases During Asymmetric Cell Correction}
\label{appendix:edge-cases}

We handle various edge cases that arise in our topology-preserving framework. 
We describe edge cases pertaining to vertex correction in \cref{sec:edge-cases-vertex-correction} and cell topology preservation in \cref{sec:edge-cases-internal-topology}. 
We describe how we handle issues arising from floating point precision in \cref{sec:edge-cases-numerical-precision}.

\subsection{Vertex Correction}
\label{sec:edge-cases-vertex-correction}
We assign values to the variables $\DSIGN$, $\RSIGN$, $\ROVERS$, and $\DLARGEST$ to ensure that, for each vertex of the mesh, the following conditions hold if and only if they hold in the ground truth:
\begin{itemize}[noitemsep]
\item $|\gamma_r'| = \gamma_s'$
\item $\gamma_r' = 0$
\item $|\gamma_d'| = |\gamma_r'| > \gamma_s'$
\item $|\gamma_d'| = \gamma_s' > |\gamma_r| $
\item $|\gamma_r'| = |\gamma_d'| = |\gamma_s|$
\end{itemize}
To accomplish this, we adjust the decompressed data using strategies similar to those in \cref{sec:method-asymmetric-vertex}.

\subsection{Cell Topology Preservation}
\label{sec:edge-cases-internal-topology}

Sometimes, edge cases arise when computing the topological invariant for preserving cell topology using {\toolname}. Each edge case yields a number of sub-cases. 
We describe non-transverse intersections in \cref{sec:edge-cases-internal-topology-non-transverse}, junction points on cell boundaries in \cref{sec:edge-cases-internal-topology-junction-point-edge}, and intersections with cell vertices in \cref{sec:edge-cases-internal-topology-corners}. Finally, we describe how we handle degenerate conics in \cref{sec:edge-cases-internal-topology-degenerate-conic}.

\subsubsection{Non-Transverse Intersections}
\label{sec:edge-cases-internal-topology-non-transverse}

When computing the topological invariant, we trace the curves $\gamma_d^2 = \gamma_s^2$ and $\gamma_r^2 = \gamma_s^2$. An edge case occurs when one of these curves intersects the cell boundary or the other curve non-transversally (i.e,~an intersection that does not satisfy the transversality condition). It is also possible for the two conics to have significant overlap but not be equal. We handle each case using a virtual perturbation. We describe these cases in \cref{tab:edge-cases-internal-topology-non-transverse}. In the left column, we provide a description of the case and how we handle it. In the middle column, we visualize the case. In the right column, we visualize the virtual perturbation.

There are also a few other cases of non-transverse intersections that cannot be easily visualized. We describe them below:

\begin{itemize}[noitemsep,leftmargin=*]
\item If, for the entire cell, $|\gamma_d| = |\gamma_r|$, then we proceed as though $|\gamma_r| > |\gamma_d|$.
\item If, for the entire cell, $|\gamma_d| = \gamma_s$, then we proceed as though $\gamma_s > |\gamma_d|$.
\item If, for the entire cell, $|\gamma_r| = \gamma_s$, then we proceed as though $\gamma_s > |\gamma_r|$.
\item If there exists a point $z \in \sigma$ where $\gamma_d(z) = \gamma_r(z) = \gamma_s(z) = 0$, we note this in our invariant. We also note whether it occurs in the interior, on an edge, or on a vertex.
\item If, for the entire cell, $\gamma_d = \gamma_r = \gamma_s = 0$, then we note this as part of our invariant, and save the cell losslessly.
\end{itemize}

Finally, consider the case where the curves $|\gamma_r| = \gamma_s$ and $|\gamma_d| = \gamma_s$ are the same curve, but $|\gamma_d| \neq |\gamma_r|$ for the entire cell. 
Because $\gamma_r$ and $\gamma_d$ are both PL functions over $\sigma$, this can occur only if $|\gamma_r| = \gamma_s$ is a line. In such a case, one of $|\gamma_r|$ or $|\gamma_d|$ will be smaller on the entire cell; we ignore whichever is smaller.

\begin{table}[!ht]
\caption{Descriptions and images of edge cases that arise in internal cell topology computation from non-transverse intersections. We describe each case and how it is handled in the first column. In the second column, we provide a visualization of the case. In the third column, we provide an image of the virtual perturbation used to handle the case.}
\label{tab:edge-cases-internal-topology-non-transverse}
\begin{tabular}{p{0.35\linewidth}|c|c} \hline
Case & Image & Perturbation \\ \hline
\para{Non-transverse edge intersection:} If one of the conics intersects the edge non-transversally, then we apply a virtual perturbation such that the intersection never occurs. & \includegraphics[width=0.25\linewidth,valign=t]{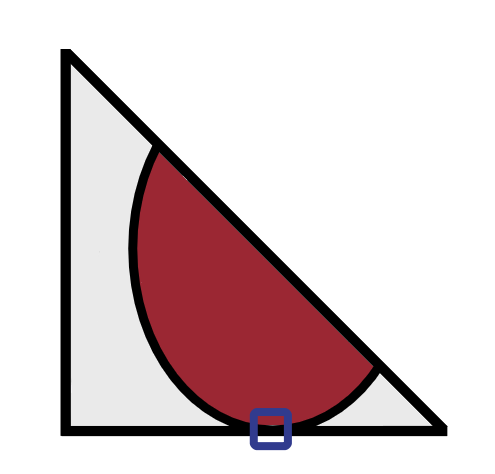} & \includegraphics[width=0.25\linewidth,valign=t]{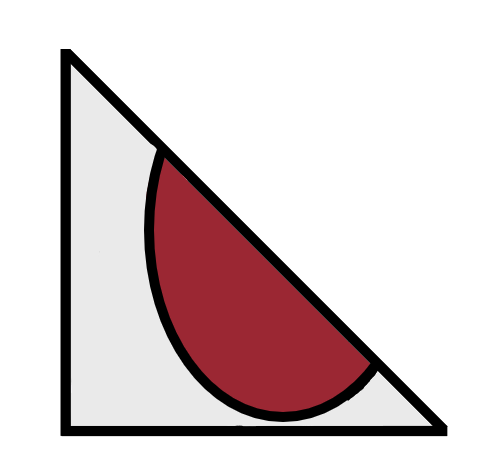} \\ \hline

\para{Non-transverse intersection between conics:} If two of the conics intersect non-transversally, then we apply a virtual perturbation such that the intersection never occurs. & \includegraphics[width=0.25\linewidth,valign=t]{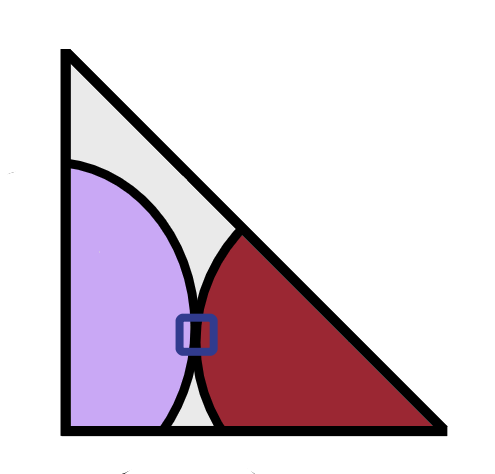} & \includegraphics[width=0.25\linewidth,valign=t]{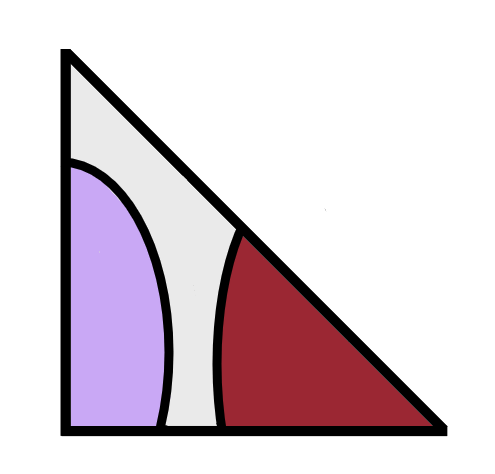} \\ \hline

\para{Partial overlap of conics:} If the conics are degenerate (e.g., intersecting lines, parallel lines, etc.), they may not be exactly identical but can still exhibit substantial overlap. In this case, we perturb the conics so that their intersection has measure zero. & \includegraphics[width=0.25\linewidth,valign=t]{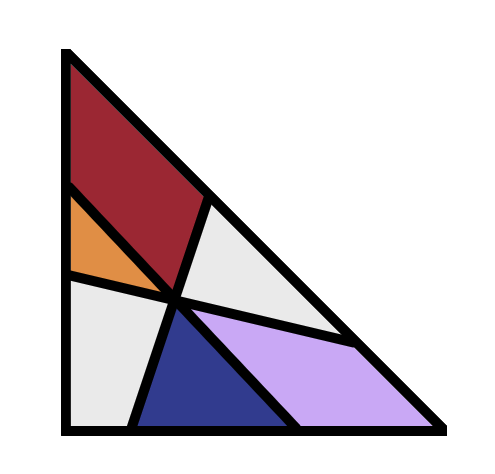} & \includegraphics[width=0.25\linewidth,valign=t]{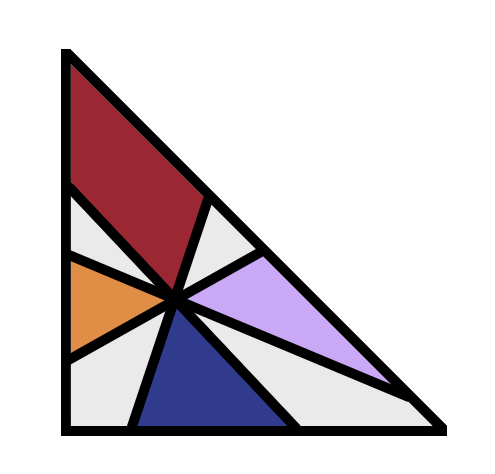} \\ \hline
\end{tabular}
\vspace{-2mm}
\end{table}

\subsubsection{Junction Point at Cell Boundary}
\label{sec:edge-cases-internal-topology-junction-point-edge}

If a junction point occurs on the boundary of a cell, this can lead to ambiguity. There are many variations of this case. We handle them all using virtual perturbations. We describe how we handle such cases in \cref{tab:edge-cases-internal-topology-junction-point-edge}.

\begin{table}
\caption{Descriptions and images of edge cases that arise in internal cell topology computation when a junction point occurs on the cell boundary. We describe each case and how it is handled in the first column. In the second column, we provide a visualization of the case. In the third column, we provide an image of the virtual perturbation used to handle the case.}
\label{tab:edge-cases-internal-topology-junction-point-edge}
\begin{tabular}{p{0.35\linewidth}|c|c} \hline
Case & Image & Perturbation \\ \hline
\para{Junction point at an edge with transverse intersection:} If a junction point occurs at an edge, and the two conics both intersect the edge transversally, then we apply a virtual perturbation such that the junction point does not occur. & \includegraphics[width=0.25\linewidth,valign=t]{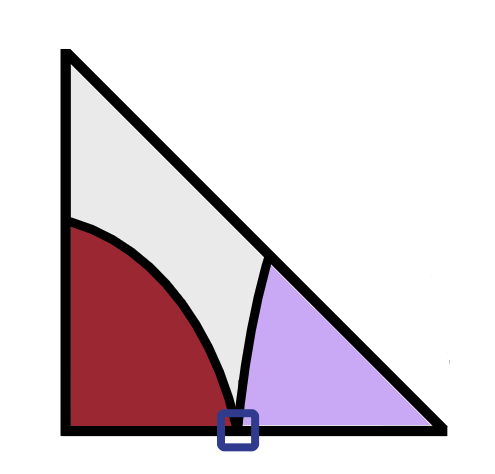} & \includegraphics[width=0.25\linewidth,valign=t]{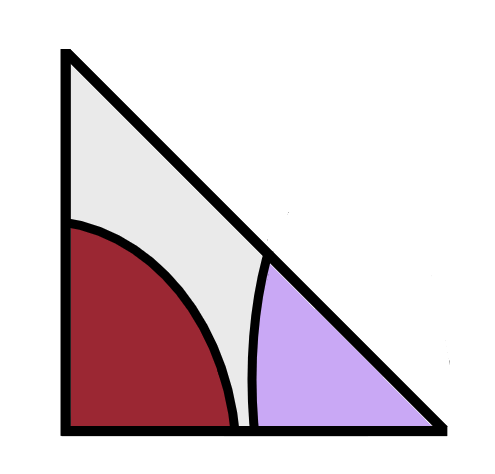} \\ \hline

\para{Junction point at an edge with single non-transverse conic-edge intersection:} If a junction point occurs at an edge, and exactly one conic intersects the edge non-transversally at the junction point, then we apply a virtual perturbation so that the nontransverse intersection does not occur. & \includegraphics[width=0.25\linewidth,valign=t]{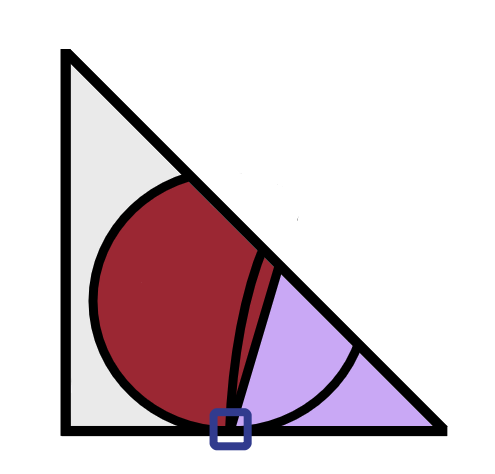} & \includegraphics[width=0.25\linewidth,valign=t]{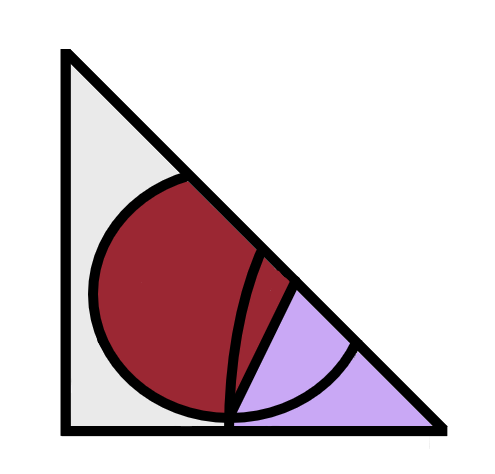} \\ \hline

\para{Junction point at an edge with only non-transverse intersections:} If a junction point occurs at an edge, and all intersections are non-transverse, then we apply a virtual perturbation so that none of the intersections occur. & \includegraphics[width=0.25\linewidth,valign=t]{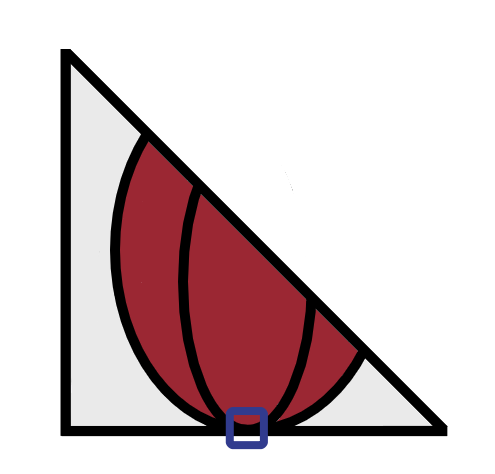} & \includegraphics[width=0.25\linewidth,valign=t]{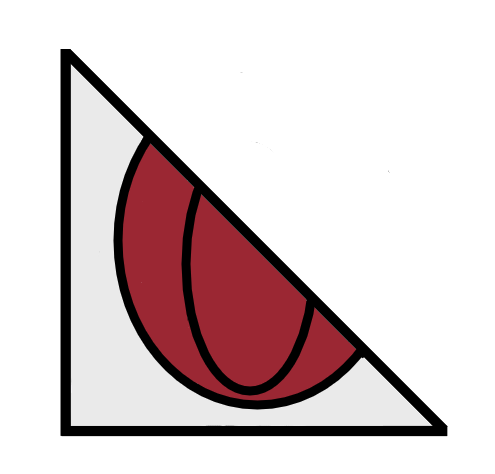} \\ \hline

\end{tabular}
\end{table}

\subsubsection{Intersections at Vertices}
\label{sec:edge-cases-internal-topology-corners}

If one of the conics intersects a vertex of the cell, this can lead to ambiguity. We describe how such cases are handled in \cref{tab:edge-cases-internal-topology-corners}. In the left column, we describe the case and how we handle it. In the right column, we provide a visualization. We also handle these cases in a similar fashion if the curves $|\gamma_r| = |\gamma_d|$ or $\gamma_r = 0$ intersects any vertex.

\begin{table}
\caption{Descriptions and images of edge cases that arise in internal cell topology computation when a conic section intersects a vertex. In the first column we describe each case and how it is handled. In the second column we provide a visualization of the case.}
\label{tab:edge-cases-internal-topology-corners}
\begin{tabular}{p{0.6\linewidth}|c} \hline

Case & Image \\ \hline

\para{Conic intersects a vertex (no topological effect):} If one conic section intersects a vertex, but that intersection does not mark the location of any topological change, then we ignore it. & \includegraphics[width=0.25\linewidth,valign=t]{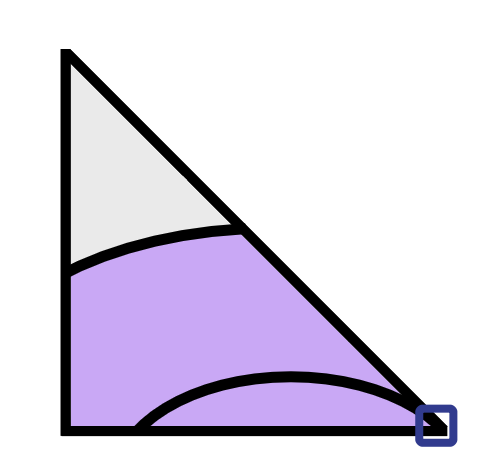} \\ \hline

\para{Conic intersects a vertex (topologically significant intersection):} If one conic section intersects a vertex, and that intersection does affect the topology, then we track this intersection in our invariant. & \includegraphics[width=0.25\linewidth,valign=t]{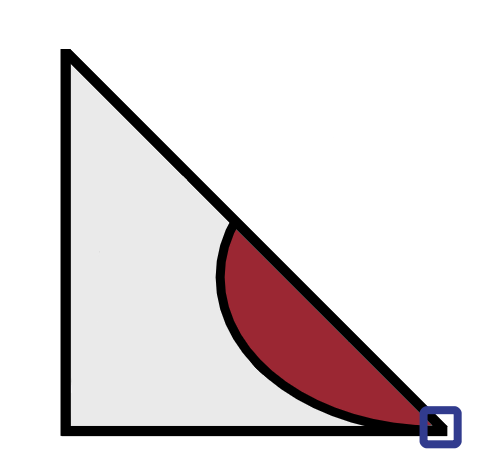} \\ \hline

\para{Two conics intersect at a vertex:} If two conic intersects intersect at a vertex, we ignore the junction point that occurs at the vertex. We treat each of the two conics as separately intersecting the vertex and handle them according to the two previous cases. & \includegraphics[width=0.25\linewidth,valign=t]{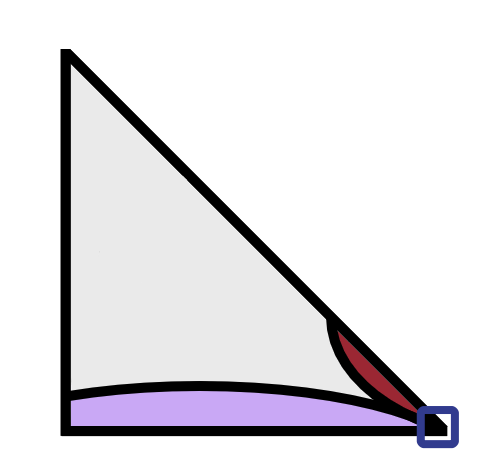} \\ \hline

\end{tabular}
\end{table}

\subsubsection{Degenerate Conic Sections}
\label{sec:edge-cases-internal-topology-degenerate-conic}

A conic section is typically a circle, ellipse, parabola, or hyperbola. However, it can also take degenerate forms, such as a single line, two parallel lines, two intersecting lines, or a single point.

If one of the conic sections, $\gamma_d^2 = \gamma_s^2$ or $\gamma_r^2 = \gamma_s^2$, is degenerate, it can cause issues when computing the invariant. Unless the conic is two parallel lines, our algorithm’s output may be affected.

For a single point, {\toolname} can be influenced because, although the point has no topological impact, the algorithm still detects its presence. For a single line, problems may arise if $\gamma_d = \gamma_s$ forms a line but $\gamma_s > |\gamma_d|$ holds on both sides, leading {\toolname} to misinterpret the topology. A similar issue can occur with the curve $\gamma_r = \gamma_s$. When the conic consists of two intersecting lines, our algorithm struggles because the intersection point alters the cell topology, yet it is neither a junction point nor an intersection with the cell boundary, making it difficult to handle directly.

We describe how we handle these cases in \cref{tab:edge-cases-internal-topology-intersecting-lines}. In the left column, we describe each case and how it is handled. In the right column, we provide an illustration.

\begin{table}[!ht]
\caption{Descriptions and images of edge cases that arise in internal cell topology computation from degenerate conic sections. In the first column we describe each case and how it is handled. In the second column we provide a visualization of the case.}
\label{tab:edge-cases-internal-topology-intersecting-lines}
\begin{tabular}{p{0.6\linewidth}|c} \hline
Case & Image \\ \hline
\para{One conic is a single point:} In this case, we apply a virtual perturbation so that the conic disappears. & \includegraphics[width=0.25\linewidth,valign=t]{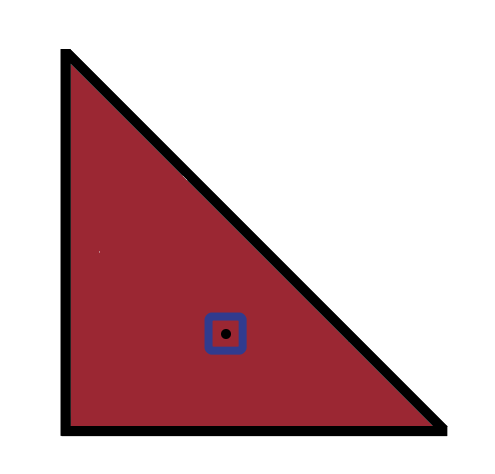} \\ \hline
\para{One conic does not separate two regions appropriately:} In particular, if the conic $\gamma_d = \gamma_s$ has $\gamma_s > |\gamma_d|$ on both sides, or if the conic $\gamma_r = \gamma_s$ has $\gamma_s > |\gamma_r|$ on both sides, we ignore the conic. & \includegraphics[width=0.25\linewidth,valign=t]{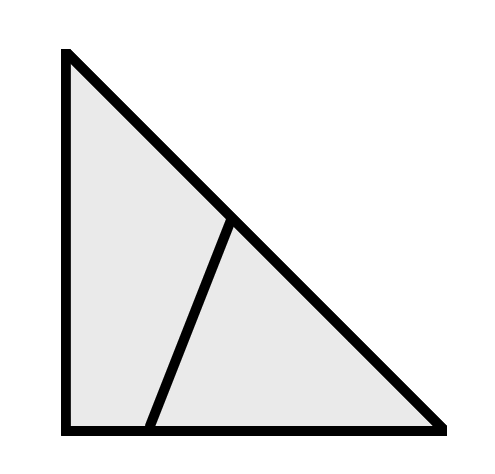} \\ \hline
\para{One conic consists of two intersecting lines. The intersection point does not lie on the boundary of two regions:} If the intersection point does not divide regions of different classifications, we ignore the intersection point. In future cases, assume that any conic that is two intersecting lines lies on the boundary between different regions. & \includegraphics[width=0.25\linewidth,valign=t]{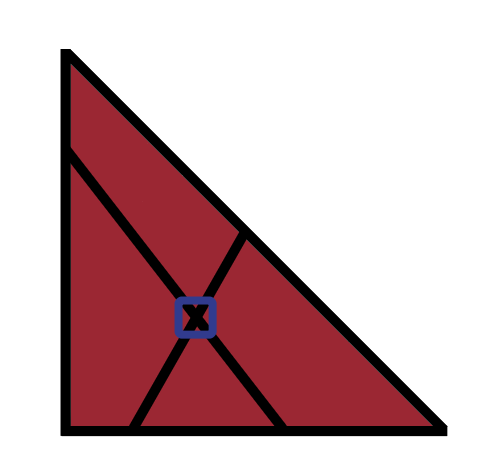} \\ \hline

\para{One conic consists of two intersecting lines. The intersection point lies on the boundary of two regions:} In this case, we track the location of the intersection point relative to any edge intersections and junction points in the invariant. & \includegraphics[width=0.25\linewidth,valign=t]{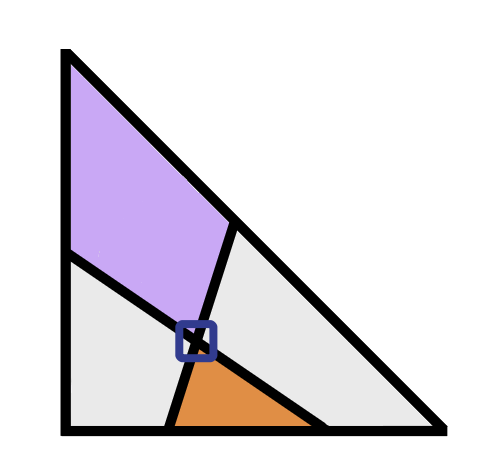} \\ \hline

\para{One conic consists of two intersecting lines. The intersection point is on an edge:} In this case, we track the location of the intersection point relative to any edge intersections and junction points in our invariant, noting the edge on which it occurs. & \includegraphics[width=0.25\linewidth,valign=t]{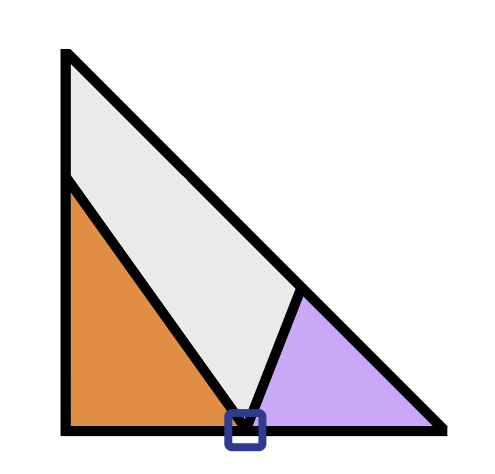} \\ \hline

\para{One conic consists of two intersecting lines. The intersection point is on a vertex:} In this case, we track the location of the intersection point relative to any edge intersections and junction points in our invariant, noting the vertex on which it occurs. & \includegraphics[width=0.25\linewidth,valign=t]{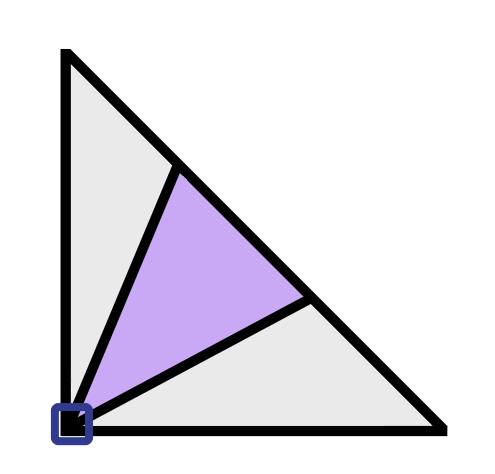} \\ \hline

\para{Both conics, each composed of two intersecting lines that share the same point of intersection:} In this case, we do not track any junction point. We handle each conic separately according to the previous cases. & \includegraphics[width=0.25\linewidth,valign=t]{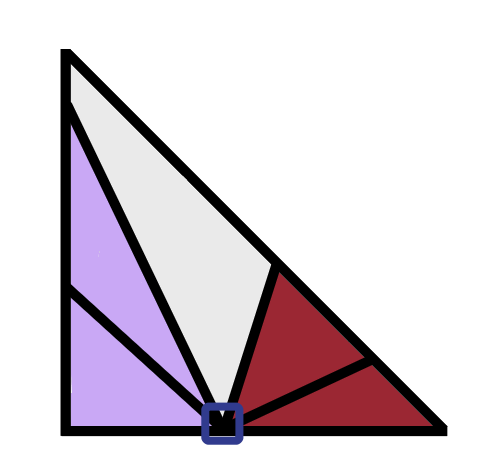} \\ \hline

\end{tabular}
\end{table}

\subsection{Numerical Precision}
\label{sec:edge-cases-numerical-precision}

The finite precision of floating-point values can cause difficulties for TFZ in certain cases. To mitigate these issues, we introduce several fixes. Admittedly, TFZ may still make mistakes in extreme scenarios, in which case it might be better implemented using integer representations of floating-point numbers.

\para{Close values.}~We consider two numbers $x$ and $y$ to be equal if $|x-y| < 10^{-10}\max(|x|,|y|)$. Similarly, we consider $x$ to be greater than $y$ if $x-y > 10^{-10}\max(|x|,|y|)$ and consider $x$ to be less than $y$ if $x-y < -10^{-10}\max(|x|,|y|)$.

When computing the classification according to the eigenvector manifold, we consider $\gamma_r$ to be equal to zero if $|\gamma_r| < 10^{-10}$. Similarly, if the number $u$ is the dot product of two unit vectors, the output of a trigonometric equation, or the $x$ or $y$ position of a point in the unit square, we consider $u$ to be equal to zero if $|u| < 10^{-10}$.

\para{Tensor normalization.}~When computing the topology of a cell $\sigma$ with tensors $T_1$, $T_2$ and $T_3$ at the vertices, we compute:
\[ x = \min\{ |y| : y \text{ is an entry of } T_1, T_2 \text{ or } T_3 \text{ and } |y| > 10^{-10} \}.\]
Set $T_1' \gets \frac{1}{|x|}T_1$, $T_2' \gets \frac{1}{|x|}T_2'$ and $T_3 \gets \frac{1}{|x|}T_3'$. We then proceed to compute the topology of $\sigma$ using $T_1'$, $T_2'$, and $T_3'$.

\para{Quadratic formula stability.} We have found that the quadratic formula can become unstable with very large or small values. Suppose that we are computing the roots of $ax^2+bx+c$. First, we compute $k = \max(|a|,|b|,|c|)$. If $k = 0$, then there are infinitely many solutions. Otherwise, we set $a' \gets \frac{a}{k}$, $b' \gets \frac{b}{k}$ and $c' \gets \frac{c}{k}$. Then, if any of $a'$, $b'$, or $c'$ has a magnitude less than $10^{-10}$, we set it equal to zero. We then compute the roots with the standard quadratic formula using $a'$, $b'$ and $c'$.

When computing the discriminant $(b')^2-4a'c'$, if $\frac{|(b')^2-4a'c'|}{(|a'|+|b'|+|c'|)^2} < 10^{-10}$ (i.e., the discriminant is very small compared to the largest coefficient) then we treat it as being equal to zero.

\para{Relative error bound.} After decompressing the data, it is possible that some entries of a tensor will be very large or small in magnitude compared to one another. For example, it is possible one entry will be $10^{-8}$ while another is $10^{-2}$. Such variability can exacerbate numerical precision issues. To resolve this issue, we enforce a relative error bound of $20$ on all entries of a given tensor. That is, if $x$ is an entry of a tensor, and $x'$ is a guess for $x$, if $x' \not = 0$ and $x \not = 0$, we require that $\frac{|x|}{|x'|} < 20$ and $\frac{|x'|}{|x|} < 20$.

%% file: appendix-lemma-proofs.tex
\section{Proofs of Theoretical Results}
\label{appendix:lemma-proofs}

In this section, we prove most of the lemmas from the background section. The proofs of \cref{lemma:eigenvector-topology} and \cref{lemma:eigenvalue-topology} are substantially more involved, and are therefore given separately in \cref{appendix:partition-correctness}.
We begin with several supporting lemmas before proving the main lemmas from \cref{sec:method}. With a slight abuse of notation, we write $M = 0$ for a matrix $M$ if all of its entries are zero. 

\begin{lemma}
Let $M$ be a symmetric matrix. Then $M$ is degenerate if and only if $D(M) = 0$.
\end{lemma}

\begin{proof}
According to \cite{zhang2008asymmetric}, the eigenvalues of $M$ are $\gamma_d \pm \gamma_s$. Thus, $M$ is degenerate if and only if $\gamma_s = 0$. Following \cref{eqn:symmetric-part}, we have
\[ D(M) = \gamma_s \begin{pmatrix}\cos(\theta) & \quad\sin(\theta) \\ \sin(\theta) & -\cos(\theta) \end{pmatrix}.\]
Thus, $D(M) = 0$ if and only if $\gamma_s = 0$, which occurs if and only if $M$ is degenerate.
\end{proof}

\begin{lemma}
Let $M_1$ and $M_2$ be two symmetric matrices such that $D(M_1) \not = 0$ and $D(M_2) \not = 0$. Then $l_{1,2} = 0$ if and only if there exists some $k \in \R$ such that $D(M_1) = kD(M_2)$.
\label{lemma:zero-k}
\end{lemma}

\begin{proof}
Denote the entries of $D(M_1)$ by
\[ D(M_1) = \begin{pmatrix} \Delta_1 & F_1 \\ F_1 & -\Delta_1\end{pmatrix} \]
Denote the entries of $D(M_2)$ similarly.

First, suppose that $D(M_1) = kD(M_2)$. Then $\Delta_1 = k\Delta_2$ and $F_1 = kF_2$. Thus, $F_2\Delta_1 - F_1\Delta_2 = F_2(k\Delta_2) - (kF_2)\Delta_2 = 0$, so $l_{1,2} = 0$.

Now suppose that $l_{1,2} = 0$. Then $F_2\Delta_1 - F_1\Delta_2 = 0$. We check two cases:

\noindent\underline{Case 1: $\Delta_2 \not = 0$}. Set $k = \frac{\Delta_1}{\Delta_2}$. Notice that $F_1\Delta_2 = F_2\Delta_1$ meaning that $F_1 = \frac{\Delta_1}{\Delta_2}F_2 = kF_2$. And clearly, $\Delta_1 = \frac{\Delta_1}{\Delta_2}\Delta_2 = k\Delta_2$. Thus, $D(M_1) = kD(M_2)$

\noindent\underline{Case 2: $\Delta_2 = 0$}. Since $D(M_2) \not = 0$ and $\Delta_2 = 0$, we must have $F_2 \not = 0$. Thus, set $k = \frac{F_1}{F_2}$ and proceed similarly to case 1.
\end{proof}

\begin{lemma}
Suppose that $x,y \in \R$ and $x',y' \in \R$ are respectively guesses for $x$ and $y$. Let $\xi$ be an error bound such that $|x - x'| \leq \xi$ and $|y - y'| \leq \xi$. Then
(a) $\left| \frac{x+y}{2} - \frac{x'+y'}{2} \right| \leq \xi$
(b) $\left| \sqrt{x^2+y^2} - \sqrt{(x')^2+(y')^2}\right| \leq \xi\sqrt{2}$
\label{lemma:deviation}
\end{lemma}

\begin{proof}
This can be verified with simple algebra.
\end{proof}

\lemmaSymmetricTheta*

\begin{proof}
Let $\gamma_{d_1}$, $\gamma_{s,1}$ and $\theta_1$ be the coefficients from decomposing $M_1$ according to \cref{eqn:symmetric-part}. Let $\gamma_{d,2}$, $\gamma_{s,2}$ and $\theta_2$ be the coefficients from decomposing $M_2$. The deviator of $M_1$ is:

\[ D(M_1) = \gamma_{s,1} \begin{pmatrix} \cos(\theta_1) & \quad\sin(\theta_1) \\ \sin(\theta_1) & -\cos(\theta_1) \end{pmatrix} \]

Let $\Delta_1$ and $F_1$ be the entries of $D(M_1)$. Then $\Delta_1 = \gamma_{s,1}\cos(\theta_1)$ and $F_1 = \gamma_{d,1}\sin(\theta_1)$. Similarly, we have $\Delta_2 = \gamma_{s,2}\cos(\theta_2)$ and $F_2 = \gamma_{s,2}\sin(\theta_2)$.

Now observe that:
\begin{align*}
l_{1,2} &= \Delta_2F_1 - \Delta_1F_2 \\
&= (\gamma_{s,2}\cos(\theta_2))(\gamma_{s,1}\sin(\theta_1)) - (\gamma_{s,1}\cos(\theta_1))(\gamma_{s,2}\sin(\theta_2)) \\
&= \gamma_{s,1}\gamma_{s,2}(\cos(\theta_2)\sin(\theta_1) - \cos(\theta_1)\sin(\theta_2))
\end{align*}
By definition, $\gamma_{s,1}$ and $\gamma_{s,2}$ are both nonnegative. By assumption, both deviators are nonzero, meaning that $\gamma_{s,1} \not = 0$ and $\gamma_{s,2} \not = 0$. Thus, the sign of $l_{1,2}$ only depends on $\theta_1$ and $\theta_2$.
\end{proof}

\lemmaSymmetricZero*

\begin{proof}
Denote the cell as $\sigma$. We first show (a). Suppose that there exists some $x \in \sigma$ with $x \not = x_1$ but $D(x) = 0$. There exists $t_1,t_2,t_3 \in [0,1]$ such that $x = t_1f(x_1)+t_2f(x_2)+t_3f(x_3)$. Because the deviator is a linear operator, $D(f(x)) = t_1D(f(x_1)) + t_2D(f(x_2)) + t_3D(f(x_3))$. 

Since $D(f(x)) = 0$, and $D(f(x_1)) = 0$, but $D(f(x_2)) \not = 0$, $x$ cannot lie on the edge between $x_1$ and $x_2$. Similarly, $x$ cannot lie on the edge between $x_1$ and $x_3$. Thus, $t_2 \not = 0$ and $t_3 \not = 0$.

Because $D(f(x)) = t_2D(f(x_2)) + t_3D(f(x_3)) = 0$, it follows that $D(f(x_2)) = \frac{-t_3}{t_2}D(f(x_3))$. So there exists $k$ such that $D(f(x_2)) = kD(f(x_3))$. By \cref{lemma:zero-k}, this implies that $l_{2,3} = 0$, proving (a) (i).

For (a) (ii), \cref{lemma:zero-k} tells us that if $l_{2,3} = 0$ there exists $k$ such that $D(f(x_2)) = kD(f(x_3))$. We just showed that if there exists some internal point $x \in C$ with $D(x) = 0$, then $D(x_2) = kD(x_3)$ with $k < 0$. 

To finish (a) (ii), we must show that if $k < 0$ then there exists some $y \not = x_1$ such that $D(f(y)) = 0$. Suppose that $k < 0$ and let $t_1 = 0$, let $t_2 = \frac{1}{1-k}$ and $t_3 = \frac{-k}{1-k}$. Then $t_1,t_2,t_3 \in (0,1)$ and $t_1+t_2+t_3 = 1$. Let $y = t_1x_1 + t_2x_2 + t_3x_3$, so $y \in \sigma$. Similar to before, $D(f(x)) = t_2D(f(x_2)) + t_3D(f(x_3))$.

Notice that $t_2D(f(x_2)) + t_3D(f(x_3)) = \frac{1}{1-k}D(f(x_2)) + \frac{-k}{1-k}D(f(x_3)) = \frac{1}{1-k}(D(f(x_2)) - kD(f(x_3))) = 0$, so $y$ is degenerate. This finishes the proof for (a) (ii).

To prove (b), note that since the two degenerate vertices have a deviator equal to zero and the deviator is a linear operator, all tensors interpolated between them will also be degenerate. 
\end{proof}

\lemmaErrorBounds*

\begin{proof}
This follows from \cref{lemma:deviation}.
\end{proof}

\lemmaSignSwap*

\begin{proof}
Let $x > 0$ and $x' < 0$. Then we have $|x - x'| = x - x'$. Thus, $x - x' < \xi$, so $x' + \xi > x > 0$.

Since $(x + \xi) > x$, it follows that $|x - (x' + \xi)| = (x' + \xi) - x = \xi + (x' - x)$. Because $x' - x < 0$, it follows that $\xi + (x' - x) < \xi$. Putting these inequalities together yields $|x - (x' + \xi)| < \xi$.

The proof for the case where $x < 0$ is similar.
\end{proof}

\lemmaMagnitudeSwap*

\begin{proof}
We show that $|x - y''| \leq \xi$. The proof that $|y - x''| \leq \xi$ is similar. Recall that $||x|-|x'|| \leq |x-x'|$ so $||x|-|x'|| \leq \xi$. Similarly, $||y|-|y'|| \leq \xi$. We check two cases for $|x|$:

\noindent\underline{Case 1: $|y'| \leq |x|$:} Then $|x'| < |y'| \leq |x|$. Since $||x|-|x'|| \leq \xi$, it must follow that $||x|-|y'|| \leq \xi$.

\noindent\underline{Case 2: $|y'| > |x|$:} Then $|y| < |x| < |y'|$. Since $||y|-|y'|| \leq \xi$, it must follow that $||x|-|y'|| \leq \xi$.

In either case, it follows that $||x|-|y'|| \leq \xi$. If $x > 0$, then $x = |x|$ and  $y'' = |y'|$, so $|x - y''| \leq \xi$. If $x < 0$, the proof is similar.
\end{proof}

%% file: appendix-partition-correctness-assumptions.tex
\section{Correctness Proofs of Topological Invariant}
\label{appendix:partition-correctness}

In this section, we build toward the proofs of \cref{lemma:eigenvector-topology} and \cref{lemma:eigenvalue-topology}. We provide definitions and assumptions in \cref{appendix:partition-correctness-definitions}. We prove preliminary lemmas in \cref{appendix:partition-correctness-supporting-lemmas}. We prove \cref{lemma:eigenvector-topology} and \cref{lemma:eigenvalue-topology} in \cref{appendix:partition-correctness-main-results}. Finally, we prove how we handle edge cases in \cref{appendix:partition-correctness-edge-cases}. For the eigenvector partition, we ignore degenerate points of the dual eigenvector field, and only focus on partition regions.

\subsection{Definitions and Assumptions}
\label{appendix:partition-correctness-definitions}
In this section, we specify all of the notations and assumptions that we will be using for the remainder of the section.

\para{Function Definitions.} Let $\sigma \subset \R^2$ be a 2D triangular cell, and let $f:\sigma \rightarrow \T$ be a PL tensor field on $\sigma$. 
With an abuse of notation, define the function $\gamma_d:\sigma \rightarrow \R$ as the function that maps each point $p \in \sigma$ to the value of $\gamma_d$ at $p$. Define similar functions for $\gamma_r$, $\gamma_s$, and $\theta$. 

Let $c:\sigma \rightarrow \R$ and $s:\sigma \rightarrow \R$ be defined by
\[ c(p) = \gamma_d(p)\cos(\theta(p)) \quad\quad s(p) = \gamma_d(p)\sin(\theta(p)) \]
Then notice that 
\[ \gamma_s(p) = \sqrt{c(p)^2 + s(p)^2} \]
Since $f$ is PL, all of its coordinate functions are affine in $\sigma \subset \R^2$. Thus, $f$ can be extended to all of $\R^2$. All of the other functions that we have just defined can thus also be extended to $\R^2$.

\para{Subsets and Regions.} We first define some subsets of $\sigma$.

\noindent \underline{Eigenvalue partition:} Let
\begin{align*}
\DP &= \{p \in \sigma : \gamma_d > |\gamma_r| \text{ and } \gamma_d > \gamma_s\} \\ 
\DN &= \{p \in \sigma : -\gamma_d > |\gamma_r| \text{ and } -\gamma_d > \gamma_s \}  
\end{align*}
Define $\RP$, and $\RN$ similarly. Let $\SA$ be the region where $\gamma_s > |\gamma_d|$ and $\gamma_s > |\gamma_r|$.

\noindent \underline{Eigenvector partition:} Define $\RRP$ as the subset of $\sigma$ where $\gamma_r > \gamma_s$ and $\RRN$ as the subset where $-\gamma_r > \gamma_s$. Define $\SRP$ as the subset where $\gamma_s > \gamma_r > 0$ and $\SRN$ as the subset where $\gamma_s > -\gamma_r > 0$.

\noindent \underline{Other subsets:} While not part of either partition, let $\DDP$ be the subset of $\sigma$ where $\gamma_d > \gamma_s$ and $\DDN$ be the subset where $-\gamma_d > \gamma_s$. 
Notice that, by definition, 
\begin{itemize}[noitemsep]
\item $\DP \subset \DDP$
\item $\DN \subset \DDN$
\item $\RP \subset \RRP$
\item $\RN \subset \RRN$.
\end{itemize}
We define a \textit{region} of $\sigma$ as a connected subset of $\sigma$. We say that a region $R$ is of \textit{type} $\DP$ if $R \subset \DP$. Extend this defintition of type to the other subsets that we have defined.
We say that two regions \textit{border} each other if their boundaries intersect.

\para{Topologically Significant Intersections.} If $e$ is an edge of $\sigma$, and $p$ lies on the intersection of $e$ with the curve $\gamma_d = \gamma_s$, say that $p$ is \textit{topologically significant} if it lies on the boundary between $\DP$ and $\SA$. Define topological significance analogously for $-\gamma_d = \gamma_s$, $\gamma_r = \gamma_s$, and $-\gamma_r = \gamma_s$.

We say that an intersection between $\gamma_d^2 = \gamma_s^2$ and $e$ is topologically significant if it separates a region of type $\SA$ from a region of either type $\DP$ or $\DN$. Define topologically significant intersections analogously for $\gamma_r^2 = \gamma_s^2$.

\para{Edge Cases.} For now, we make the following assumptions in order to avoid edge cases. Some of these cases are avoided by {\toolname} using a symbolic perturbation (see~\cref{appendix:edge-cases}), and thus we can assume that they never occur. Other cases can occur. In \cref{appendix:partition-correctness-edge-cases}, we demonstrate that when such cases occur that our strategy still works. We assume that:
\begin{itemize}[noitemsep]
\item[(i)] $\gamma_d^2 = \gamma_s^2$ and $\gamma_r^2 = \gamma_s^2$ only intersect each other and each edge of $\sigma$ transversally. 
\item[(ii)] No junction points occur on the edges of $\sigma$. 
\item[(iii)] No two of the functions $|\gamma_d|$, $|\gamma_r|$, or $\gamma_s$ are exactly equal on all of $\R^2$.
\item[(iv)] The conic sections $\gamma_d^2 = \gamma_s^2$ and $\gamma_r^2 = \gamma_s^2$ are either empty or have infinitely many points.
\item[(v)] The curve $\gamma_d = \gamma_s$ will border one region of type $\DDP$ and another where $\gamma_s > |\gamma_d|$. Make similar assumptions about $-\gamma_d = \gamma_s$, $\gamma_r = \gamma_s$, and $-\gamma_r = \gamma_s$.
\item[(vi)] $|\gamma_d|$, $|\gamma_r|$ and $\gamma_s$ are never equal on a vertex of $\sigma$.
\item[(vii)] There are no vertices of $\sigma$ where $\gamma_r = 0$.
\item[(viii)] There are no points in $\sigma$ where $\gamma_d = \gamma_s = 0$ or $\gamma_r = \gamma_s = 0$.
\end{itemize}

\begin{figure}
\begin{overpic}[width=\linewidth]{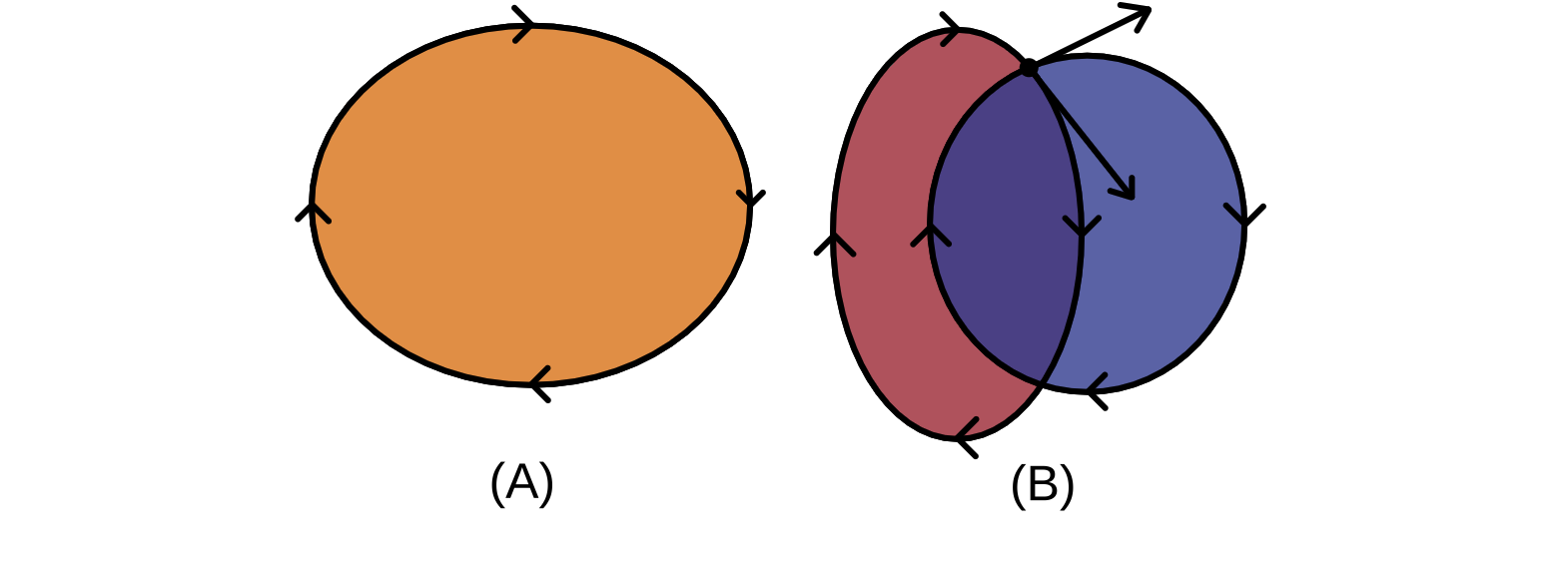}
\put(29, 22){interior}
\put(66,34){$p$}
\end{overpic}
\vspace{-6mm}
\caption{(A) A connected region with the interior shaded. Arrows denote clockwise orientation. (B) Two connected regions overlap. At point $p$, they intersect. Here, the curve bounding the red region enters the blue region. The curve bounding the blue region leaves the red region.}
\label{fig:proofs-orientation}
\vspace{-6mm}
\end{figure}

\para{Clockwise Orientation.} Suppose that $c$ is a closed curve that divides the plane into two regions. Suppose that one of them is labeled as the ``interior.'' Then we define \textit{clockwise orientation} or $c$ to be such that, when traveling along $c$, the ``interior'' of $c$ is on the right of $c$. We demonstrate this orientation in \cref{fig:proofs-orientation}(A).

If a curve $c$ divides the plane into two regions $R_1$ and $R_2$, and $c$ is said to be clockwise oriented, then label one of $R_1$ or $R_2$ to be the interior of $c$ consistent with the previous definition.

Suppose that $c_1$ and $c_2$ are two clockwise oriented curves. Let $p$ be a point where they intersect transversally. We say that $c_1$ \textit{enters} the interior of $c_2$ at $p$ if the oriented vector tangent to $c_1$ at $p$ points into the interior of $c_2$. In \cref{fig:proofs-orientation}(B), we illustrate the case where one curve enters another. At point $p$, the curve bounding the red region enters the blue region, while the curve bounding the blue region leaves the red region.

We assume that $\gamma_r = \gamma_s$ separates a region of type $\DDP$ from a region where $\gamma_s > |\gamma_d|$. In this case, we denote that $\DDP$ is the interior of $\gamma_d = \gamma_s$, and orient $\gamma_d = \gamma_s$ using clockwise orientation. We proceed similarly for the curves $-\gamma_d = \gamma_s$, $\gamma_r = \gamma_s$, and $-\gamma_r = \gamma_s$.

%% file: appendix-partition-correctness-support.tex
\subsection{Supporting Lemmas}
\label{appendix:partition-correctness-supporting-lemmas}
In order to prove \cref{lemma:eigenvector-topology} and \cref{lemma:eigenvalue-topology}, we first prove some intermediate results.
\begin{lemma}
The functions $\gamma_d$ and $\gamma_r$, $c$ and $s$ are affine over $\sigma \subset \R^2$.
\label{lemma:its-all-PL}
\end{lemma}
\begin{proof}
For $p \in \sigma$, recall that
\begin{itemize}[noitemsep]
\item $\gamma_d(p) = \frac{1}{2}(f(p)_{1,1}+f(p)_{2,2})$
\item $\gamma_r(p) = \frac{1}{2}(f(p)_{2,1}-f(p)_{1,2})$
\item $c(p) = \frac{1}{2}(f(p)_{1,1}-f(p)_{2,2})$
\item $s(p) = \frac{1}{2}(f(p)_{1,2}+f(p)_{2,1})$
\end{itemize}
Because each $f_{i,j}$ is affine, so are the functions above.
\end{proof}

\begin{lemma}
The function $\gamma_s$ is convex.
\label{lemma:s-is-convex}
\end{lemma}
\begin{proof}
Fix $p_1,p_2 \in \sigma$. Denote $\theta_1 := \theta(p_1)$ and $\theta_2 := \theta(p_2)$. Notice that
\begin{align*}
&\quad c(p_1)c(p_2)+s(p_1)s(p_2) \\
&= \gamma_s(p_1)\cos(\theta_1)\gamma_s(p_2)\cos(\theta_2)+\gamma_s(p_1)\sin(\theta_1)\gamma_s(p_2)\sin(\theta_2) \\
&= \gamma_s(p_1)\gamma_s(p_2)(\cos(\theta_1)\cos(\theta_2)+\sin(\theta_1)\sin(\theta_2)) \\
&= \gamma_s(p_1)\gamma_s(p_2)\cos(\theta_1-\theta_2) \\
&\leq \gamma_s(p_1)\gamma_s(p_2)
\end{align*}
Thus, $c(p_1)c(p_2) + s(p_1)s(p_2) \leq \gamma_s(p_1)\gamma_s(p_2)$. Using this fact, let $t \in [0,1]$ and observe that
\begin{align*}
& \quad \gamma_s(tp_1 + (1-t)p_2)^2 \\
&= c(tp_1+(1-t)p_2)^2 + s(tp_1+(1-t)p_2)^2 \\
&= (tc(p_1)+(1-t)c(p_2))^2 + (ts(p_1)+(1-t)s(p_2))^2 \\
&= t^2(c(p_1)^2+s(p_1)^2) + 2t(1-t)(c(p_1)c(p_2)+s(p_1)s(p_2)) \\
&\quad + (1-t)^2(c(p_2)^2+s(p_2)^2) \\
&= t^2\gamma_s(p_1)^2 + 2t(1-t)(c(p_1)c(p_2)+s(p_1)s(p_2)) \\
&\quad+ (1-t)^2\gamma_s(p_2)^2 \\
&\leq t^2\gamma_s(p_1)^2 + 2t(1-t)\gamma_s(p_1)\gamma_s(p_2) + (1-t)^2\gamma_s(p_2) \\
&= (t\gamma_s(p_1)+(1-t)\gamma_s(p_2))^2
\end{align*}
Thus, $\gamma_s(tp_1+(1-t)p_2)^2 \leq (t\gamma_s(p_1)+(1-t)\gamma_s(p_2))^2$. Because $\gamma_s$ is nonnegative, this implies that $\gamma_s(tp_1+(1-t)p_2) \leq t\gamma_s(p_1)+(1-t)\gamma_s(p_2)$. Thus, $\gamma_s$ is convex.
\end{proof}

\begin{lemma}
$\DP$, $\DN$, $\RP$, $\RN$, $\RRP$, $\RRN$, $\DDP$, and $\DDN$ are convex.
\label{lemma:regions-are-convex}
\end{lemma}
\begin{proof}
We provide a proof for $\DP$. The other region types follow similarly. Let $p_1, p_2 \in \DP$. Let $p_3$ lie on the segment between $p_1$ and $p_2$. Since $\gamma_d$ is affine, and $|\gamma_r|$ and $\gamma_s$ are convex, it must hold that $\gamma_d(p) > |\gamma_r(p)|$ and $\gamma_d(p) > \gamma_s(p)$. Thus, $p \in \DP$, so $\DP$ is convex.
\end{proof}
\begin{corollary}
Each of $\DP$, $\DN$, $\RP$, $\RN$, $\RRP$, $\RRN$, $\DDP$, and $\DDN$ have only one connected component each.
\label{corollary:one-connected-component}
\end{corollary}
Because each of the subsets in \cref{corollary:one-connected-component} have one connected component, we may refer to them hereafter as regions.

\begin{lemma}
The following pairs of regions cannot border each other:
\begin{itemize}[noitemsep]
\item[(i)] $\DDP$ cannot border $\DDN$
\item[(ii)] $\RRP$ cannot border $\RRN$
\item[(iii)] $\DP$ cannot border $\DN$
\item[(iv)] $\RP$ cannot border $\RN$
\end{itemize}
\label{lemma:illegal-intersections}
\end{lemma}
\begin{proof}
We show (i) and (iii), where (ii) and (iv) follow similarly. Suppose, to the contrary, that $\DDP$ bordered $\DDN$. Let $x$ be a point that separates them. Since $\gamma_d > 0$ one one side of $x$, and $\gamma_d < 0$ on the other, it must follow that $\gamma_d(x) = 0$. Since $|\gamma_d| > \gamma_s$ for all points nearby to $x$, it must follow that $|\gamma_d(x)| \geq \gamma_s(x)$, implying that $\gamma_d(x) = \gamma_s(x) = 0$. However, we previously assumed that this is impossible. This gives (i).

Because (i) is true and $\DP \subset \DDP$ and $\DN \subset \DDN$, (iii) must follow.
\end{proof}

\begin{lemma}
$\RRP$ cannot border $\SRN$ and $\RRN$ cannot border $\SRP$.
\label{lemma:illegal-intersections-2}
\end{lemma}
\begin{proof}
The proof is similar to that of \cref{lemma:illegal-intersections}.
\end{proof}

\begin{lemma}
The following statements are true:
\begin{itemize}[noitemsep]
\item[(1)] If $\gamma_d^2 = \gamma_s^2$ is an ellipse contained entirely within $\sigma$, then either $\DDP = \emptyset$ or $\DDN = \emptyset$.
\item[(2)] If $\gamma_d = \gamma_s$ is an ellipse contained within $\sigma$, then $\DDN = \emptyset$. If $-\gamma_d = \gamma_s$ is an ellipse contained within $\sigma$, then $\DDP = \emptyset$.
\end{itemize}
Analogous claims are true about $\gamma_r^2 = \gamma_s^2$.
\label{lemma:internal-ellipse-other}
\end{lemma}
\begin{proof}
We prove the claims for $\gamma_d^2 = \gamma_s^2$. Notice that because there are no points where $\gamma_d = \gamma_s = 0$, the curves $\gamma_d = \gamma_s$ and $-\gamma_d = \gamma_s$ must be disjoint.

We first show (1). The curve $\gamma_d^2 = \gamma_s^2$ must border $\DDP$ or $\DDN$, so one of those sets is not empty. Without loss of generality, suppose that $\DDP \not = \emptyset$. Since $\DDP \not = \emptyset$, the region $\DDP$ must have, as its boundary, the curve $\gamma_d = \gamma_s$. So the curve $\gamma_d = \gamma_s$ is not empty.

Notice that $\gamma_d^2 = \gamma_s^2$ will be the union of the curves $\gamma_d = \gamma_s$ and $-\gamma_d = \gamma_s$. Since $\gamma_d^2 = \gamma_s^2$ is an ellipse, it has one connected component. Since $-\gamma_d = \gamma_s$ cannot intersect $\gamma_d = \gamma_s$, and $\gamma_d = \gamma_s$ is not empty, it must follow that $-\gamma_d = \gamma_s$ is empty.

If $\DDN \not = \emptyset$, then it would need to be bounded by the non-empty curve $-\gamma_d = \gamma_s$. Thus, $\DDN = \emptyset$. This gives (1).

We now show (2). Without loss of generality, suppose that $\gamma_d = \gamma_s$ is an ellipse contained within $\sigma$. The conic section $\gamma_d^2 = \gamma_s^2$ is a conic section that contains the curve $\gamma_d = \gamma_s$ and the curve $-\gamma_d = \gamma_s$. Thus, $\gamma_d^2 = \gamma_s^2$ contains an ellipse. The only possible conic section that $\gamma_d^2 = \gamma_s^2$ could be, then, is an ellipse. Invoking (1) gives $\DDN = \emptyset$.
\end{proof}

\begin{lemma}
If the conic section $\gamma_d^2 = \gamma_s^2$ is an ellipse contained entirely within $\sigma$, then the interior of $\gamma_d^2 = \gamma_s^2$ must intersect $\RRP$ or $\RRN$.
The analogous claim for $\gamma_r^2 = \gamma_s^2$ is also true.
\label{lemma:internal-ellipse-is-both}
\end{lemma}
\begin{proof}
Extend $f$ to $\R^2$. The conic section $\gamma_d^2 = \gamma_s^2$ is the union of the curve $\gamma_d = \gamma_s$ and $-\gamma_d = \gamma_s$. Since there is no point where $\gamma_d = \gamma_s = 0$, these curves do not intersect. Since $\gamma_d^2 = \gamma_s^2$ has only one connected component, it must be equal to $\gamma_d = \gamma_s$ or $-\gamma_d = \gamma_s$. Without loss of generality, assume that the conic section $\gamma_d^2 = \gamma_s^2$ is equal to the curve $\gamma_d = \gamma_s$.

By assumption, the curve $\gamma_d = \gamma_s$ separates a region of type $\DDP$ from a region where $\gamma_s > |\gamma_d|$. Because $\DDP$ is convex, and $\gamma_d = \gamma_s$ is an ellipse, $\DDP$ must lie in the interior of the ellipse.

We claim that the gradient directions of $c$ and $s$ are linearly independent. Suppose, to the contrary, that they were not. Then there exists some unit vector $v$ perpendicular to the gradients of $c$ and $s$. Since $c$ and $s$ are affine, they both have constant gradients. Thus, $c$ and $s$ are constant in direction $v$, and thus so is $\gamma_s$.

Let $p \in \DDP$. Let $l$ be the line through $p$ in direction $v$. Since $\gamma_d$ is affine, we can move along line $l$ in one direction $w$ (either $w = v$ or $w = -v$) so that $\gamma_d$ does not decrease. Also, $\gamma_s$ is constant along $l$. Thus, the ray in direction $w$ rooted at $p$ satisfies $\gamma_d > \gamma_s$, and the entire ray belongs to $\DDP$. This violates the assumption that $\DDP$ is contained within the ellipse.

Thus, the gradient directions of $c$ and $s$ are independent. So the curves $c(x,y) = 0$ and $s(x,y) = 0$ are two lines that are not parallel. Thus, they will intersect, and there is a point $q$ where $c(q) = s(q) = \gamma_s(q) = 0$. By \cref{lemma:internal-ellipse-other}, $\DDN = \emptyset$ so $\gamma_d(q) \geq 0$. By assumption $,\gamma_d(q) \not = \gamma_s(q) = 0$. Hence, $\gamma_d(q) > \gamma_s(q) = 0$, so $q \in \DDP$ and $q$ lies in the interior of the ellipse $\gamma_d^2 = \gamma_s^2$.

Trivially, $|\gamma_r(q)| > \gamma_s(q)$ so either $q \in \RRP$ or $q \in \RRN$. In either case, the claim is proven.
\end{proof}

\begin{lemma}
Let $e$ be a line segment. If $\DDP \cap e \not = \emptyset$ and $\DDN \cap e \not = \emptyset$, then one vertex of $e$ is of type $\DDP$ and the other is of type $\DDN$.
The analogous claim for $\RRP$ and $\RRN$ is also true.
\label{lemma:opposite-vertices}
\end{lemma}
\begin{proof}
Extend $f$ to $\R^2$. Let $\phi:\R \rightarrow \R^2$ parametrize the line containing $e$, such that $\phi([0,1]) = e$. One can verify that the functions $(\gamma_d \circ \phi)^2$ and $(\gamma_s \circ \phi)^2$ are quadratic functions $[0,1] \rightarrow \R$.

Since $\RRP \cap e$ cannot border $\RRN \cap e$, there must be some interval $I := (t_1,t_2) \in [0,1]$ such that $\phi(I)$ lies between $\RRP \cap e$ and $\RRN \cap e$, and $\gamma_s \circ \phi > |\gamma_d \circ \phi|$ on all of $I$. Further, $(\gamma_d \circ \phi)^2$ and $(\gamma_s \circ \phi)^2$ must intersect at $t_1$ and $t_2$. Since those functions are quadratic, they can only intersect at most twice. 

Thus, outside of $I$, there will be no other points where $|\gamma_d| = \gamma_s$. This implies that the only region of $e$ where $\gamma_s > |\gamma_d|$ that borders $\DDP \cap e$ or $\DDN \cap e$ is $\phi(I)$, which lies between $\DDP$ and $\DDN$. So one vertex of $e$ must lie within $\DDP$ and the other within $\DDN$.
\end{proof}

\begin{figure}
\begin{overpic}[width=\linewidth]{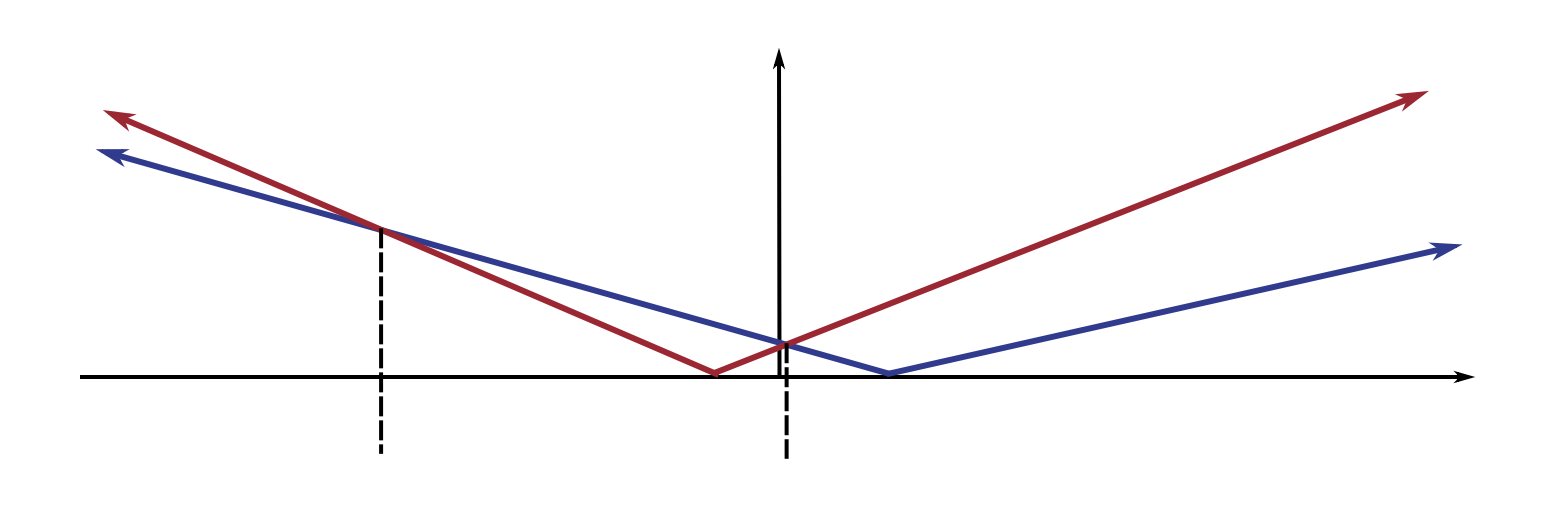}

\put (95.5,8.5){$x$}
\put (51,30){$y$}
\put (59, 24){$y = |\gamma_d(x)|$}
\put (75, 19){$y = |\gamma_r(x)|$}
\put (5, 5){$-\gamma_d > |\gamma_r|$}
\put (28, 5){$-\gamma_r > |\gamma_d|$}
\put (65, 5){$\gamma_d > |\gamma_r|$}

\end{overpic}
\vspace{-6mm}
\caption{The two functions $|\gamma_d|$ and $\gamma_r|$ can intersect at most twice. Further, every region where $|\gamma_d| > |\gamma_r|$ will border every region where $|\gamma_r| > |\gamma_d$}
\label{fig:proofs-absolute-value}
\vspace{-6mm}
\end{figure}

\begin{figure}[!h]

\begin{overpic}[width=\linewidth]{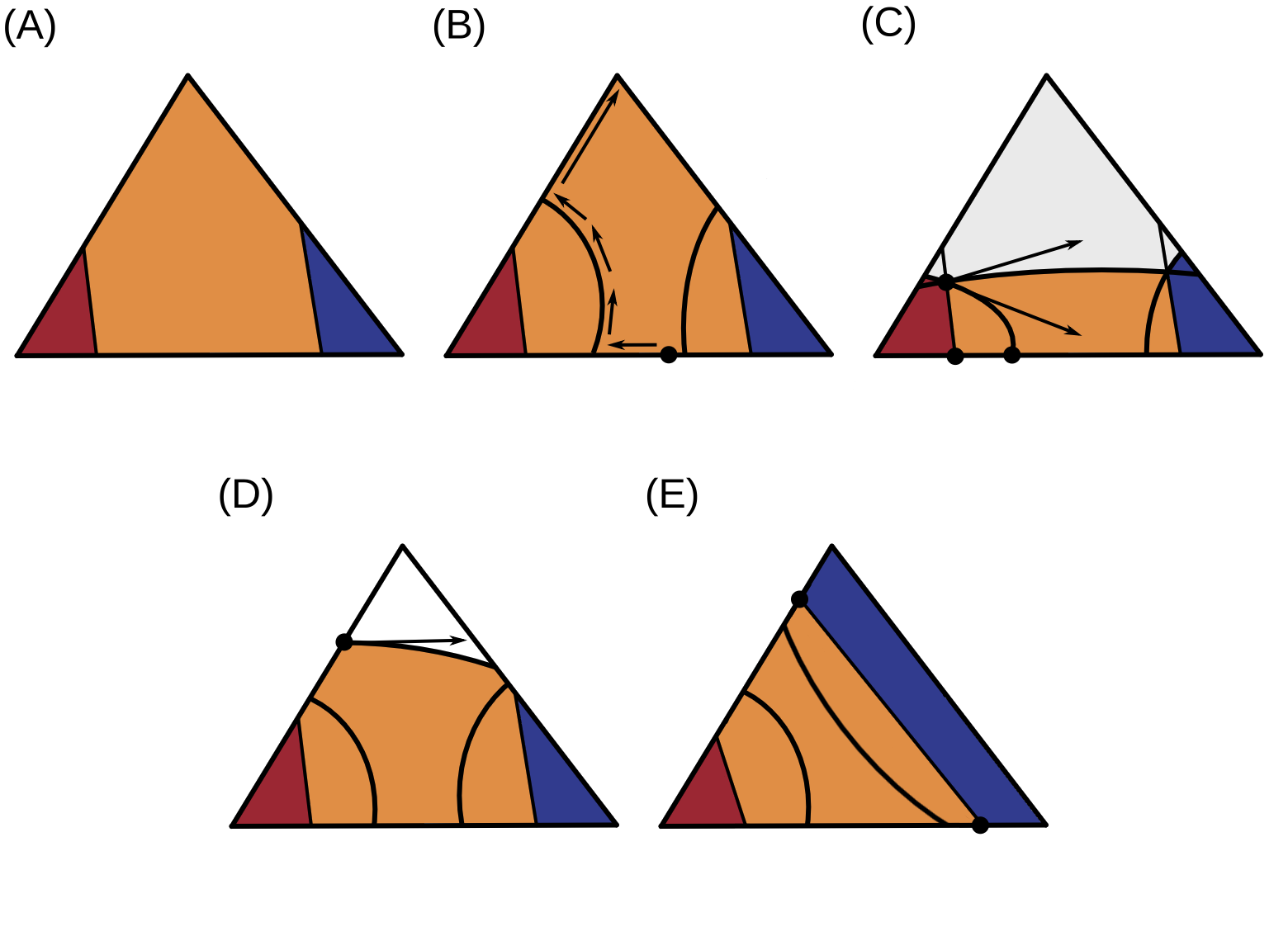}

\put (1, 43) {\smaller[2]{$v_1$}}
\put (15, 43) {\smaller[2]{$e$}}
\put (29, 43) {\smaller[2]{$v_2$}}
\put (6, 60) {\smaller[2]{$e_1$}}
\put (13, 69) {\smaller[2]{$v_3$}}
\put (22, 60) {\smaller[2]{$e_2$}}
\put (7.5, 48) {\smaller[2]{$\gamma_r = \gamma_d$}}
\put (10, 53) {\smaller[2]{$-\gamma_r = \gamma_d$}}

\put (51, 48) {\smaller[2]{$p$}}
\put (40, 43) {\smaller[2]{$\gamma_r = \gamma_s$}}

\put (74, 53) {\smaller[2]{$q$}}
\put (78, 43) {\smaller[2]{$r$}}
\put (74, 43) {\smaller[2]{$q'$}}

\put (24, 24) {\smaller[2]{$q$}}

\put (59,27) {\smaller[2]{$q$}}
\put (75.5,6) {\smaller[2]{$q'$}}

\end{overpic}
\vspace{-8mm}
\caption{Sample eigenvector partitions of $\sigma$ as used in \cref{lemma:edge-intersections}. The regions are Orange: $\DP$. Red: $\RP$. Blue: $\RN$. White $\SA$. (A) We label the edges and vertices of $\sigma$, along with the line $\gamma_r = \gamma_d$, which separates $\DP$ from $\RP$, and the line $-\gamma_r = \gamma_d$, which separates $\DP$ from $\RN$. (B) We demonstrate the path from $p$ up to $v_3$. The path goes towards $v_2$, then up the curve $\gamma_r = \gamma_s$ until hitting $e_1$. It then continues from $e_1$ to $v_3$. In this case, $v_3 \in \DP$. (C) Here, $q$ is a junction point between $\gamma_r = \gamma_s$ and $\gamma_d = \gamma_s$. We label the point $r$ where $\gamma_r = \gamma_s$ intersects $e$. We also label the oriented tangent vectors of each curve at point $p$. Here $\gamma_r = \gamma_s$ enters $\DP$. (D) Here $q$ is a topologically significant intersection point between $\gamma_d = \gamma_s$ and $e_1$. We show the oriented tangent vector of $\gamma_d = \gamma_s$ at $q$, where it points into $\sigma$. (E) Here $q$ separates $\DP$ from $\RN$, and lies on the line $-\gamma_r = \gamma_d$. We also label the point $q'$ used in the proof.}
\label{fig:edge-intersections}
\vspace{-4mm}
\end{figure}

\begin{lemma}
Let $e$ be a line segment. Let $E_1 = \{p \in e : \gamma_d(p) > |\gamma_r(p)|\}$ and $E_2 = \{p \in e : \gamma_r(p) > |\gamma_d(p)|\}$. If $E_1$ and $E_2$ are not empty, then they will border each other. 
The analogous claims using $-\gamma_d$ or $-\gamma_r$ are also true.
\label{lemma:mini-region-border}
\label{lemma:not-all-four-support}
\end{lemma}
\begin{proof}
Let $\phi:[0,1] \rightarrow e$ parametrize $e$. Notice that the functions $|\gamma_d \circ \phi|$ and $|\gamma_r \circ \phi|$ are the absolute values of affine functions on $[0,1]$. Thus, every region where $|\gamma_d \circ \phi| \geq |\gamma_r \circ \phi|$ will border every region where $|\gamma_d \circ \phi| \leq |\gamma_r \circ \phi|$. We illustrate an example in \cref{fig:proofs-absolute-value}.
\end{proof}

\begin{corollary}
Let $e$ be a line segment. Then at least one of the following sets must be empty
\begin{itemize}[noitemsep]
\item $\{p \in e : \gamma_d(p) > |\gamma_r(p)|\}$
\item $\{p \in e : -\gamma_d(p) > |\gamma_r(p)|\}$
\item $\{p \in e : \gamma_r(p) > |\gamma_d(p)|\}$
\item $\{p \in e : -\gamma_r(p) > |\gamma_d(p)|\}$
\end{itemize}
\end{corollary}

\begin{corollary}
The regions $\DP$, $\DN$, $\RP$ and $\RN$ cannot all intersect a segment $e$.
\label{lemma:not-all-four}
\end{corollary}
\begin{lemma}
Let $e$ be an edge of $\sigma$ with vertices $v_1$ and $v_2$. Suppose that $v_1$ comes after $v_2$ when traversing $\sigma$ clockwise. Let $v_3$ be the remaining vertex of $\sigma$. Let $e_1$ be the edge connecting $v_1$ to $v_3$, and let $e_2$ be the edge connecting $v_2$ to $v_3$. (See \cref{fig:edge-intersections}(A) for a diagram.)

Suppose that $\DP$ is not empty, yet does not contain either $v_1$ or $v_2$, and the curve $\gamma_d = \gamma_s$ does not have any topologically significant intersections with $e$. Then $\DP$ will intersect edge $e$ if and only if the following are true:

\begin{itemize}[noitemsep]
\item[(a)] One of $v_1$ or $v_2$ is of type $\RP$ and the other is $\RN$. (assume, w.l.o.g. that $v_1 \in \RP$).
\item[(b)] $\gamma_r^2 = \gamma_s^2$ has no topologically significant intersections with $e$
\item[(c)] At least one of the following is true:
\begin{itemize}[noitemsep]
\item[(i)] The curve $\gamma_d = \gamma_s$ never intersects $e_1$ or $\gamma_r = \gamma_s$ and $v_3$ is of type $\DP$.
\item[(ii)] The curve $\gamma_d = \gamma_s$ never intersects $e_2$ or $-\gamma_r = \gamma_s$ and $v_3$ is of type $\DP$.
\item[(iii)] There is exactly one junction point $p$ between $\gamma_d = \gamma_s$ and $\gamma_r = \gamma_s$. At that junction point, $\gamma_r = \gamma_s$ enters $\DDP$.
\item[(iv)] There is exactly one junction point $p$ between $\gamma_d = \gamma_s$ and $-\gamma_r = \gamma_s$. At that junction point, $-\gamma_r = \gamma_s$ leaves $\DDP$.
\item[(v)] The curve $\gamma_d = \gamma_s$ has exactly one topologically significant intersection with $e_1$, where it enters $\sigma$. $\gamma_d = \gamma_s$ never intersects $\gamma_r = \gamma_s$.
\item[(vi)] The curve $\gamma_d = \gamma_s$ has exactly one topologically significant intersection with $e_2$ where it leaves $\sigma$. $\gamma_d = \gamma_s$ never intersects $-\gamma_r = \gamma_s$.
\item[(vii)] The curve $\gamma_d = \gamma_s$ has no junction points or topologically significant intersections with any edge of $\sigma$. $v_3$ is of type $R_-$.
\end{itemize}
\end{itemize}
\label{lemma:edge-intersections}
\end{lemma}
\begin{proof}
Suppose that $D_+$ borders $e$ but $\gamma_d = \gamma_s$ has no topologically significant intersections with $e$, and neither vertex of $e$ is $\DP$. We first show that if $\DP$ borders $e$, then (a) (b) and (c) must be true. 

Because neither vertex is of type $\DP$, the region $\DP \cap e$ must be an interval in the interior of $e$ bordered by a different region on either side. $\DP \cap e$ cannot border $\DN \cap e$. Because $\gamma_d = \gamma_s$ has no topologically significant intersections with $e$, $\DP \cap e$ cannot border $\SA \cap e$. Thus, $\DP \cap e$ can only border $\RP \cap e$ and $\RN \cap e$.

Because $\RP$ and $\RN$ are convex, so are $\RP \cap e$ and $\RN \cap e$. Thus, $\RP \cap e$ and $\RN \cap e$ have one connected component each, and $\DP \cap e$ will border $\RP$ one one side and $\RN$ on the other side.

Since $\RP$ and $\RN$ both border $e$, by \cref{lemma:opposite-vertices}, one vertex must lie in $\RRP$ and the other in $\RRN$. Without loss of generality, suppose that $v_1 \in \RRP$. By assumption, neither vertex is type $\DP$. By \cref{lemma:not-all-four}, neither vertex is type $\DN$. Also, neither vertex can be type $\SA$ because $|\gamma_r| > |\gamma_s|$ at both vertices. Thus, we must have that $v_1$ is type $\RP$ and $v_2$ is type $\RN$. This gives (a). 

Because (a) is true, both $\RP \cap e$ and $\RN \cap e$ must only border $\DP$, meaning that $\gamma_r^2 = \gamma_s^2$ will have no topologically significant intersections with $e$, giving (b).

We now show (c). Because $\gamma_r(v_1) > 0$ and $\gamma_r(v_2) < 0$, there must be a point $p \in e$ such that $\gamma_r(p) = 0$. Since $e$ only intersects $R_+$, $D_+$ and $D_-$, it must hold that $p \in \RP$. Because $\gamma_s(p) \geq 0$, it follows that $\gamma_r(p) \leq \gamma_s(p)$. However, $\gamma_r(v_1) > \gamma_s(v_1)$. Thus, if we move towards $v_1$ from $p$, we must encounter the curve $\gamma_r = \gamma_s$. 

Thus, $\gamma_r = \gamma_s$ intersects edge $e$. Using similar reasoning, $-\gamma_r = \gamma_s$ must intersect $e$. Since the conic $\gamma_r^2 = \gamma_s^2$ cannot intersect $e$ more than twice, the curves $\gamma_r = \gamma_s$ and $-\gamma_r = \gamma_s$ intersect $e$ once each.

Because $\gamma_r = \gamma_s$ transversally intersects the boundary of $\sigma$ at edge $e$, and $\partial \sigma$ is a closed loop, the curve $\gamma_r = \gamma_s$ must intersect the boundary of $\sigma$ somewhere else. Because $\gamma_r = \gamma_s$ only intersects edge $e$ once, it must intersect $e_1$ or $e_2$. We check both cases:

\noindent \underline{case 1: $\gamma_r = \gamma_s$ intersects $e_1$:} In this case, we travel along the path from $p$ to $\gamma_r = \gamma_s$, then up $\gamma_r = \gamma_s$ until reaching $e_1$, and then along $e_1$ until hitting vertex $v_3$. We illustrate this path in \cref{fig:edge-intersections}(B).

When traveling along this path, one will either stay in the region $D_+$ the entire time, or it will leave $D_+$. If we stay in $D_+$ the entire time, then $v_3$ will be of type $D_+$. Further, because the entire path lies within $D_+$, no junction point was encountered betewen $\gamma_d=\gamma_s$ and $\gamma_r=\gamma_s$, giving (c)(i). In this case, $\sigma$ will look similar to the depiction in \cref{fig:edge-intersections}(B). Otherwise, there will be some point $q$ where this path leaves $D_+$. We check three cases for $q$, and illustrate each case as \cref{fig:edge-intersections}.

\noindent \underline{Case 1.a: $q$ lies on $\gamma_r = \gamma_s$:} We illustrate this case in \cref{fig:edge-intersections}(C). In this case, $q$ will be a junction point between $\gamma_d = \gamma_s$ and $\gamma_r = \gamma_s$. Since $v_1 \in \RRP$, the orientation of $\gamma_r = \gamma_s$ must be such that $v_1$ is on the right side of $\gamma_r = \gamma_s$. As a result, at the point $r$ where $\gamma_r = \gamma_s$ intersects $e$ (see \cref{fig:edge-intersections}(C)), the oriented tangent vector of $\gamma_r = \gamma_s$ must point outside of $\sigma$.

Thus, $\gamma_r = \gamma_s$ is oriented in the opposite direction as the path that we trace. Therefore, at junction point $q$, the oriented vector tangent to $\gamma_r = \gamma_s$ must point into the interior of the curve $\gamma_d = \gamma_s$, so $\gamma_r = \gamma_s$ enters $\DDP$ at $q$. We now show that $q$ can be the only junction point between $\gamma_d$ and $\gamma_r$.

Since $\gamma_d(p) > \gamma_r(p)$, but $\gamma_r(v_1) > \gamma_d(v_1)$, then there must be some point $q' \in e$ between $p$ and $v_1$ such that $\gamma_d(q') = \gamma_r(q')$. Notice that the line $\gamma_d = \gamma_r$ runs through $q$ and $q'$. Denote that line by $l$. Any junction point must lie on line $l$.

Let $\phi:\R \rightarrow l$ parametrize $l$ such that $\phi(0) = q'$ and $\phi(1) = q$. Suppose that there existed some $t \not = 1$ such that $\phi(t)$ was a junction point. We check three subcases:

\noindent \underline{Case 1.a.i: $t < 0$:} Since $\phi(0)$ lies on the boundary of $\sigma$, and $\phi(1)$ lies in the interior, then $\phi(t)$ must lie outside of $\sigma$. Thus, $\phi(t)$ cannot be a junction point.

\noindent \underline{Case 1.a.ii: $0 \leq t < 1$:} Notice that
\begin{align*}
(\gamma_d \circ \phi)(t) &= (\gamma_d \circ \phi)((1-t)(0)+(t)(1)) \\
&= (1-t)(\gamma_d \circ \phi)(0) + t(\gamma_d \circ \phi)(1) \\
&> (1-t)(\gamma_s \circ \phi)(0) + t(\gamma_s \circ \phi)(1) \\
&\geq (\gamma_s \circ \phi)((1-t)(0)+t(1)) \\
&= (\gamma_s \circ \phi)(t)
\end{align*}
Thus, $(\gamma_d \circ \phi)(t) > (\gamma_s \circ \phi)(t)$, so $\phi(t)$ cannot be a junction point.

\noindent \underline{Case 1.a.iii: $t > 1$:} In this case, $\gamma_d(t) = \gamma_s(t)$. Notice that
\[ 1 = \left( 1 - \frac{1}{t} \right)(0) + \left(\frac{1}{t}\right)(t) \]
And therefore
\begin{align*}
(\gamma_d \circ \phi)(1) &= (\gamma_d \circ \phi)\left( \left(1 - \frac{1}{t}\right)(0) + \left(\frac{1}{t}\right)(t) \right) \\
&= \left(1 - \frac{1}{t}\right)(\gamma_d \circ \phi)(0) + \left(\frac{1}{t}\right)\gamma_d(t) \\
&> \left(1 - \frac{1}{t}\right)(\gamma_s \circ \phi)(0) + \left(\frac{1}{t}\right)(\gamma_s \circ \phi)(t) \\
&\geq (\gamma_s \circ \phi)(1)
\end{align*}
Thus, $(\gamma_d \circ \phi)(1) > (\gamma_s \circ \phi)(1)$, which is a contradiction. Thus, $\phi(1)$ is the only junction point between $\gamma_d = \gamma_s$ and $\gamma_r = \gamma_s$. This proves (c)(iii).

\underline{Case 1.b: $q$ lies on $e_1$ and $q$ separates $D_+$ from $S$:} We illustrate this case in \cref{fig:edge-intersections}(C). Similar reasoning to the previous case gives (c)(v).

\underline{Case 1.c: $q$ lies on the border between $D_+$ and $R_-$:} In this case $\DP \cap e_1$ will not border $\SA \cap e_1$, as it borders $\RP \cap e_1$ and $\RN \cap e_1$. Thus, there must be no topologically significant intersections between $\gamma_d = \gamma_s$ and edge $e_1$. By assumption, there can be no junction points between $\gamma_d = \gamma_s$ and $\gamma_r = \gamma_s$. We illustrate this case in \cref{fig:edge-intersections}(D).

By assumption, $\gamma_d(q) = -\gamma_r(q)$. Let $q'$ be the point on edge $e$ where $\gamma_d(q') = -\gamma_r(q')$. Let $l$ denote the line connecting $q$ to $q'$. Because $\gamma_d$ and $\gamma_r$ are affine, then $\gamma_d = -\gamma_r$ on all of $l$. Thus, $l$ is the line $\gamma_d = -\gamma_r$ and all junction points between $\gamma_d = \gamma_s$ and $-\gamma_r = \gamma_s$ must occur on $l$.

By assumption, $\gamma_d(q) > \gamma_s(q)$ and $\gamma_d(q') > \gamma_s(q')$. Because $\gamma_d$ is affine and $\gamma_s$ is convex, it follows that $\gamma_d > \gamma_s$ on all of $l$. Thus, there can be no junction points between $\gamma_d = \gamma_s$ and $-\gamma_r = \gamma_s$.

Because $v_1 \in \RP$, and $\RN$ intersects $e_1$, it follows from \cref{lemma:opposite-vertices} that $v_3 \in \RRN$. Thus, either $v_3 \in \RN$ or $v_3 \in \DN$. But because $\DP$ borders $e_1$, from \cref{lemma:not-all-four}, it cannot hold that $v_3 \in \DN$. So $v_3 \in \RN$. Since $v_2 \in \RN$, the entire edge $e_2 \subset \RN$. Thus, there are no topologically significant intersections between $\gamma_d = \gamma_s$ and $e_3$. This gives (c)(vii).

Thus, in all of case 1, we will either have (c)(i), (c)(iii), (c)(v), or (c)(vii).

\noindent \underline{Case 2: $\gamma_r = \gamma_s$ intersects $e_2$:} In this case, we can proceed analogously to case 1, but instead of moving from $p$ towards $v_1$, we move from $p$ towards $v_2$. We then  travel up the curve $-\gamma_d = \gamma_s$, which will intersect $e_2$. Arguing analogously to the previous case will give either (c)(ii), (c)(iv), (c)(vi), or (c)(vii).

Thus, if $\DP$ borders edge $e$, then conditions (a), (b), and (c) will  be true. Now suppose that (a), (b), and (c) are true. If one of (c)(i) - (c)(vi) are true, we can find a point $p \in D_+$. By arguing in reverse, we can find a path to the edge $e$ that never leaves $\DP$, so $\DP$ will intersect $e$.

Now suppose that (c)(vii) is true. On edge $e$, $\RP \cap e$ cannot border $\RN \cap e$. Since $\gamma_r^2 = \gamma_s^2$ has no topologically significant intersections with $e$, then $\RP \cap e$ and $\RN \cap e$ cannot border any region of type $\SA$. Thus, each of $\RP \cap e$ or $\RN \cap e$ must border $\DP \cap e$ or $\DN \cap e$. We show that $\DN \cap e = \emptyset$.

Instead, suppose that $\DN \cap e \not = \emptyset$. Then define the following sets.
\begin{itemize}[noitemsep]
\item $E_1 = \{ p \in \sigma : \gamma_d(p) > \gamma_r(p)$ and $\gamma_d(p) > -\gamma_r(p)\}$
\item $E_2 = \{ p \in \sigma : -\gamma_d(p) > \gamma_r(p)$ and $-\gamma_d(p) > -\gamma_r(p)\}$
\item $E_3 = \{ p \in \sigma : \gamma_r > \gamma_d$ and $\gamma_r > -\gamma_d\}$
\item $E_4 = \{p \in \sigma : -\gamma_r > \gamma_d$ and $-\gamma_r > -\gamma_d\}$
\end{itemize}
Notice that $\DP \subset E_1$, $\DN \subset E_2$, $\RP \subset E_3$, and $\RN \subset E_4$. Thus, none of the $E_i$ are empty. It is easy to verify that each $E_i$ is connected, and the $\{E_i\}$ are disjoint. Further, whenever two of the $E_i$ border each other, the boundary will lie along one of the lines $\gamma_r = \gamma_d$, or $-\gamma_r = \gamma_d$.

In order for this to occur, the lines $\gamma_r = \gamma_d$ and $-\gamma_r = \gamma_d$ must divide $\sigma$ into at least four regions (one for each of the $E_i$). This can only happen if $\gamma_r = \gamma_s$ intersects $-\gamma_r = \gamma_s$ within $\sigma$. Let $q$ be the intersection point. Then $\gamma_r(q) = -\gamma_r(q)$, implying that $\gamma_r(q) = \gamma_d(q) = 0$. By assumption, it cannot hold that $\gamma_s(q) = 0$, so $q \in \SA$.

Let $p \in \DP$. Let $s$ be the segment between $p$ and $q$. Then $\SA \cap s \not = \emptyset$ and $\DP \cap s \not = \emptyset$. Because $\gamma_d(p) > \gamma_r(p)$ and $\gamma_d(q) = \gamma_r(q)$, and $\gamma_d$ and $\gamma_r$ are both affine, there will be no point on $s$ where $\gamma_r > \gamma_d$, so $\RP \cap s = \emptyset$. Similarly, $\RN \cap s = \emptyset$. Also, it is not possible for $\DP \cap s$ to border $\DN \cap s$. Thus, $\DP \cap s$ will border $\SA \cap s$, meaning that $\DP$ borders $\SA$.

Because $\DP$ borders $\SA$ it follows that the curve $\gamma_d = \gamma_s$ is not empty. Further, a portion of the curve $\gamma_d = \gamma_s$ separates $\DP$ from $\SA$. Because $\gamma_d = \gamma_s$ has no junction points and no topologically significant intersections with any edge, it follows that $\gamma_d = \gamma_s$ must be a circle or ellipse in the interior of $\sigma$. By \cref{lemma:internal-ellipse-other}, this would imply that $\DN = \emptyset$.

Thus, $\DN$ cannot intersect $\sigma$ at all. Therefore, $\RP \cap e$ borders $\DP \cap e$, implying that $\DP$ borders $e$.
\end{proof}

\begin{lemma}
Within $\sigma$, $D_+$ will border $R_+$ if and only if $D_+ \not = \emptyset$, $R_+ \not = \emptyset$, and one of the following is true:
\begin{itemize}[noitemsep]
\item[(a)] There is a junction point between $\gamma_d = \gamma_s$ and $\gamma_r = \gamma_s$.
\item[(b)] There exists an edge $e$ of $\sigma$, such that $D_+ \cap e \not = \emptyset$ and $R_+ \cap e \not = \emptyset$ and one of the following statements is true:
\begin{itemize}[noitemsep]
\item[(i)] One vertex of $e$ is of type $D_+$, and $\gamma_d = \gamma_s$ has no topologically significant intersections with $e$.
\item[(ii)] One vertex of $e$ is of type $R_+$ and $\gamma_r = \gamma_s$ has no topologically significant intersections with $e$.
\item[(iii)] Neither vertex is of type $R_+$ or $D_+$, and $\gamma_d = \gamma_s$ and $\gamma_r = \gamma_s$ each have at most one topologically significant intersection with $e$.
\end{itemize}
\end{itemize}
\noindent Analogous lemmas can be made with each pair in $\{D_+,D_-\} \times \{R_+,R_-\}$.
\label{lemma:region-boundaries}
\end{lemma}

\begin{figure}[!h]
\begin{overpic}[width=\linewidth]{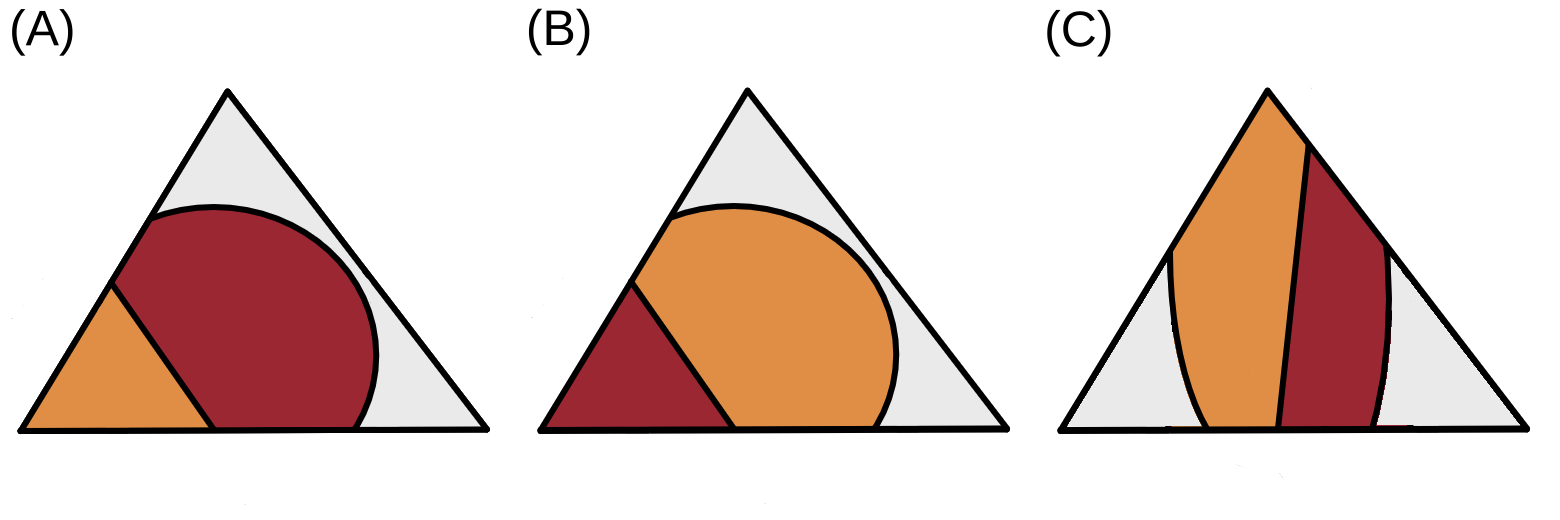}

\put (0,3) {\smaller[2]{$\phi(t_d)$}}
\put (10,3) {\smaller[2]{$\phi(t)$}}
\put (23,3) {\smaller[2]{$\phi(t_r)$}}

\put (72,3) {\smaller[2]{$\phi(t_d)$}}
\put (80,3) {\smaller[2]{$\phi(t)$}}
\put (88,3) {\smaller[2]{$\phi(t_r)$}}

\end{overpic}
\vspace{-4mm}
\caption{We demonstrate various cases in \cref{lemma:region-boundaries}. Here the edge $e$ corresponds to the bottom edge of the triangle. The orange region is $\DP$ while the red region is $\RP$. (A) Case 1 of the proof, where one vertex is of type $\DP$. (B) Case 2 of the proof, where one vertex is of type $\RP$. (C) Case 3 of the proof, where neither vertex is of type $\RP$ or $\DP$.}
\label{fig:region-boundaries}
\vspace{-6mm}
\end{figure}
\begin{proof}
We first show that if $\DP$ border $\RP$, then one of statements (a) or (b) must be true. If $\DP$ borders $\RP$, then consider the line $\gamma_r = \gamma_d$. Choose some point $p \in \sigma$ on that line that is on the border between $\DP$ and $\RP$, and trace the line. When tracing along the line, one will either encounter a region of type $\SA$, whereby a junction point is reached, giving (a), or one will hit an edge $e$. 

If an edge $e$ is hit, then trivially $\DP \cap e \not = \emptyset$ and $\RP \cap e \not = \emptyset$. Let $\phi:[0,1] \rightarrow e$ be a parametrization. Let $t \in [0,1]$ be such that $\phi(t)$ is the point where $\gamma_d = \gamma_r$ intersects $e$. Then there will be intervals $I_R,I_D \in [0,1]$ that have $t$ as an endpoint such that $\phi(I_1) = R_+ \cap e$ and $\phi(I_2) = D_+ \cap e$. Let $t_r, t_d \in [0,1]$ be such that $I_R = [t_r,t]$ and $I_D = [t,t_d]$.

We now check three cases for the vertices of $e$. We illustrate each case in \cref{fig:region-boundaries}, where the bottom edge of the triangle corresponds to $e$..

\noindent \underline{Case 1: One vertex is of type $\DP$:} We illustrate this case in \cref{fig:region-boundaries}(A). This case implies that $t_d = 1$. Thus, the only region that $\DP \cap e$ borders is $\RP \cap e$ with boundary at $\phi(t)$. Thus, $\DP \cap e$ will never border $\SA \cap e$, so $\gamma_d = \gamma_s$ will have no topologically significant intersections with $e$. This gives (b)(i).

\noindent \underline{Case 2: One vertex is of type $\RP$:} We illustrate this case in \cref{fig:region-boundaries}(B). This case is similar to case 1 and gives (b)(ii)

\noindent \underline{Case 3: Neither vertex is of type $\DP$ or $\RP$:} We illustrate this case in \cref{fig:region-boundaries}(C). In this case, $\DP \cap e$ and $\RP \cap e$ border each other at $\phi(t)$. Thus, $\gamma_d = \gamma_s$ could only have a topologically significant intersection at $\phi(t_d)$ and $\gamma_r = \gamma_s$ could only have a topologically significant intersection at $\phi(t_r)$. Thus, each curve could only have at most one topologically significant intersection. This gives (b)(iii).

Thus, we have shown that $\RP$ bordering $\DP$ implies that (a) or (b) is true. We now prove the opposite direction. We show that each condition implies that $\DP$ borders $\RP$. It is trivial that (a) implies that $\DP$ borders $\RP$. We now show this for (b). Let $\phi:[0,1] \rightarrow e$ parametrize $e$.

\para{(b)(i):} By \cref{lemma:mini-region-border}, the region $E_1 := \{ p \in e : \gamma_d > |\gamma_r|\}$ will border the region $E_2 := \{ p \in e : \gamma_r(p) > |\gamma_d(p)|\}$. Because $D_+$ contains a vertex $v$ of $e$, $E_1$ contains a vertex $v$ of $e$ (w.l.o.g. suppose that $E_1$ contains $\phi(0)$). Thus, there exists some $t \in [0,1]$ so that $E_1 = \phi([0,x])$. Because $E_1$ contains a vertex of $e$, it can only border one of $E_2$ or $E_3 := \{p \in e : -\gamma_r(p) > |\gamma_d(p)|\}$ at $\phi(x)$. Because $E_1$ borders $E_2$, it thus cannot border $E_3$.

Because $\gamma_d = \gamma_s$ does not have any topologically significant intersections with $e$, the region $D_+$ must border either $R_+$ or $R_-$. Notice that $D_+ \subset E_1$, $R_+ \subset E_2$, and $R_- \subset E_3$. Because $E_1$ does not border $E_3$, then $D_+$ cannot border $R_-$. Thus, $D_+$ borders $R_+$.

\para{(b)(ii):} This is similar to (b)(i).

\para{(b)(iii):} Because neither vertex is of class $D_+$, the region $D_+ \cap e$ must border a different region on both sides. Since $\gamma_d = \gamma_s$ has at most one topologically significant intersection with $e$, then one of the regions that borders $D_+$ must be $R_+$ or $R_-$. We check two cases:

\noindent \underline{case 1: $R_- \cap e = \emptyset$:} In this case, $D_+$ cannot border $R_-$ so $D_+$ borders $R_+$.

\noindent \underline{case 2: $R_- \cap e \not = \emptyset$:} In this case, by \cref{lemma:not-all-four}, we must have $D_- \cap e = \emptyset$. Arguing symmetrically, we can find that $R_+$ must border $D_+$ or $D_-$. Since $D_- = \emptyset$, it must follow that $R_+$ borders $D_+$.
\end{proof}

\begin{lemma}
Suppose that $\RRP$ and $\RRN$ both intersect $\sigma$. Excluding degenerate points, the eigenvector graph of $\sigma$ will be a linear graph of the structure $\RRP$--$\SRP$--$\SRN$--$\RRN$.
\label{lemma:eigenvector-single-r}
\end{lemma}
\begin{figure}
\begin{overpic}[width=\linewidth]{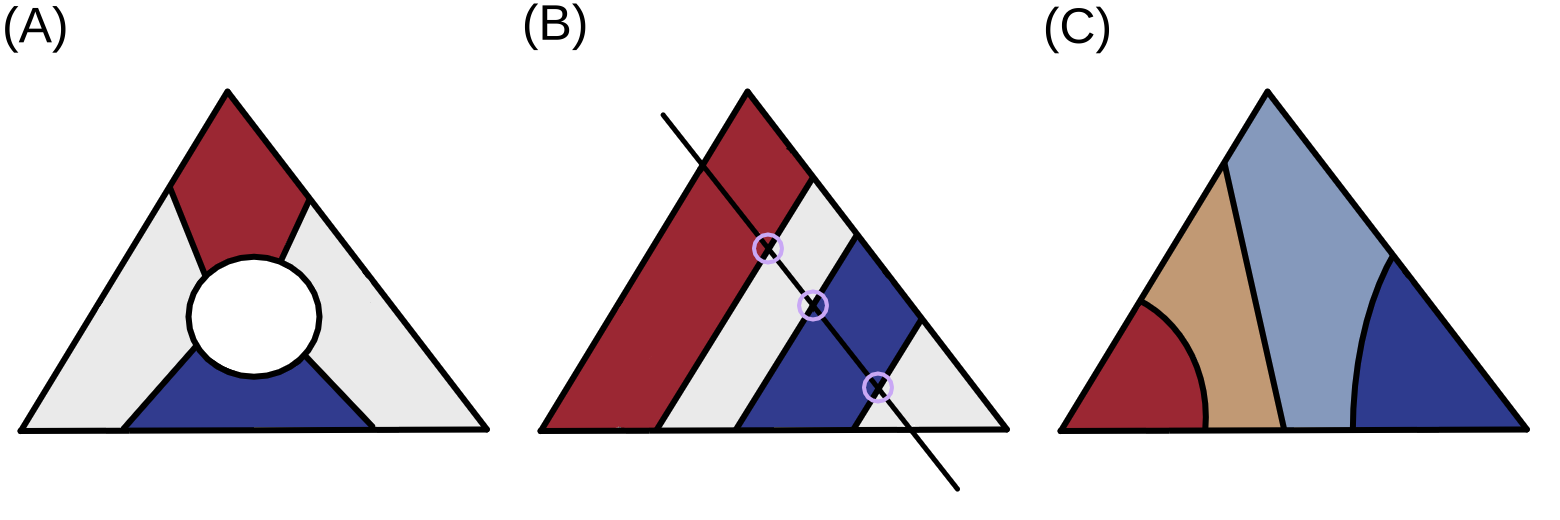}
\put (13,20) {\smaller[2]{\textcolor{white}{$\RRP$}}}
\put (22,10) {\smaller[2]{$\SE_1$}}
\put (15,7) {\smaller[2]{\textcolor{white}{$\RRN$}}}
\put (6,10) {\smaller[2]{$\SE_2$}}
\put (15,12) {\smaller[2]{$C$}}

\put (37.5, 8) {\smaller[2]{\textcolor{white}{$\RRP$}}}
\put (44.5, 8) {\smaller[2]{$\SE_1$}}
\put (50, 8) {\smaller[2]{\textcolor{white}{$\RRN$}}}
\put (58, 8) {\smaller[2]{$\SE_2$}}
\put (63, 1) {\smaller[2]{$l$}}

\put (71, 8) {\smaller[2]{\textcolor{white}{$\RRP$}}}
\put (75, 15) {\smaller[2]{\textcolor{white}{$\SRP$}}}
\put (82, 15) {\smaller[2]{\textcolor{white}{$\SRN$}}}
\put (89, 8) {\smaller[2]{\textcolor{white}{$\RRN$}}}
\put (63,16) {\smaller[2]{$\gamma_r = \gamma_s$}}
\put (80, 3) {\smaller[2]{$\gamma_r = 0$}}
\put (87.5, 19) {\smaller[2]{$-\gamma_r = \gamma_s$}}

\end{overpic}
\vspace{-4mm}
\caption{(A) We show regions $\RRP$, $\RRN$, $\SE_1$ and $\SE_2$ from the proof of \cref{lemma:eigenvector-border-edge-r}. It is not possible to fill in the region $C$ to prevent $\RRP$ from bordering $\RRN$ and $\SE_1$ from bordering $\SE_2$. (B) We show regions $\RRP$, $\RRN$, $\SE_1$ and $\SE_2$ from the proof of \cref{lemma:eigenvector-border-edge-r}. In this case, the line $l$ would intersect $\gamma_r^2 = \gamma_s^2$ three times (marked in purple). (C) A partition corresponding to the only possible eigenvector graph if $\RRP \not = \emptyset$ and $\RRN \not = \emptyset$.}
\label{fig:eigenvector-border-edge-r}
\vspace{-6mm}
\end{figure}
\begin{proof}
Let $\SE = \SRP \cup \SRN$. We claim that $\SE$ has only one connected component. Suppose, to the contrary, that $\SE$ had two connected components, which we will call $\SE_1$ and $\SE_2$.

It is not possible for both $\RRP$ and $\RRN$ to border both $\SE_1$ and $\SE_2$. Because $\RRP$ and $\RRN$ do not border each other, and $\SE_1$ and $\SE_2$ do not border each other, such a case would be geometrically impossible, as we demonstrate in \cref{fig:eigenvector-border-edge-r}(A). In \cref{fig:eigenvector-border-edge-r}(A), there is no way to fill in the region $C$ such that $\RRP$ does not border $\RRN$ and $\SRP$ does not border $\SRN$. Thus, at least one of $\RRP$ or $\RRN$ borders only one of $\SE_1$ or $\SE_2$. Without loss of generality, suppose that $\RRP$ borders $\SE_1$, but not $\SE_2$. 

Let $l$ be a line connecting $\RRP$ to $\SE_2$. Since $\RRP$ does not border $\RRN$ or $\SE_2$, it follows that, after exiting $\RRP$, the line $l$ must enter $\SE_1$. Thus, the line $l$ passes between $\SE_1$ and $\SE_2$. Because $\RRP$ and $\RRN$ are convex, and $l$ has already left $\RRP$, $l$ must pass through $\RRN$ between $\SE_1$ and $\SE_2$. In total, when traveling from $\RRP$ to $\SE_2$, the line $l$ passes through $\RRP \rightarrow \SE_1 \rightarrow \RRN \rightarrow \SE_2$. We illustrate this case in \cref{fig:eigenvector-border-edge-r}(B), where we highlight the intersectiosn between $\gamma_r^2 = \gamma_s^2$ and $l$ in purple.

Whenever $l$ passes from a region where $\gamma_r^2 > \gamma_s^2$ to one where $\gamma_s^2 > \gamma_r^2$ (or vice versa), $l$ must intersect the curve $\gamma_r^2 = \gamma_s^2$. In the path from $\RRP \rightarrow \SE_1 \rightarrow \RRN \rightarrow \SE_2$, then, the line $l$ must intersect $\gamma_r^2 = \gamma_s^2$ exactly three times. However, since $\gamma_r^2 = \gamma_s^2$ is a conic section, it is impossible for a line to intersect it exactly three times. Thus, there is only one connected component of type $S$.

Therefore, $\gamma_r^2 = \gamma_s^2$ will divide $\sigma$ into three regions: $\RRP$, $\RRN$, and $\SE$. Since there is no point where $\gamma_r = \gamma_s = 0$, the line $\gamma_r = 0$ must not intersect $\gamma_r^s = \gamma_s^2$. Thus, the line $\gamma_r = 0$ will divide $\SE$ into two connected components corresponding to $\SRP$ and $\SRN$. We illustrate the final partition in \cref{fig:eigenvector-border-edge-r}(C). Trivially, $\SRP$ borders $\SRN$. $\RRP$ must only border $\SRP$, and $\RRN$ must only border $\SRN$. Thus, we $\sigma$ will have en eigenvector graph of the structure $\RRP$--$\SRP$--$\SRN$--$\RRN$.
\end{proof}

\begin{lemma}
Let $e$ be an edge of $\sigma$. Then $\RRP$ will border $e$ if and only if $\RRP$ contains a vertex of $e$, or if $\gamma_r = \gamma_s$ intersects $e$. The analogous claim for $\RRN$ is also true.
\label{lemma:eigenvector-border-edge-r}
\end{lemma}
\begin{proof}
It is trivial that these conditions are sufficent for $\RRP$ to border $e$. We now show that they are necessary. Suppose that $\RRP$ borders $e$. If $\RRP$ borders a vertex of $e$, then the condition is met. Otherwise, $\RRP \cap e \not = e$. From \cref{lemma:illegal-intersections} and \cref{lemma:illegal-intersections-2}, this implies that $\RRP \cap e$ must border $\SRP \cap e$ in which case $\RRP \cap e$ will be separated from its neighbor along the curve $\gamma_r^2 = \gamma_s^2$, whereby $\gamma_r = \gamma_s$ will intersect $e$.
\end{proof}

\begin{lemma}
Let $e$ be an edge of $\sigma$. Suppose that $\SRP$ has only one connected component. Then $\SRP$ will intersect $e$ if and only if $\SRP$ contains a vertex of $e$, if $\gamma_r = \gamma_s$ intersects $e$, or if $\gamma_r = 0$ intersects $e$. The analogous claim for $\SRN$ is also true. 
\label{lemma:eigenvector-border-edge-s}
\end{lemma}
\begin{proof}
The proof is very similar to that of \cref{lemma:eigenvector-border-edge-r}.
\end{proof}

\begin{figure}
\includegraphics[width=\linewidth]{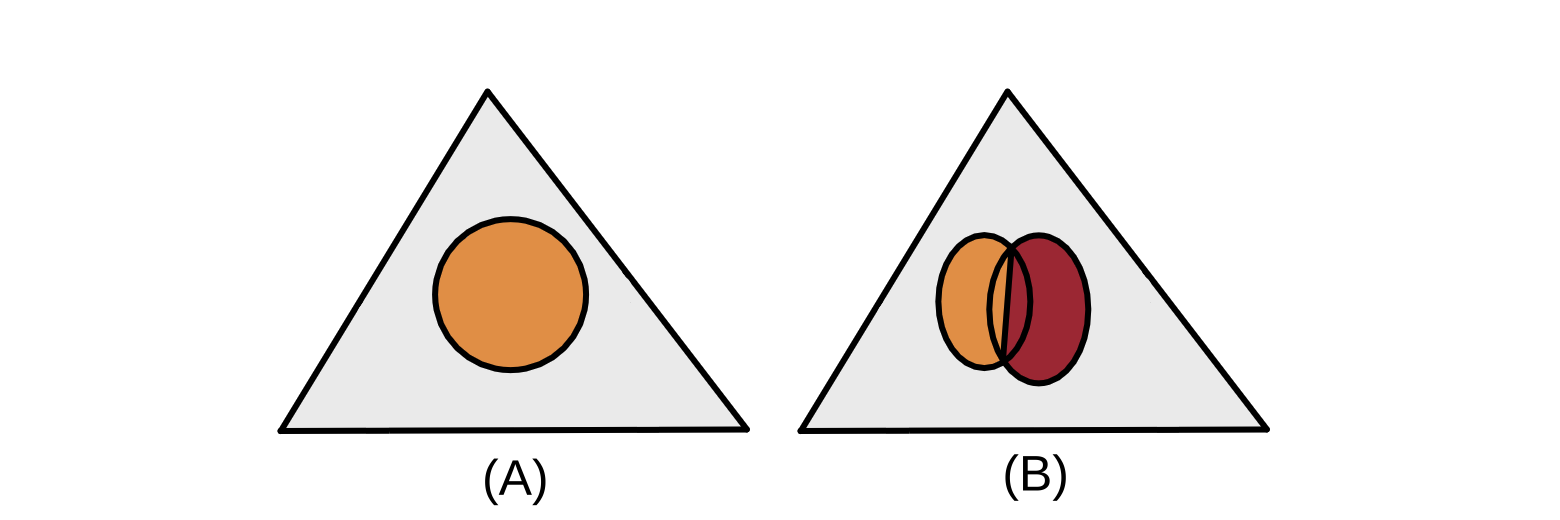}
\vspace{-8mm}
\caption{We depict the two possible ways that a region of type $\SA$ can completely surround another region.  (A) The case where $\DP$ (or whichever region is surrounded) has no junction points. (B) The case where $\DP$ has exactly two junction points with $\RP$ (where either $\DP$ or $\RP$ can be substituted for $\DN$ or $\RN$).}
\label{fig:internal-ellipse-detection}
\vspace{-6mm}
\end{figure}

\begin{lemma}
We say that that a region $R \subset \sigma$ of type $\SA$ \textit{completely surrounds} $\DP$ if $\DP \not = \emptyset$ and all paths starting from any point in $\DP$ will encounter $\SA$ before an edge of $\sigma$.
\begin{itemize}[noitemsep]
\item[(1)]This situation will occur if and only if all of the following are true:
\begin{itemize}[noitemsep]
\item[(a)] $\DP$ intersects the interior of $\sigma$.
\item[(b)] $\gamma_d = \gamma_s$ has no topologically significant intersections with any edges of $\sigma$.
\item[(c)] All vertices of $\sigma$ are of type $\SA$.
\item[(d)] $\DP$ has no junction points, or the following are true:
\begin{itemize}[noitemsep]
\item[(i)] The curve $\gamma_d = \gamma_s$ has exactly two junction points. The junction points are either both with $\gamma_r = \gamma_s$, or both with $-\gamma_r = \gamma_s$ (w.l.o.g. suppose they are with $\gamma_r = \gamma_s$).
\item[(ii)] The curve $\gamma_r = \gamma_s$ has exactly two junction points. Both junction points are with $\gamma_d = \gamma_s$.
\item[(iii)] The curve $\gamma_r = \gamma_s$ has no topologically significant intersections with any edge of $\sigma$.
\end{itemize}
\end{itemize}
\item[(2)] If a region $R$ of type $\SA$ completely surrounds $\DP$, then the eigenvector partition of $\sigma$ will only be one region of type $\SA$, and any regions completely surrounded by $\SA$.
\end{itemize}
Analogous lemmas can be shown for $\DN$, $\RP$, and $\RN$.
\label{lemma:internal-ellipse-detection}
\end{lemma}

\begin{proof}
In \cref{fig:internal-ellipse-detection}, we illustrate the two ways that a region of type $\SA$ can completely surround a region of another type. \cref{fig:internal-ellipse-detection}(A) corresponds to the case where $\DP$ has no junction points, whereas (B) corresponds to the case where $\DP$ has exactly two junction points with $\RP$.

We first show (1). First, suppose that $R$ completely surrounds $\DP$. (a) is trivial. Also, $\DP$ cannot intersect any edges, giving (b).

In order for $R$, which is type $\SA$, to surround $\DP$, it must follow that $\DDP$ is also completely surrounded by $R$. As such, the curve $\gamma_d = \gamma_s$ is a connected component of a conic section that does not intersect any edges of $\sigma$, so it must be a circle or ellipse whose interior is of type $\DDP$. By \cref{lemma:internal-ellipse-is-both}, either $\RRP$ or $\RRN$ must intersect $\gamma_d = \gamma_r$. Without loss of generality, suppose that $\RRP$ intersects $\gamma_d = \gamma_r$.

By definition, $\RRP \cap \SA = \emptyset$. By assumption, every path starting at $\DDP$ that passes through $\RRP$ must intersect $\SA$ before intersecting an edge of $\sigma$. Since $\DDP$ intersects $\RRP$, it follows that $\SA$ also completely surrounds $\RRP$. Thus, the curve $\gamma_r = \gamma_s$ is also a circle or ellipse in the interior of $R$. From \cref{lemma:internal-ellipse-other}, $\DDN = \emptyset$ and $\RRN = \emptyset$, implying that $\DN = \emptyset$ and $\RN = \emptyset$. Since $\DDP$ nor $\RRP$ are ellipses in the interior of $\sigma$, and $\DDN = \emptyset$ and $\RRN = \emptyset$, it follows that all three vertices of $\sigma$ must be of type $\SA$, giving (c).


If $\gamma_r = \gamma_s$ lies completely inside of $\gamma_d = \gamma_s$, then any points in $\RRN$ must lie in the interior of $\DP$. However, since $\DP$ is convex, this forces $\RP = \emptyset$. Thus, $\DP$ will not border any other regions, giving (d).

Now suppose that $\gamma_r = \gamma_s$ does not lie completely inside of $\gamma_d = \gamma_s$. Then the two curves will intersect at points along the line where $\gamma_d = \gamma_r$. Since $\gamma_d = \gamma_s$ and $\gamma_r = \gamma_s$ are both subsets of conic sections, the line $\gamma_d = \gamma_r$ can only intersect each curve at most twice. Thus, $\gamma_d = \gamma_s$ and $\gamma_r = \gamma_s$ will intersect exactly two times. Since $\DDN = \emptyset$ and $\RRN = \emptyset$, these will be the only junction points along either curve. This fact gives (d)(i) and (d)(ii).

Also, because $\gamma_r = \gamma_s$ is an ellipse in the interior of $\sigma$, (d)(iii) must follow, so we have (d).

Thus, if $R$ contains $\DP$ in its interior, then (a), (b), (c), and (d) must follow.

Now suppose that conditions (a), (b), (c), and (d) are true. Without loss of generality, assume that $\gamma_d = \gamma_s$ has no junction points with $-\gamma_r = \gamma_s$.

We claim that $\DP$ will not border any edge $e$. Suppose, to the contrary, that it did. Since no vertex is of type $\DP$, $\DP \cap e$ will border two other regions within $e$. If either region is of type $\SA$, then $\gamma_d = \gamma_s$ will have a topologically significant intersection between that region and $\DP$, which we assume cannot occur. Thus, $\DP \cap e$ can only border $\RP \cap e$ and $\RN \cap e$. Since $\RP \cap e$ and $\RN \cap e$ are connected, $\DP \cap e$ will border both $\RP \cap e$ and $\RN \cap e$, intersecting one region on each side. In that case, by \cref{lemma:opposite-vertices}, then, one vertex of $e$ is of type $\RRP$ and the other is of type $\RRN$. In that case, neither vertex can be of type $\SA$, violating assumption (c). Thus, $\DP$ does not border any edge $e$.

We also claim that $\DP$ cannot border $\RN$. Suppose that it did. By \cref{lemma:region-boundaries}, if $\DP$ did border $\RN$, then either $\gamma_d = \gamma_s$ and $-\gamma_r = \gamma_s$ would share a junction point, or $\DP$ and $\RN$ would both intersect some edge $e$ of $\sigma$. However, neither case is possible by assumption. Thus, $\DP$ does not border $\RN$.

We now check two cases:

\underline{Case 1: $\gamma_d = \gamma_s$ has no junction points:} Then, by \cref{lemma:region-boundaries}, $\DP$ will not border $\RP$. Thus, $\DP$ can only border a region $R$ of type $\SA$, and does not border any vertices or edges of $\sigma$. Therefore, $\DP$ must be completely surrounded by $R$.

\underline{Case 2: $\gamma_d = \gamma_s$ has exactly two junction points with $\gamma_r = \gamma_s$:} In this case, then clearly $\DP$ will border $\RP$. Using similar reasoning to before, we can show that $\RP$ will only border $\DP$ and a region of type $\SA$. We can also show that $\RP$ will not border any edges of $\sigma$. Thus, $\RP \cup \DP$ will only border regions of type $\SA$, and will not border any edges of $\sigma$. This is possible only if $\RP \cup \DP$ is completely surrounded by a single region $R$ of type $\SA$, meaning that both $\RP$ and $\DP$ lie within the interior of $R$. This completes the proof of (1).

To show (2), in the reasoning above, we demonstrated that if $\DP$ lies in the interior of a region $R$ of type $S$, then the curves $\gamma_d = \gamma_s$ and $\gamma_r = \gamma_s$ are both ellipses contained in the interior of $\sigma$. Further, either $\RP = \emptyset$, or $\RP$ and $\DP$ are both completely surrounded by $R$.

By \cref{lemma:internal-ellipse-other}, it follows that $\DDN = \emptyset$ and $\RRN = \emptyset$. Thus, the only regions within $\sigma$ will be $\DP$, $\RP$ (if it is not empty), and $R$, which has type $\SA$.
\end{proof}

%% file: appendix-partition-correctness-main-results.tex
\subsection{Proofs of Main Results}
\label{appendix:partition-correctness-main-results}

We now prove the main results. For \cref{lemma:eigenvector-topology}, the goal is to use our invariant to recover the structure of the eigenvector graph of $\sigma$. We also wish to recover: for each node $n$ of the eigenvector graph, which vertices and edges of $\sigma$ border the regions within $\sigma$ corresponding to $n$.
\lemmaEigenvectorTopology*
\begin{figure}
\includegraphics[width=\linewidth]{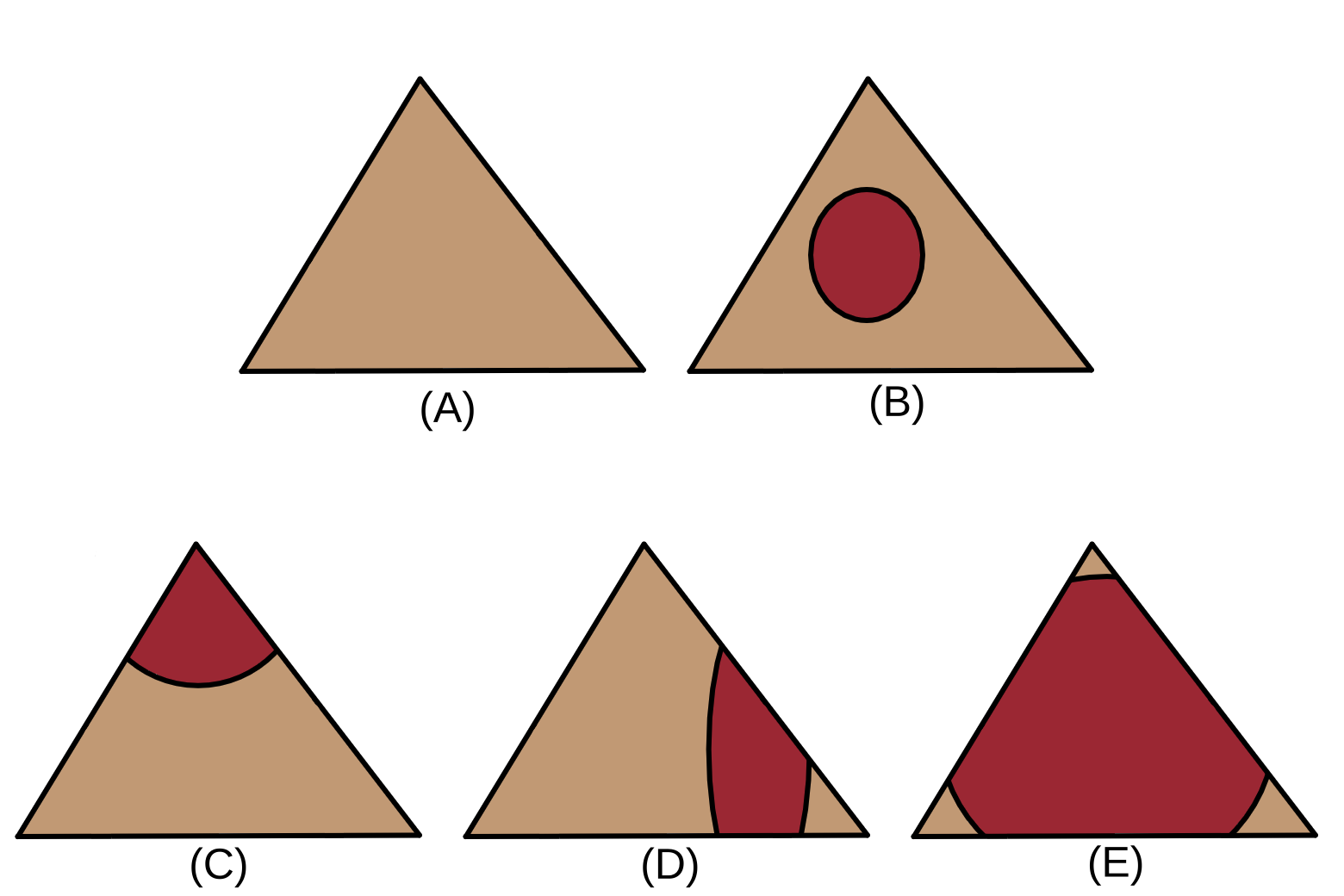}
\vspace{-4mm}
\caption{We demonstrate the different possible ways that $\gamma_d = \gamma_s$ can divide $\sigma$ into multi regions. The red region is the interior of $\gamma_d = \gamma_s$, and the orange region is the outside. (A) $\gamma_d = \gamma_s$ does not intersect $\sigma$. (B) $\gamma_d = \gamma_s$ does not intersect the boundary of $\sigma$, but forms a region within its interior. (C) $\gamma_d = \gamma_s$ intersects the boundary of $\sigma$ twice, forming two regions. (D) $\gamma_d = \gamma_s$ intersects the boundary of $\sigma$ four times, forming three regions. (E) $\gamma_d = \gamma_s$ intersects the boundary of $\sigma$ six times, forming four regions.}
\label{fig:eigenvector-topology}
\vspace{-6mm}
\end{figure}
\begin{proof}
If $\RRP$ and $\RRN$ are both inside of $\sigma$, then by \cref{lemma:eigenvector-single-r}, there is only one possible eigenvector graph. That eigenvector graph will have one node of each type $\RRP$, $\SRP$, $\SRN$ and $\RRN$. Which regions border which vertices can be determined from the classifications of each vertex, and \cref{lemma:eigenvector-border-edge-r} gives which edges are bordered by $\RRP$ and $\RRN$. From the vertex classification, we know whether $\gamma_r > 0$ for each vertex. From that information, we can deduce which edges intersect the line $\gamma_r = 0$. Then, \cref{lemma:eigenvector-border-edge-s} allows us to deduce which edges border $\SRP$ and $\SRN$.

Now suppose that $\RRP$ and $\RRN$ do not both intersect $\sigma$. Without loss of generality, suppose that $\RRN$ does not intersect $\sigma$. Since $\RRP$ is convex, its boundary, given by $\gamma_r = \gamma_s$, will intersect the boundary of $\sigma$ either zero, two, four, or six times. In doing so, $\gamma_r = \gamma_s$ will divide $\sigma$ into, respectively, one, two, three, or four regions. If $\gamma_r = \gamma_s$ never intersects the boundary of $\sigma$, it can also possibly be an ellipse in the interior of $\sigma$, thus dividing $\sigma$ into two regions. We illustrate each of these cases in \cref{fig:eigenvector-topology}. Further, the line $\gamma_r = 0$ We now check each case.

\underline{Case 1: $\gamma_r = \gamma_s$ intersects $\sigma$ zero times:} We check four subcases:

\underline{Case 1(i): $\partial \RRP = \emptyset$; all vertices are the same classification:} In this case, every vertex must be of the same class, either $\RRP$, $\SRP$ or $\SRN$. The eigenvector graph will be a single node, and its corresponding region will border every vertex and edge of $\sigma$.

\underline{Case 1(ii): $\partial \RRP = \emptyset$; some vertices have different classifications:} In this case, $\sigma$ will have regions of multiple classifications. None of them can be $\RRP$, as otherwise $\partial \RRP$ would separate $\RRP$ from another region. Thus, some vertices will be $\SRP$ and others $\SRN$. By inspecting the vertex classifications, we can find which edges will intersect the line $\gamma_r = 0$. The line $\gamma_r = 0$ will divide $\sigma$ into two regions, one of type $\SRP$ and the other of type $\SRN$. Using \cref{lemma:eigenvector-border-edge-s}, we can deduce which edges intersect each of the two regions.

\underline{Case 1(iii): $\partial \RRP \not = \emptyset$; all vertices are the same classification:} In this case, $\RRP$ will be a convex region contained within the interior of $\sigma$. Because $\RRP$ can only border $\SRP$, it must be contained within the interior of $\SRP$. Thus, every vertex must be $\SRP$. As a result, the eigenvector graph will have two nodes, one of type $\RRP$ and the other of type $\SRP$. The region of type $\SRP$ will border every vertex and edge.

\underline{Case 1(iv): $\partial \RRP \not = \emptyset$; not all vertices are the same classification:} In this case, $\RRP$ will be a convex region in the interior of $\sigma$. As before, it must be contained within $\SRP$. In this case, some vertices must be of type $\SRP$ and others of type $\SRN$. The topology can be found by using similar reasoning to before.

\underline{Case 2: $\gamma_r = \gamma_s$ intersects $\partial \sigma$ twice:} In this case, $\gamma_r = \gamma_s$ will divide $\sigma$ into two regions. We can determine which is which by inspecting the classifications of the vertices.

We can inspect the classifications of the vertices to detect if the line $\gamma_r = 0$ runs through $\sigma$. If it does, it will divide the region where $\gamma_r < \gamma_s$ into two regions, one of type $\SRP$ and the other $\SRN$, where $\RRP$ will border $\SRP$. If $\gamma_r = 0$ does not run through $\sigma$, then $\sigma$ will be divided into two regions, one of type $\SRP$ and the other of type $\RRP$. In either case, we can compute the eigenvector graph. Then, which regions intersect which vertices and edges can be computed from \cref{lemma:eigenvector-border-edge-r} and \cref{lemma:eigenvector-border-edge-s}.

\underline{Case 3: $\gamma_r = \gamma_s$ intersects $\partial \sigma$ four times:} In this case, $\gamma_r = \gamma_s$ will divide $\sigma$ into three regions. In two regions, $\gamma_s > \gamma_r$. The third region, which lies between the other two, will satisfy $\gamma_r > \gamma_s$.

Because $\RRP$ is convex, any region where $\gamma_s > \gamma_r$ (that lies outside of $\RRP$) must contain at least one vertex of $\sigma$. One can verify that two regions where $\gamma_s > \gamma_r$ will be connected if and only if both regions border some edge $e$ that does not intersect $\gamma_r = \gamma_s$. By inspecting which vertices and edges are intersected by $\gamma_r = \gamma_s$, we can uniquely identify the different regions where $\gamma_r > \gamma_s$.

Using similar reasoning previous cases, we can compute the regions that intersect the line $\gamma_r = 0$. Any region where $\gamma_s > \gamma_r$ not intersected by the line $\gamma_r = 0$ will be entirely classified as $\SRP$. Any region where $\gamma_s > \gamma_r$ intersected by $\gamma_r = 0$ will be divided into two regions, one of type $\SRP$ and the other $\SRN$. The region of type $\SRP$ border $\RRP$, while $\SRN$ will not border $\RRP$.

Which vertices intersect $\RRP$ follow from their classifications, and \cref{lemma:eigenvector-border-edge-r} determines which edges border $\RRP$. For any region where $\gamma_s > \gamma_r$, if that region borders an edge $e$, then it must border some vertex $v$ of $e$. Thus, we can identify all edges and vertices bordering that region. For any region where $\gamma_s > \gamma_r$ that intersects the line $\gamma_r = 0$, by inspecting the position of the line we can determine which edges intersect the sub-regions of type $\SRP$ and $\SRN$.

\underline{Case 4: $\gamma_r = \gamma_s$ intersects $\partial \sigma$ six times:} In this case, $\gamma_r = \gamma_s$ will divde $\sigma$ into four regions. There will be three regions where $\gamma_s > \gamma_r$. From here, computing the topology follows closely to the previous case.
\end{proof}

\lemmaEigenvalueTopology*
\begin{figure}
\begin{overpic}[width=0.9\linewidth]{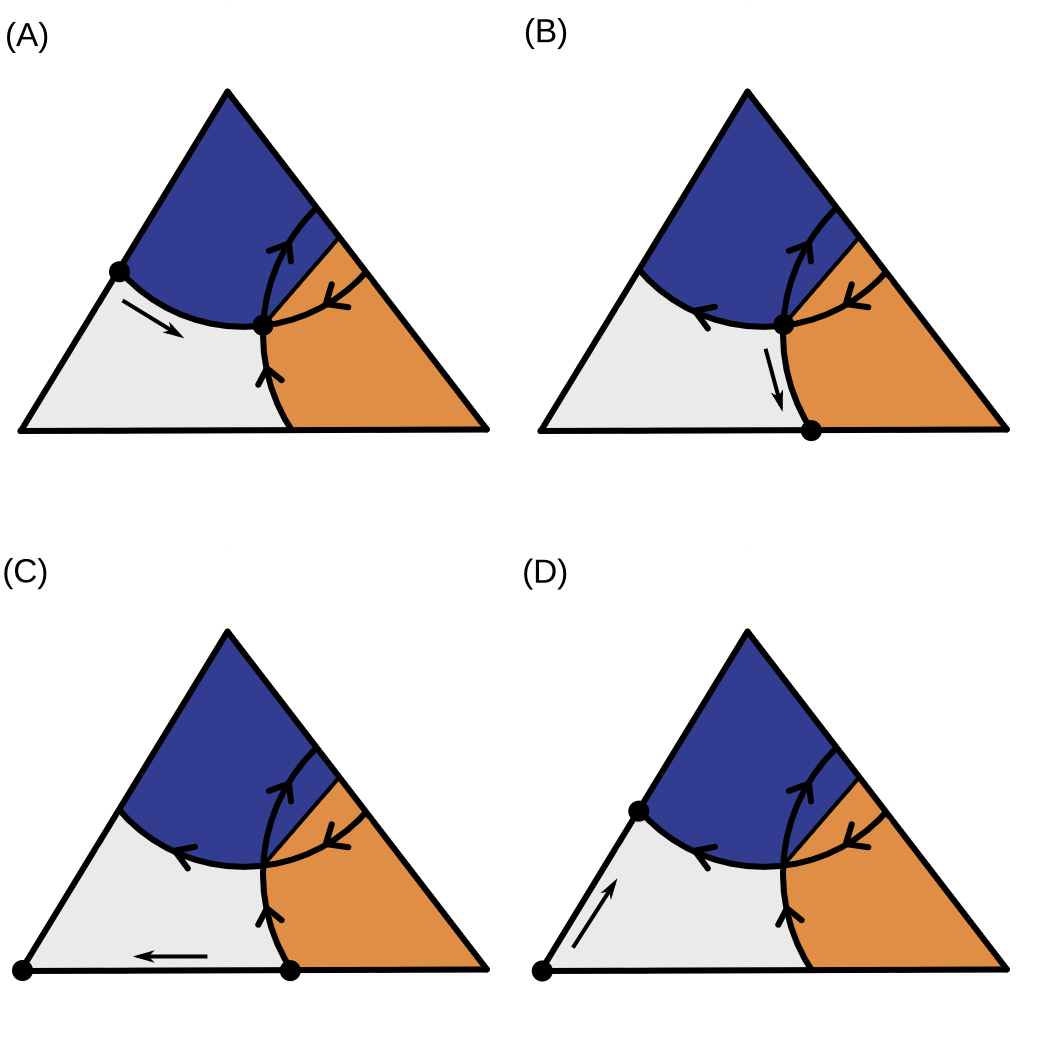}
\put (7, 74){$p_i$}
\put (21, 65){$p_i'$}


\put (71, 71){\textcolor{white}{$p_i$}}
\put (76, 54){$p_i'$}


\put (25, 3){$p_i$}
\put (1, 3){$p_i'$}


\put (51, 3) {$p_i$}
\put (57, 21) {$p_i'$}

\end{overpic}
\vspace{-4mm}
\caption{If $p$ is the vertex of the region of type $\SA$, we identify the next vertex along the boundary of $\SA$ when traveling clockwise. (A) $p_i$ is a point where a conic section enters $\sigma$. (B) $p_i$ is a junction point. (C) $p_i$ is a point where a conic section leaves $\sigma$. (D) $p_i$ is a vertex of $\sigma$.}
\label{fig:eigenvalue-topology}
\vspace{-6mm}
\end{figure}
\begin{proof}
We assume that every region of the eigenvalue partition does not completely surround any other region. Because $\DP$, $\DN$, $\RP$, and $\RN$ are convex, the only type of region that could surround another region is the region of type $\SA$. \Cref{lemma:internal-ellipse-detection}(1) gives a necessary and sufficient condition for such a situation to occur that can be determined from (a), (b), (c), and (d). Further, \Cref{lemma:internal-ellipse-detection}(2) demonstrates that, when such a case arises, the topology is very easy to recover from (a) (b) (c) and (d).

We first claim that the number of nodes in the eigenvector graph of each region type can be recovered. 

If any vertex is of type $\DP$, as in (a), or if $\gamma_d = \gamma_s$ has any junction points or topologically significant intersections with edges of $\sigma$, as described in (b), then clearly $\DP \cap \sigma \not = \emptyset$. If $\gamma_d = \gamma_s$ does not have any intersections, then from (d) we will know whether or not $\DP \cap \sigma = \emptyset$. The same is true for $\DN$, $\RP$, and $\RN$. Thus, we will know whether there is at least one node of each of these types in the eigenvalue graph. By \cref{lemma:regions-are-convex}, $\DP$, $\DN$, $\RP$, and $\RN$ each have at most one connected component, so there can be at most one node of each of these types in the eigenvalue graph. Thus, we can recover the number of nodes in the eigenvalue graph of types $\DP$, $\DN$, $\RP$, and $\RN$.

We are only left to find the number of type $\SA$. We will compute the number of nodes of type $\SA$ by tracing the boundary of each region within the eigenvalue partition of $\sigma$ using the points from (b).

Let $P = \{p_1,p_2,\ldots,p_n\}$ be all of the points from (b), along with the vertices of $\sigma$ that are of type $\SA$ (which we can identify from (a)). It is easy to show that each $p \in P$ borders one region of type $\SA$.

The boundary of each region $R$ of type $\SA$ must consist of sections of the conics $\gamma_d^2 = \gamma_s^2$ and $\gamma_r^2 = \gamma_s^2$, as well as sections of edges of $\sigma$. Any time that two of these segments meet on the boundary of $R$, they will intersect at some point $p \in P$. Thus, for each region of type $\SA$, by traversing the boundary one will encounter a loop of points in $P$.

Now let $p_i \in P$ be a vertex of the boundary of some region $R$ of type $S$. We claim that (b) and (c) allow us to compute the next vertex $p_i'$ that immediately follows $p_i$ when tracing $R$ clockwise. We check four cases. We illustrate each case in \cref{fig:eigenvalue-topology}.

\noindent \underline{Case 1: A conic section $l$ leaves $\sigma$ at $p_i$:} In this case, $p_i'$ will be the next point in $P$ that $l$ intersects when traveling counterclockwise along $l$. We can find this point $p_i'$ from (b). We illustrate this case in \cref{fig:eigenvalue-topology}(A).

\noindent \underline{Case 2: $p_i$ is a junction point:} Denote the two curves that meet at the junction point as $l_1$ and $l_2$. At the intersection, one of the conics will be entering the interior of the other at $p_i$. Without loss of generality, suppose that $l_1$ enters the interior of $l_2$. Then we find the point $p_i'$ along $l_1$ that comes after $p_i$ when traveling counterclockwise along $l_1$. we illustrate this case in \cref{fig:eigenvalue-topology}(B).

\noindent \underline{Case 3: A conic section $l$ enters $\sigma$ at $p_i$:} Let $e$ be the edge where $l$ enters $\sigma$. Let $v_1$ and $v_2$ be the vertices of $e$, such that $v_1$ comes after $v_2$ when traversing $\sigma$ clockwise. In this case, then the boundary of $R$ will travel towards $v_1$. Then, $p_i'$ will be the closest point in $P$ to $p_i$ along $e$ when traveling towards $v_1$. We can identify $p_i'$ by identifying the relative positions of each point in $P \cap e$ relative to vertices $v_1$ and $v_2$. We illustrate this case in \cref{fig:eigenvalue-topology}(C).

\noindent \underline{Case 4: $p_i$ is a vertex of $\sigma$:} In this case, let $e$ be the edge that comes after $p_i$ when traversing $\sigma$ in clockwise order. Then $p_i'$ will be the point on $e$ closest to $p_i$. We illustrate this case in \cref{fig:eigenvector-topology}(D).

Across all cases, we can find the next point $p_i'$ that follows after point $p_i$. This allows us to define a permutation on the points in $P$. Each cycle of the permutation will correspond to a different connected component of the boundary of some region $R$ of type $\SA$. Since no region of type $\SA$ surrounds any other region, the boundary of each region $R$ of type $\SA$ will be connected. Thus, we can uniquely identify each region of type $\SA$.

Next, we find which edges and vertices are adjacent to each region. Clearly, the construction of each region of type $S$ will yield which vertices and edges it borders. 

For the other types of regions, we can inspect the vertex classifications and edge intersections to identify which regions intersect each edge. For $\DP$, if a vertex $v$ is of type $\DP$, then $\DP$ will border $v$ as well as the two edges the join at $v$. Also, $\DP$ must border an edge $e$ if $\gamma_d = \gamma_s$ has a topologically significant intersection with $e$. If, for some edge $e$, $\DP$ does not intersect either vertex of $e$, and $\gamma_d = \gamma_s$ has no topologically significant intersections with $e$, then \cref{lemma:edge-intersections} gives a necessary and sufficient condition for $\DP$ to border $e$ that can be determined from (a), (b), (c), and (d). We can use the same approach for $\DN$, $\RP$, and $\RN$.

Finally, we must determine which regions of the eigenvalue partition border one another. The construction of each region of type $S$ will yield which other regions it borders. \Cref{lemma:region-boundaries} provides a necessary and sufficient condition for two regions of types $\DP$, $\DN$, $\RP$, or $\RN$ to border each other that can be computed from (a), (b), (c), (d), and which edges and vertices each region borders.
\end{proof}

%% file: appendix-partition-correctness-edge-cases.tex
\subsection{Handling of Edge Cases}
\label{appendix:partition-correctness-edge-cases}

In the previous proofs, we made the following assumptions:

\begin{itemize}[noitemsep]
\item[(i)] $\gamma_d^2 = \gamma_s^2$ and $\gamma_r^2 = \gamma_s^2$ only intersect each other and each edge of $\sigma$ transversally. 
\item[(ii)] No junction points occur on the edges of $\sigma$. 
\item[(iii)] No two of the functions $|\gamma_d|$, $|\gamma_r|$, or $\gamma_s$ are exactly equal on all of $\R^2$.
\item[(iv)] The conic sections $\gamma_d^2 = \gamma_s^2$ and $\gamma_r^2 = \gamma_s^2$ are either empty or have infinitely many points.
\item[(v)] The curve $\gamma_d = \gamma_s$ will border one region of type $\DDP$ and another where $\gamma_s > |\gamma_d|$. Make similar assumptions about $-\gamma_d = \gamma_s$, $\gamma_r = \gamma_s$, and $-\gamma_r = \gamma_s$.
\item[(vi)] $|\gamma_d|$, $|\gamma_r|$ and $\gamma_s$ are never equal on a vertex of $\sigma$.
\item[(vii)] There are no vertices of $\sigma$ where $\gamma_r = 0$.
\item[(viii)] There are no points in $\sigma$ where $\gamma_d = \gamma_s = 0$ or $\gamma_r = \gamma_s = 0$.
\end{itemize}

In \cref{appendix:edge-cases}, we ensure (i)-(v) using symbolic perturbations. Thus, these are fair assumptions to make in our proofs. 

For (vi) or (vii), it is possible that $|\gamma_d| = |\gamma_r|$ or $|\gamma_d| = \gamma_s$ or $|\gamma_r| = \gamma_s$ or $\gamma_r = 0$ on some vertex of $\sigma$. In that case, we track it in our topological invariant. It is not difficult to modify our existing proofs to accommodate these cases.

(viii) is more difficult to handle. If assumption (viii) does not hold, and there exists some $z \in \sigma$ satisfying $\gamma_d(z) = \gamma_s(z) = 0$, but $|\gamma_r(z)| > 0$, then the point $z$ will lie within $\RP$ or $\RN$ and will not lie on the boundary between any different regions. Thus, we can apply a small perturbation to $\gamma_d$ or $\gamma_s$ to ensure that $\gamma_d(z) \not = 0$ or $\gamma_s(z) \not = 0$ (ensuring that (viii) holds) without altering the topology. It is similar if $\gamma_r(z) = \gamma_s(z) = 0 < |\gamma_d(z)|$.

Thus, we are only left to handle the case where there is a point $z$ where $\gamma_d(z) = \gamma_r(z) = \gamma_s(z) = 0$. To handle the topology, we employ the following lemmas:

\begin{lemma}
Suppose that $z \in \sigma$ is a point where $\gamma_d(z) = \gamma_r(z) = \gamma_s(z) = 0$. Let $p$ lie on the boundary of $\sigma$. Suppose that $p \in \DP$. Let $q$ lie on the segment between $p$ and $z$. Then $q \in \DP$.

Similar lemmas are true for $\DN$, $\RP$, $\RN$, $\RRN$, $\RRP$ and $\RRN$.
\end{lemma}

\begin{proof}
This follows from the fact that $\gamma_d$ and $\gamma_r$ are affine, and $\gamma_s$ is convex.
\end{proof}

\begin{lemma}
Suppose that $z \in \sigma$ is a point where $\gamma_d(z) = \gamma_r(z) = \gamma_s(z) = 0$. Let $p$ lie on the boundary of $\sigma$. Suppose that $p \in \SRP$. Let $q$ lie on the segment between $p$ and $z$. Then $q \in \SRP$.

Similar lemmas are true for $\SRN$ and $\SA$.
\end{lemma}

\begin{proof}
Let $l$ be the line connecting $p$ to $z$. Let $\phi$ parametrize $l$ with $\phi(0) = z$ and $\phi(1) = p$. Notice that $(\gamma_r \circ \phi)^2$ and $(\gamma_s \circ \phi)^2$ are both quadratic functions. Both of these functions will have a vertex at $t = 0$, where they intersect. Since $(\gamma_r \circ \phi)(1)^2 < (\gamma_s \circ \phi)(1)^2$, it follows that $(\gamma_r \circ \phi)(t)^2 < (\gamma_s \circ \phi)(t)^2$ for every $t > 0$. Thus, $\gamma_s(q)^2 > \gamma_r(q)^2$.

Additionally, because $\gamma_r$ is affine, $\gamma_r(p) > 0$, and $\gamma_r(z) = 0$, it follows that $\gamma_r(q) > 0$. Thus, $q \in \SRP$.
\end{proof}

Thus, for both the eigenvector and eigenvalue partitions, the toplogy of $\sigma$ is determined entirely based on the classifications of points along the boundary. By analyzing where the curves $\gamma_d = \gamma_s$ and $\gamma_r = \gamma_s$ have topologically significant intersections with the boundary of $\sigma$, as well as the classifications of each vertex, we are able to determine the topology around the boundary of $\sigma$. If there is a point $z \in \sigma$ where $\gamma_d(z) = \gamma_r(z) = \gamma_s(z) = 0$, we also note this in our invariant. Thus, when there is a point where $\gamma_d(z) = \gamma_r(z) = \gamma_s(z)$, our invariant is enough to recover the topology of $\sigma$.